\documentclass[journal,draftcls,onecolumn,11pt]{IEEEtran}
\usepackage{mathrsfs}
\usepackage{bbding}
\usepackage{amsfonts}
\usepackage{amsmath}
\usepackage{graphicx}
\usepackage{subfigure}
\usepackage{latexsym}
\usepackage{amssymb}
\usepackage{amsmath}
\usepackage{enumerate} 
\usepackage{color}
\usepackage{float}
\usepackage{amsthm}

\newtheorem{Assumption}{Assumption}
\newtheorem{Algorithm}{Algorithm}
\newtheorem{Definition}{Definition}
\newtheorem{Lemma}{Lemma}
\newtheorem{Problem}{Problem}
\newtheorem{Remark}{Remark}
\newtheorem{Theorem}{Theorem}
\newtheorem{Corollary}{Corollary}

\allowdisplaybreaks[4]

\begin{document}
\title{Distributed Stochastic Optimization for \\ Weakly Coupled Systems  with Applications to Wireless Communications}

\author
{\IEEEauthorblockN{Fan Zhang, \emph{StMIEEE}, Ying Cui, \emph{MIEEE}, Vincent K. N. Lau, \emph{FIEEE}, An Liu, \emph{MIEEE}} \thanks{Fan Zhang, Vincent K. N. Lau and An Liu are with Department of Electronic and Computer Engineering, Hong Kong University of Science and Technology, Hong Kong. Ying Cui is with Department of Electrical and Computer Engineering, Northeastern University, USA.}
}
\maketitle

\begin{abstract} 
In this paper, a framework is proposed to simplify solving the infinite horizon average cost problem for the weakly coupled multi-dimensional systems. Specifically, to address the computational complexity issue,   we first introduce a \emph{virtual continuous time system} (VCTS) and  obtain the associated \emph{fluid value function}.  The relationship between  the VCTS and the original discrete time system is further established.  To facilitate  the low complexity distributed implementation and  address the coupling challenge, we model the weakly coupled system as a perturbation of a \emph{decoupled base system} and  study the  decoupled base system. The fluid value function of the VCTS is approximated by the sum of the \emph{per-flow fluid value functions} and  the approximation error is established using perturbation analysis. Finally,  we obtain a low complexity distributed solution  based on the per-flow fluid value function approximation. We apply the framework to solve a delay-optimal control problem  for the $K$-pair interference networks and obtain a distributed power control algorithm.  The proposed algorithm is compared  with various baseline schemes through simulations and it is shown that significant delay performance gain can be achieved.
\end{abstract}



\section{Introduction}
Stochastic optimization plays a key role in solving various optimal control problems under stochastic evolutions and it has  a wide range of applications in multi-disciplinary areas such as control of complex computer networks \cite{CompApp}, radio resource optimization in wireless systems \cite{WirelessApp}  as well as financial engineering \cite{EcomApp}. A common approach to solve stochastic optimization problems is via Markov Decision Process (MDP) \cite{DP_Bertsekas}, \cite{Cao}.  In the MDP approach, the state process evolves stochastically as a controlled Markov chain.  The optimal control policy is obtained by solving the well-known \emph{Bellman equation}. However, brute-force value iteration or policy iteration \cite{NeuNet} cannot lead to viable solutions  due to the \emph{curse of dimensionality}. Specifically, the size of the state space and action space grows exponentially with the dimension of the state variables. The huge state space issue not only makes  MDP problems intractable from computational complexity  perspective but also has an exponentially large memory requirement for storing the value functions and policies.  Furthermore, despite the complexity and storage issues, brute-force solution of  MDP problems is also undesirable because it leads to centralized control, which induces huge signaling overhead in collecting the global system state information and broadcasting the overall control decisions at the controller. 

To address the above issues, \emph{approximate dynamic programming}\footnote{Approximate dynamic programming can be classified into two categories, namely \emph{explicit value function approximation} and \emph{implicit value function approximation} \cite{DP_Bertsekas}. The approximation technique discussed in this paper belongs  to the former category.} (ADP)  is proposed in \cite{ADP1, ADP2} to obtain approximate solutions to  MDP problems.  One approach in ADP is called \emph{state aggregation} \cite{StateAggr_1}, \cite{StateAggr_2},  where the state  space of the Markov chain is partitioned  into disjoint regions. The states belonging to a partition region share the same control action and the same value function.  Instead of finding solutions of the Bellman equation in the original huge state space, the controller  solves a simpler problem in the aggregated state space (or reduced state space), ignoring irrelevant state information.  While the size of the state space could be significantly reduced,  it still involves solving a system of Bellman equations in the reduced state space, which is hard to deal with for large systems. Another approach in ADP is called  \emph{basis function approximation} \cite{BasisFunc1}, \cite{BasisFunc2}, i.e., approximating the value function by a linear combination of preselected basis functions. In such  approximation architecture, the value function is computed by mapping it to a low dimensional  function space  spanned by the basis functions.  In order to reduce the approximation error, proper weight associated with each basis function is then calculated by solving the \emph{projected Bellman equation} \cite{DP_Bertsekas}. Some existing works \cite{ParaBasis1}, \cite{ParaBasis2} discussed the basis function adaptation problem where the basis functions are parameterized and their parameters are tuned by minimizing an objective function according to some evaluation criteria (e.g., Bellman error). State aggregation technique can be viewed as a special case of the basis function approximation technique \cite{ADP_Bertsekas}, where the basis function space is  determined accordingly when  a method of constructing an aggregated state space is given. Both approaches can be used to solve MDP problems for systems with general dynamic structures but they fail to exploit the potential problem structures and it is also non-trivial to apply these techniques to obtain distributed solutions. 

In this paper, we are interested in distributed stochastic optimization for multi-dimensional systems with weakly coupled dynamics in control variables. Specifically, there are $K$ control agents in the system with  $K$ sub-system  state variables $(\mathbf{x}_1, \dots, \mathbf{x}_K)$. The evolution of the $k$-th sub-system state $\mathbf{x}_k$ is weakly affected by the control actions of the $j$-th agent for all $j\neq k$. To solve the stochastic optimization problem, we first construct a \emph{virtual continuous time system} (VCTS) using the \emph{fluid limit approximation approach}. The Hamilton-Jacobi-Bellman (HJB) equation associated with the optimal control problem of the VCTS is closely related to the Bellman equation associated with the optimal control problem of the original discrete time system (ODTS). Note that although the VCTS approach is related to the fluid limit approximation approach as discussed in \cite{CompApp}, \cite{fluidiff}--\cite{closeformJ}, there is a subtle difference between them. For instance, the fluid limit approximation approach is based on the {\em functional law of large numbers}  \cite{FLLN}  for the  state dynamic evolution while the proposed VCTS approach  is based on problem transformation. In order to obtain a viable solution  for the multi-dimensional  systems with weakly coupled dynamics, there are several first order technical challenges that need to be addressed.
\begin{itemize}
	\item	\textbf{Challenges due to the Coupled State Dynamics and Distributed Solutions:}
	For multi-dimensional systems with coupled state dynamics, the HJB equation associated with the VCTS is a multi-dimensional partial differential equation (PDE), which is quite challenging in general. Furthermore, despite the complexity issue involved, the solution structure requires centralized implementation, which is undesirable from both  signaling and computational perspectives. There are some existing works using the fluid limit approximation approach to solve MDP problems \cite{fluidiff}--\cite{fluidRel} but they focus mostly on single dimensional problems \cite{fluidiff} or centralized solutions for multi-dimensional problems in large scale networks \cite{fluid2, fluidRel}. It is highly non-trivial to apply these existing fluid techniques to derive distributed solutions. 

	\item \textbf{Challenges on the Error-Bound between  the VCTS and the ODTS:}
	It is well-known that the fluid value function  is closely related to the relative value function of the discrete time system. For example, for single dimensional systems \cite{fluidiff} or heavily congested networks with centralized control \cite{fluid2, fluid3}, the fluid value function is a useful approximator for the  relative value function of the ODTS. Furthermore, the error bound between the fluid value function and the relative value function  is shown to be $O( \| \mathbf{x}\| ^2)$ for large  state  $\mathbf{x}$ \cite{CompApp, fluidRel}.  However, extending these results to the multi-dimensional systems with distributed control policies is highly non-trivial. 
	
	\item \textbf{Challenges due to  the Non-Linear Per-Stage Cost in Control Variables:}
	There are a number of existing works using fluid limit approximation to solve MDP problems \cite{fluidiff}--\cite{closeformJ}. However, most of the related works only considered linear cost function \cite{fluid2}--\cite{fluid3} and relied on the exact closed-form solution of the associated fluid value function \cite{fluidiff}, \cite{closeformJ}. Yet, in many applications, such as wireless communications, the data rate is a highly non-linear function of the transmit power and these existing  results based on the closed-form fluid limit approximation cannot be directly applied to the general case with non-linear per-stage cost in control variables. 
\end{itemize}

In this paper, we shall focus on distributed stochastic optimization for multi-dimensional systems with weakly coupled dynamics. We consider an infinite horizon average cost stochastic optimization problem with general non-linear per-stage cost function. We first introduce a VCTS to transform the original discrete time average cost optimization problem into a continuous time total cost problem. The motivation of solving the problem in the continuous time domain is to leverage the well-established mathematical tools from calculus and differential equations. By solving the \emph{HJB equation} associated with the total cost problem for the VCTS, we obtain the \emph{fluid value function}, which can be used to approximate the relative value function of the ODTS. We then establish the relationship between the fluid value function of the VCTS and the relative value function of the discrete time system. To address the low complexity distributed solution requirement and the coupling challenge, the weakly coupled system can be modeled as a \emph{perturbation} of a  decoupled base system.  The solution to the stochastic optimization problem for the weakly coupled system can be expressed as  solutions of $K$ distributed \emph{per-flow HJB Equations}. By solving the per-flow HJB equation, we obtain the \emph{per-flow fluid value function}, which can be used to generate localized control actions at each agent based on locally observable system states. Using perturbation analysis, we establish the gap between the fluid value function of the  VCTS and the sum of the per-flow fluid value functions.  Finally, we show that solving the Bellman equation associated with the original stochastic optimization problem using  per-flow fluid value function approximation is equivalent to solving a deterministic network utility maximization (NUM) problem \cite{NUM1}, \cite{NUM2}  and we propose a distributed algorithm for solving the associated NUM problem. We shall illustrate the above framework of distributed stochastic optimization using an application example in wireless communications. In the example, we consider the delay optimal control problem for $K$-pair interference networks, where the queue state evolution of each transmitter is weakly affected by the control actions of the other transmitters due to the cross-channel interference.  The delay performance of the proposed distributed solution is compared with various baseline schemes  through simulations and it is shown that substantial delay performance gain can be achieved.

\section{System Model and Problem Formulation}
In this section, we  elaborate on the weakly coupled multi-dimensional systems. We  then formulate the associated  infinite horizon average cost stochastic optimization problem and discuss the general solution. We  illustrate the application of the framework using an application example in wireless communications.

\begin{figure}
  \centering
  \includegraphics[width=6.4in]{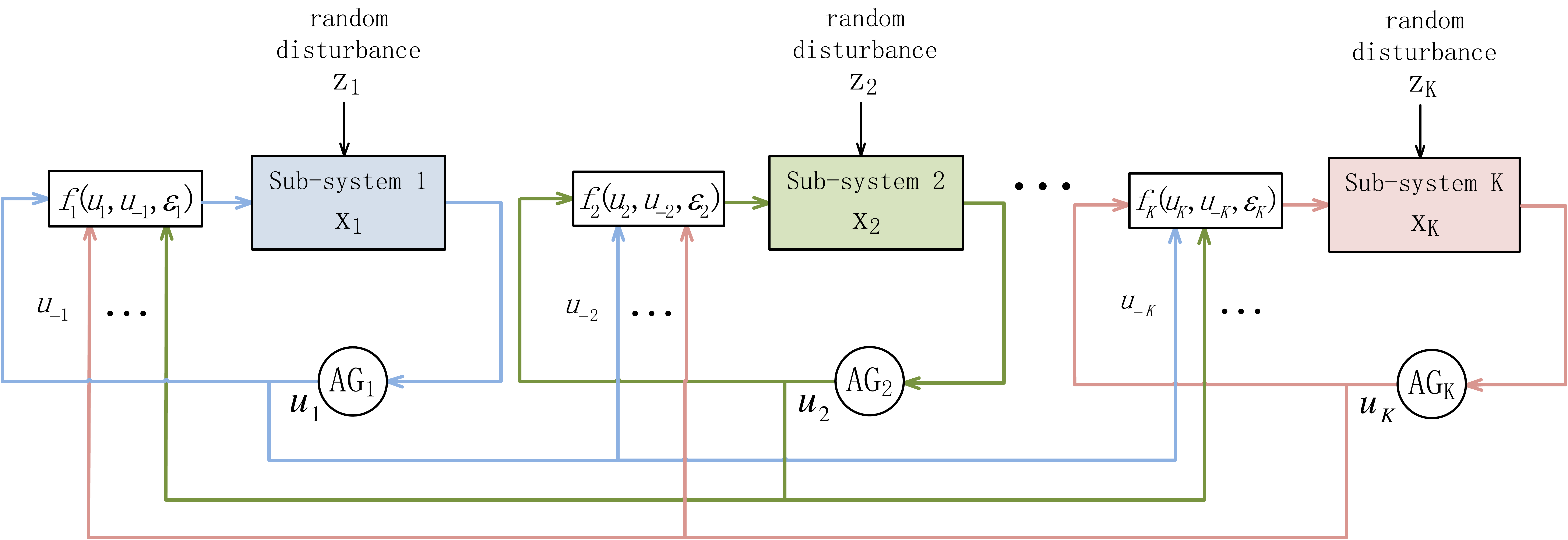}
  \caption{System model of a $K$-dimensional  multi-agent control system. $\text{AG}_k$ is the control agent associated with the $k$-th sub-system.}
  \label{Kagent}
\end{figure}

\subsection{Multi-Agent Control System with Weakly Coupled Dynamics}		\label{multidim_sys}
We consider a  multi-dimensional system with weakly coupled dynamics as shown in Fig.~\ref{Kagent}. The  system consists of $K$ control agents indexed by $k \in \mathcal{K}$, where $\mathcal{K}=\{1,\dots,K\}$.  The time dimension is discretized  into \emph{decision epochs} indexed by $t$ with  epoch duration $\Delta$. The weakly coupled multi-agent control system is a time-homogeneous MDP model which can be characterized by four elements, namely the \emph{state space}, the \emph{action space}, the \emph{state transition kernel} and the \emph{system cost}. We elaborate on the multi-agent control system as follows.

The global system state $\mathbf{x}\left(t\right) = \left(\mathbf{x}_1\left(t\right),\dots, \mathbf{x}_K\left(t\right)\right) \in  \boldsymbol{\mathcal{X}} \triangleq \mathcal{X}_1 \times \dots \times \mathcal{X}_K$ at the $t$-th epoch is partitioned into $K$ sub-system states, where $\mathbf{x}_k\left(t\right)\in \mathcal{X}_k $ denote the $k$-th sub-system state at the $t$-th epoch. $\boldsymbol{\mathcal{X}} $ and  $\mathcal{X}_k$ are the global system state space and sub-system state space, respectively.  The $k$-th control agent generates a set of control actions $\mathbf{u}_k\left(t\right) \in \mathcal{U}_k$ at the $t$-th epoch, where $\mathcal{U}_k$ is a compact action space for the $k$-th sub-system. Furthermore, we denote $\mathbf{u}\left(t\right)= \left(\mathbf{u}_1\left(t\right), \dots, \mathbf{u}_K\left(t\right) \right) \in \boldsymbol{\mathcal{U}} \triangleq \mathcal{U}_1 \times \dots \times \mathcal{U}_K$ as the global control action, where $ \boldsymbol{\mathcal{U}}$ is the global  action space. Given an initial global system state $\mathbf{x}(0) \in \boldsymbol{\mathcal{X}}$, the $k$-th sub-system state $\mathbf{x}_k\left(t\right)$ evolves according to the following dynamics:
\begin{equation}	\label{DiscreteDym}
	\mathbf{x}_k(t+1) = \mathbf{x}_k\left(t\right) +  \mathbf{f}_k\left(\mathbf{u}_k\left(t\right), \mathbf{u}_{-k}\left(t\right),\boldsymbol{\epsilon}_k\right) \Delta + \mathbf{z}_k\left(t\right)\Delta, \quad \mathbf{x}_k(0) \in \mathcal{X}_k
\end{equation}
where $ \mathbf{u}_{-k}\left(t\right) = \{\mathbf{u}_j\left(t\right): \forall j \neq k\}$ and $\boldsymbol{\epsilon}_k = \{\epsilon_{kj}: \forall j \neq k\}$. $\epsilon_{kj}$ ($ k \neq j$) is the \emph{coupling parameter} that measures the extent to which the control actions of the $j$-th agent affect the  state evolution of the $k$-th sub-system. In this paper, we assume that $\mathbf{f}_k\left(\mathbf{u}_k\left(t\right), \mathbf{u}_{-k}\left(t\right),\boldsymbol{\epsilon}_k\right)$ has the form $\mathbf{f}_k\left(\mathbf{u}_k\left(t\right), \epsilon_{k,1}\mathbf{u}_{1}\left(t\right),\cdots, \epsilon_{k,k-1}\mathbf{u}_{k-1}\left(t\right), \right.\\ \left. \epsilon_{k,k+1}\mathbf{u}_{k+1}\left(t\right),\cdots,\epsilon_{k,K}\mathbf{u}_{K}\left(t\right)\right)$, and $\mathbf{f}_k\left(\mathbf{u}_k, \mathbf{u}_{-k},\boldsymbol{\epsilon}_k\right)$ is  continuously differentiable  w.r.t. each element in  $\mathbf{u}_k$, $\mathbf{u}_{-k}$ and $\boldsymbol{\epsilon}_k$.  $\mathbf{z}_k  \in \boldsymbol{\mathcal{Z}}$  is a random disturbance process for the $k$-th sub-system, where $\boldsymbol{\mathcal{Z}}$ is the  disturbance space.  We have the following assumption on the disturbance process. 
\begin{Assumption} [Disturbance Process Model]	\label{assum_dis}
	The disturbance process $\mathbf{z}_{k}\left(t\right)$  is i.i.d. over decision epochs according to a general distribution  $\Pr[\mathbf{z}_{k}]$ with finite mean $\mathbb{E}[\mathbf{z}_{k}]=\overline{\mathbf{z}}_{k}$ for all $k \in \mathcal{K}$.	Furthermore, the disturbance processes $\{ \mathbf{z}_{k}\left(t\right) \}$ are independent w.r.t. $k$.~\hfill\IEEEQED
\end{Assumption}

Given the global system state $\mathbf{x}\left(t\right)$ and the global control action $\mathbf{u}\left(t\right)$ at the $t$-th epoch, the system evolves as a \emph{controlled Markov chain} according to the following  transition kernel:
\begin{align}
	 &\Pr[ \mathbf{x}(t+1) \in  A|\mathbf{x}\left(t\right), \mathbf{u}\left(t\right)] = \prod_k  \underbrace{\Pr[\mathbf{x}_k(t+1) \in A_k|\mathbf{x}_k\left(t\right),\mathbf{u}\left(t\right) ]}_{\text{transition kernel of the $k$-th sub-system}}   \notag \\
	  =&\prod_k  \Pr[ \mathbf{x}_k\left(t\right) +  \mathbf{f}_k\left(\mathbf{u}_k\left(t\right), \mathbf{u}_{-k}\left(t\right),\boldsymbol{\epsilon}_k\right)\Delta + \mathbf{z}_k\left(t\right)\Delta \in A_k ] \label{transition kernel}
\end{align}
where $A = \otimes_{k=1}^K A_k$ is a  measurable set\footnote{For discrete sub-system state space $\mathcal{X}_k$, $A_k$ is a singleton with exactly one element that belongs to $\mathcal{X}_k$. For continuous  sub-system state space $\mathcal{X}_k$, $A_k$ is a measurable subset of $\mathcal{X}_k$.}  with $A_k \in \mathcal{X}_k$  for all $k \in \mathcal{K}$. Note that the transition kernel in (\ref{transition kernel}) is time-homogeneous.

We consider a  system cost function $c:\boldsymbol{\mathcal{X}} \times \boldsymbol{\mathcal{U}} \rightarrow [0, \infty]$ given by
\begin{equation}	\label{syscost}
	c(\mathbf{x}, \mathbf{u}) = \sum_{k=1}^K  c_k(\mathbf{x}_k, \mathbf{u}_k)
\end{equation}
where  $c_k: \mathcal{X}_k \times \mathcal{U}_n \rightarrow [0, \infty)$ is the cost function of the $k$-th sub-system given by
\begin{equation}	\label{costfunc}
	c_k(\mathbf{x}_{k}, \mathbf{u}_{k}) = \alpha_k \| \mathbf{x}_k \| _{\mathbf{v}_k,1} + g_k\left(\mathbf{u}_k \right)  
\end{equation}
where $\alpha_k$ is  some  positive constant, $ \| \cdot \| _{\mathbf{v}_k,1} $ is a weighted $\mathcal{L}_1$ norm with $\mathbf{v}_k$ being the corresponding weight vector and $g_k$ is a continuously differentiable  function w.r.t. each element in vector $\mathbf{u}_k$.  In addition, we assume $\left| g_k \left(\mathbf{u}_k \right) \right| < \infty$ for all $\mathbf{u}_k \in \mathcal{U}_k$.
\begin{Definition} [Decoupled and Weakly Coupled  Systems]		\label{WCandDC}
	A multi-dimensional system in (\ref{DiscreteDym}) is called \emph{decoupled} if  all the coupling parameters  are equal to zero. A multi-dimensional system in (\ref{DiscreteDym}) is called \emph{weakly coupled} if  all the coupling parameters $\{  \epsilon_{kj}: \forall k, j \in \mathcal{K},  k \neq j \}$  are very small.  For weakly coupled systems, the state dynamics of each sub-system is weakly affected by the control actions of the other sub-systems as shown in the state evolution equation in  (\ref{DiscreteDym}).  In addition, define $\epsilon = \max \left\{\left|\epsilon_{kj}\right|: \forall k, j \in \mathcal{K},  k \neq j  \right\}$. Then, we have   $|\epsilon_{kj}| \leq \epsilon$ for all $k, j \in \mathcal{K}, k \neq j$.~\hfill\IEEEQED
\end{Definition}

Note that the above framework considers the coupling due to the control actions only. Yet, it has already covered  many  examples in wireless communications such as the interference networks with weak interfering cross links and the multi-user MIMO broadcast channels with perturbated channel state information. We shall illustrate the application of the framework using interference networks as an example  in Section \ref{ExampleSec}. 

\subsection{Control Policy and Problem Formulation}
In this subsection, we  introduce the  stationary centralized   control policy and formulate the infinite horizon average cost  stochastic optimization problem.
\begin{Definition} [Stationary Centralized  Control Policy]		\label{StationDef}
A  stationary centralized  control policy of the $k$-th sub-system $\Omega_k: \boldsymbol{\mathcal{X}} \rightarrow \mathcal{U}_k$ is defined as a mapping from the global system state space $\boldsymbol{\mathcal{X}}$ to the  action space of the $k$-th sub-system $\mathcal{U}_k$. Given a global system state realization $\mathbf{x} \in \boldsymbol{\mathcal{X}}$, the control action of the $k$-th sub-system is given by $\Omega_k\left( \mathbf{x} \right) = \mathbf{u}_k  \in \mathcal{U}_k$.  Furthermore, let $\Omega= \{ \Omega_k:\forall  k \}$ denote the aggregation of the control policies for all the $K$ sub-systems.~\hfill\IEEEQED
\end{Definition}

\begin{Assumption}	[Admissible Control Policy]
	A policy $\Omega$ is assumed to be admissible if the following requirements are satisfied:
	\begin{itemize}
		\item it is unichain, i.e., the controlled Markov chain $\left\{\mathbf{x}\left(t \right)\right\}$ has a single recurrent class (and possibly some transient states) \cite{DP_Bertsekas}.
		\item it is $n$-th order stable, i.e., $\lim_{T \rightarrow \infty} \mathbb{E}^{\Omega} \left[\left|\mathbf{Q} \right|^n \right] < \infty$.
		\item $\mathbf{u}=\Omega(\mathbf{x})$ satisfies some constraints\footnote{We shall illustrate the specific constraints for the feasible control policy of the application example in Section \ref{ExampleSec}.} depending on the different application scenarios.~\hfill\IEEEQED
	\end{itemize}
\end{Assumption}

Given an admissible control policy $\Omega$, the average cost of the system starting from a given initial global system state $\mathbf{x}\left(0 \right)$ is given by
\begin{equation}		\label{Utility}
	\overline{L}^{\Omega}\left(\mathbf{x}\left(0 \right) \right) = \limsup_{T \rightarrow \infty} \frac{1}{T} \sum_{t=0}^{T-1} \mathbb{E}^{\Omega} [ c(\mathbf{x}\left(t\right), \Omega(\mathbf{x}\left(t\right))) ]
\end{equation}
where $\mathbb{E}^{\Omega}$ means taking expectation  w.r.t. the probability measure induced by the control policy $\Omega$. 

We consider an infinite horizon average cost problem. The objective is to find an optimal policy such that the average cost in (\ref{Utility}) is minimized\footnote{Substituting the expression of $c(\mathbf{x},\mathbf{u})$ in (\ref{syscost}) into  (\ref{Utility}), the average cost of the system can be written as $\overline{L}^{\Omega} \left(\mathbf{x}\left(0 \right) \right) =\sum_{k=1}^K \overline{L}_k^{\Omega}\left(\mathbf{x}\left(0 \right) \right) $, where $\overline{L}_k^{\Omega}\left(\mathbf{x}\left(0 \right) \right) =\limsup_{T \rightarrow \infty} \frac{1}{T} \sum_{t=1}^{T} \mathbb{E}^{\Omega} [c_k(\mathbf{x}_k,\mathbf{u}_k))]$  is the average cost of the $k$-th sub-system. Therefore, minimizing (\ref{Utility}) is equivalent to minimizing the  sum average cost of each sub-system.}. Specifically, we have:
\begin{Problem}	[Infinite Horizon Average Cost  Problem]  \label{IHAC_MDP}
	The infinite horizon average cost problem for  the multi-agent control system is formulated as follows:
 	\begin{equation}		\label{org_MDPp}
		\min_{\Omega} \overline{L}^{\Omega}\left(\mathbf{x}\left(0 \right) \right)
	\end{equation}
where $\overline{L}^{\Omega}\left(\mathbf{x}\left(0 \right) \right)$ is given in (\ref{Utility}).		~\hfill\IEEEQED
\end{Problem}		

Under the assumption of the admissible control policy, the optimal control policy of Problem \ref{IHAC_MDP} is independent of the initial state $\mathbf{x}\left(0 \right) $  and can be obtained by solving the \emph{Bellman equation} \cite{DP_Bertsekas}, which is summarized in the following lemma:
\begin{Lemma} [Sufficient Conditions for Optimality under ODTS]	\label{LemBel}
	If there exists a ($\theta^\ast, \{ V\left(\mathbf{x} \right) \}$) that satisfies the following Bellman equation:
	\begin{equation}	\label{OrgBel}
		\theta^\ast \Delta + V\left(\mathbf{x} \right) = \min_{ \Omega\left( \mathbf{x} \right)} \left[ c\left(\mathbf{x}, \Omega\left(\mathbf{x}\right)\right) \Delta + \left(\mathrm{\Gamma}_{\Omega}V\right)\left(\mathbf{x}\right)  \right], \quad \forall \mathbf{x} \in \boldsymbol{\mathcal{X}}
	\end{equation}
where the operator $\Gamma_{\Omega}$ on $V\left(\mathbf{x} \right)$ is defined as\footnote{Because the transition kernel in (\ref{transition kernel}) is time-homogeneous, $ (\Gamma_{\Omega}V)\left( \mathbf{x} \right)$ is independent of $t$.} $(\Gamma_{\Omega}V)\left( \mathbf{x} \right)  = \mathbb{E}\left[ V\left(\mathbf{x}\left(t+1\right)\right) \big| \mathbf{x}\left(t\right)=\mathbf{x}, \Omega\left(\mathbf{x}\right) \right]$. Suppose for all admissible control policy $\Omega$ and initial global system state $\mathbf{x}\left(0 \right)$, the following transversality condition is satisfied:
\begin{align}	\label{transodts}
	\lim_{T \rightarrow \infty} \frac{1}{T}\big[V\left(\mathbf{x}\left(0 \right) \right) - \mathbb{E}\left[ V\left(\mathbf{x}\left(T \right) \right) |\mathbf{x}\left(0 \right), \left\{\Omega(\mathbf{x}\left(t \right)): 0 \leq t \leq T \right\} \right]\big]=0
\end{align}
Then $\theta^\ast=\overline{L}^{\ast} = \min_{\Omega}\overline{L}(\Omega) $ is the optimal average cost. If $\Omega^{\ast}\left( \mathbf{x} \right)$ attains the minimum of the R.H.S. of (\ref{OrgBel}) for all $\mathbf{x} \in \boldsymbol{\mathcal{X}}$, $\Omega^{\ast}$ is the optimal control policy. $V\left(\mathbf{x} \right)$  is called the \emph{relative value function}.~\hfill\IEEEQED
\end{Lemma}
\begin{proof}
	Please refer to \cite{DP_Bertsekas} for details.
\end{proof}

Based on Lemma \ref{LemBel}, we establish the following corollary on  the approximate optimal solution.
\begin{Corollary}	[Approximate Optimal Solution]	\label{cor1}
	If there exists a ($\tilde{\theta}^\ast, \{ \widetilde{V}\left(\mathbf{x} \right) \}$) that satisfies the following \emph{approximate} Bellman equation:
	\begin{equation}		\label{simbelman}
		\tilde{\theta}^\ast = \min_{ \Omega\left( \mathbf{x} \right)} \left[ c\left(\mathbf{x}, \Omega\left(\mathbf{x}\right)\right) +  \nabla_{\mathbf{x}}V \left(\mathbf{x} \right) \left[  \mathbf{f}\left( \Omega\left( \mathbf{x}\right), \boldsymbol{\epsilon}\right)+\overline{\mathbf{z}}  \right]^T  \right], \quad \forall \mathbf{x} \in \boldsymbol{\mathcal{X}}
	\end{equation}
	where\footnote{$\nabla_{\mathbf{x}_k} \widetilde{V}\left(\mathbf{x}\right)$  is a row  vector with each element being the first order partial derivative of $\widetilde{V}\left(\mathbf{x}\right)$ w.r.t. each component in vector $\mathbf{x}_k$.} $\nabla_{\mathbf{x}} \widetilde{V}\left(\mathbf{x}\right) \triangleq \left( \nabla_{\mathbf{x}_1} \widetilde{V}\left(\mathbf{x}\right), \dots, \nabla_{\mathbf{x}_K} \widetilde{V}\left(\mathbf{x}\right) \right)$,    $\mathbf{f}\left(\mathbf{u}, \boldsymbol{\epsilon}\right) \triangleq \left(\mathbf{f}_1\left(\mathbf{u}_1, \mathbf{u}_{-1},\boldsymbol{\epsilon}_1 \right),\dots,\mathbf{f}_K\left(\mathbf{u}_K, \mathbf{u}_{-K},\boldsymbol{\epsilon}_K \right)  \right)$, $\overline{\mathbf{z}}=\left(\overline{\mathbf{z}}_1, \dots, \overline{\mathbf{z}}_K \right)$ and $\boldsymbol{\epsilon} \triangleq \left\{\boldsymbol{\epsilon}_k:\forall k \right\}$. Suppose for all admissible control policy $\Omega$ and initial global system state $\mathbf{x}\left(0 \right)$, the transversality condition  in (\ref{transodts}) is satisfied. Then, 
	\begin{align}
		\theta^\ast&=\tilde{\theta}^\ast+\mathcal{O}\left(\Delta \right)	\\
		V\left(\mathbf{x} \right)&=\widetilde{V}\left(\mathbf{x} \right)+\mathcal{O}\left(\Delta \right)
	\end{align}~\hfill\IEEEQED
\end{Corollary}
\begin{proof}  
Please refer to Appendix	 A.	
\end{proof}

Deriving the optimal control policy from (\ref{OrgBel}) (or from (\ref{simbelman})) requires the knowledge of the relative value function $\{ V\left(\mathbf{x} \right)\}$. However, obtaining the relative value function is not trivial as it involves solving  a large  system of nonlinear fixed point equations.  Brute-force approaches  such as value iteration or policy iteration \cite{DP_Bertsekas} require huge complexity and cannot lead to any implementable solutions. Furthermore, deriving the optimal control policy $\Omega^{\ast}$ requires knowledge of the \emph{global system state}, which is undesirable from the signaling loading perspective.  We shall obtain low complexity  distributed solutions using the virtual continuous time system (VCTS) approach   in  Section \ref{Fluid_Approach}.

\begin{figure}
  \centering
  \includegraphics[width=4in]{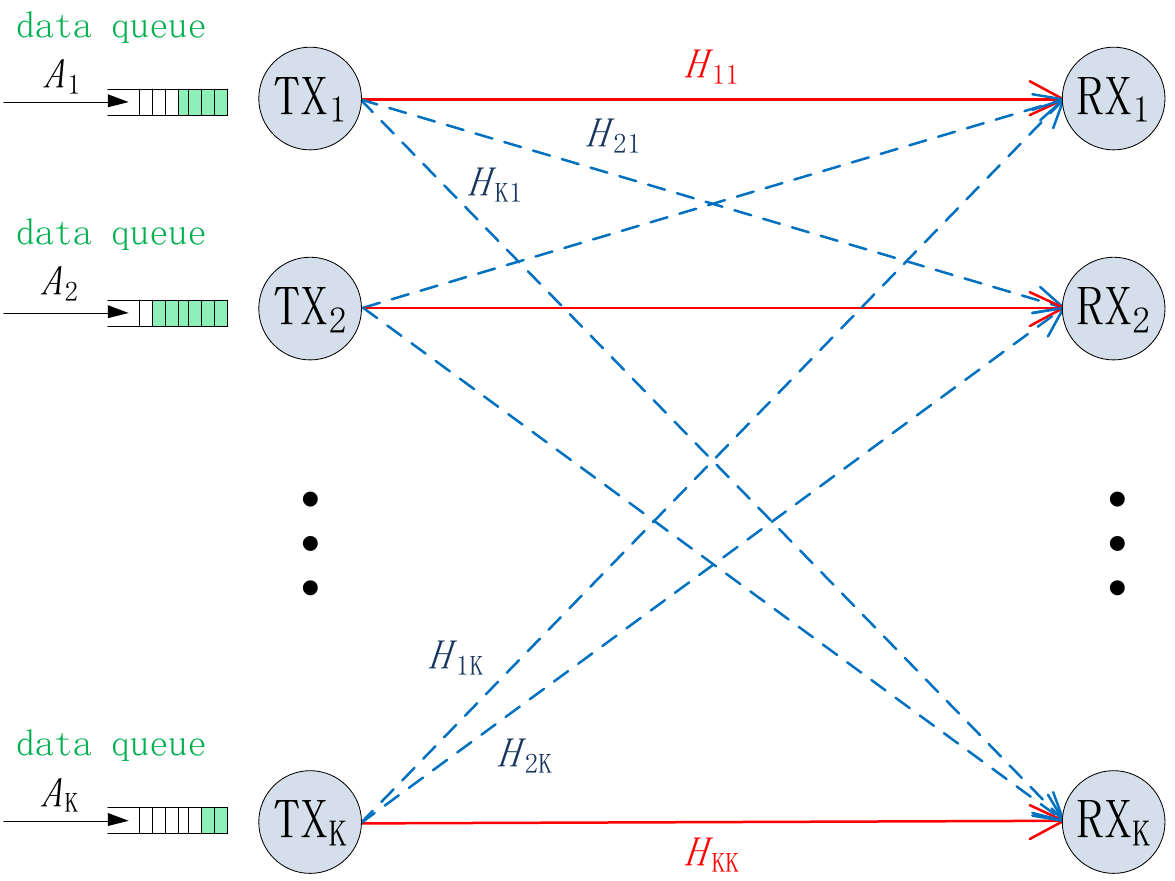}
  \caption{System model of  a K-pair interference network. Each transmitter maintains a data queue  for the bursty traffic flow towards the desired receiver in the network.}
  \label{K-pair}
\end{figure}

\subsection{Application Example -- K--Pair Interference Networks} \label{ExampleSec}

In this subsection, we  illustrate the application of the above multi-agent control framework using interference networks as an example.  We consider a $K$-pair interference network as illustrated in Fig.~\ref{K-pair}. The $k$-th transmitter sends information  to the $k$-th receiver and all the $K$ transmitter-receiver (Tx-Rx) pairs share the same spectrum (and hence, they potentially interfere with each other).  The received signal at the $k$-th receiver is given by
\begin{equation}
	r_k = \sqrt{L_{kk}} H_{kk} s_k + \sum_{j \neq k} \sqrt{L_{kj}} H_{kj} s_j + z_k
\end{equation}
where $L_{kj}$ and $H_{kj}$ are the long term path gain and microscopic  fading coefficient  from the $j$-th transmitter to the $k$-th receiver, respectively. $H_{kj}$ follows a complex Gaussian distribution with unit variance, i.e., $H_{kj} \sim \mathcal{CN}(0,1)$. $s_k$ is the symbol sent by the $k$-th transmitter, and $z_k \sim \mathcal{CN}(0,1)$ is i.i.d.  complex Gaussian channel noise. Each Tx-Rx pair in Fig.~\ref{K-pair} corresponds to one sub-system according to the general model in Section \ref{multidim_sys}. The $k$-th transmitter in this example corresponds to  the $k$-th control agent in the general model.  Denote the global CSI (microscopic  fading coefficient) as $\mathbf{H} =\{H_{kj}: \forall k,j \in \mathcal{K} \}$.  The time dimension is partitioned into decision epochs indexed by $t$ with  duration $\tau$. We have the following assumption on the channel model.

\begin{figure}
  \centering
  \includegraphics[width=5.5in]{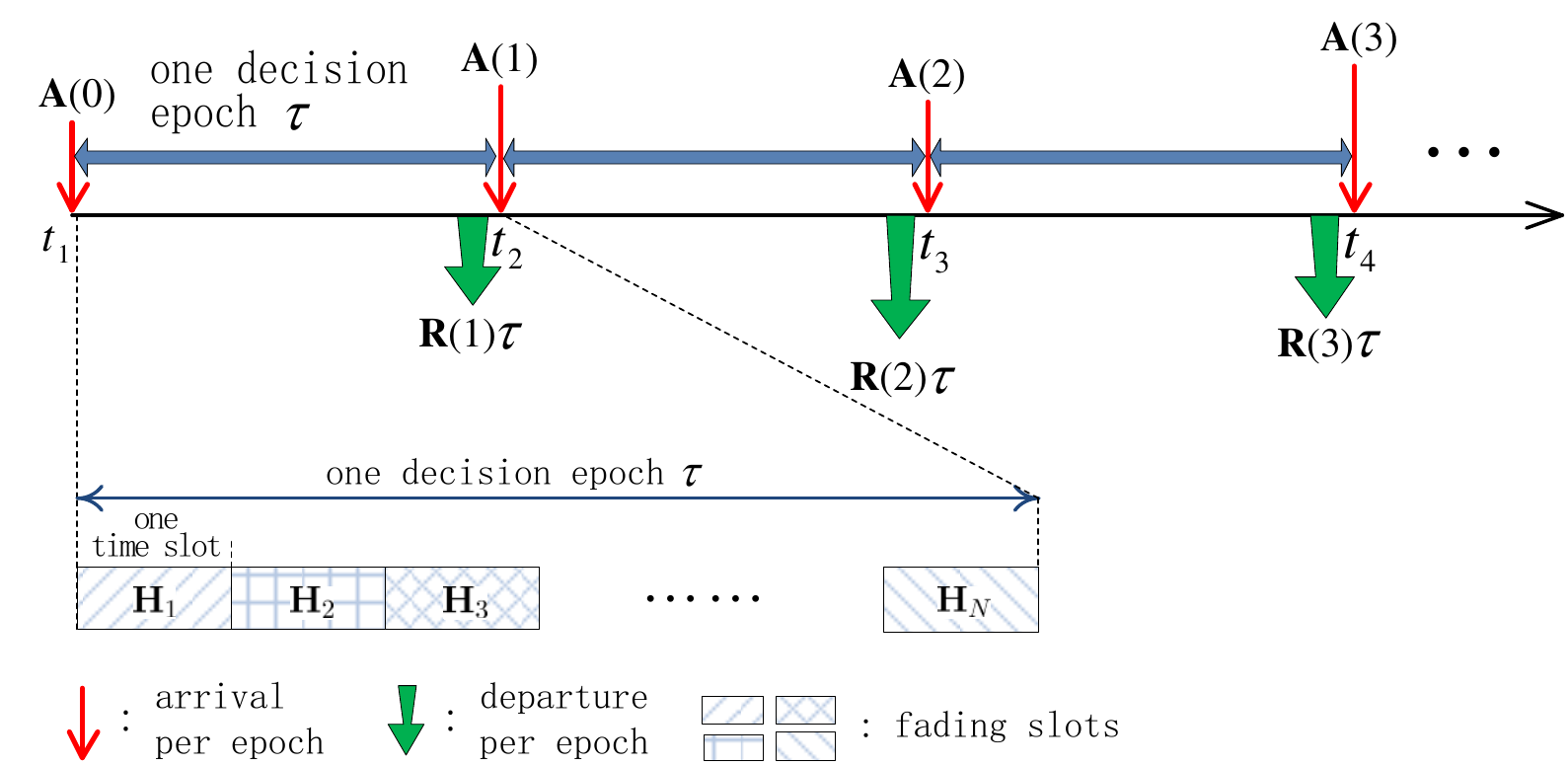}
  \caption{Illustration of the decision epochs, time slots and the associated arrival and departure events in the application example.  $\mathbf{A}\left( t\right)$ is the random packet arrival at the beginning of the $(t+1)$-th decision epoch. The decision epoch is divided into equal-sized time slots with i.i.d. microscopic fading coefficient $\mathbf{H}_1, \dots, \mathbf{H}_N$. $\mathbf{R}\left(t\right)=\left(R_1\left(t\right),\dots,R_K\left(t\right) \right)$ is the controlled global departure process at the end of the $t$-th epoch,  where $R_k\left( t\right)$ is the average rate (averaged over $\mathbf{H}$) of the $k$-th Tx-Rx pair given in (\ref{rate1}).}
  \label{blockfading}
\end{figure}

\begin{Assumption}	[Channel Model]	\label{CSIassum}	
We assume fast fading on the microscopic fading coefficient  $\mathbf{H}$. The decision epoch is divided into equal-sized time slots as shown in Fig.~\ref{blockfading}. The slot duration is sufficiently small compared with $\tau$. $H_{kj}\left(t \right)$ remains constant within each slot and is i.i.d. over slots\footnote{The assumption on the microscopic fading coefficient could be justified in many applications. For example,  in frequency hopping  systems, the channel fading remains constant within one slot (hop) and is i.i.d. over slots (hops) when the frequency is hopped from one channel to another.}. Furthermore, $\left\{H_{kj}\left(t \right)\right\}$ is independent w.r.t. $\left\{k,j \right\}$. The path gain  $\{L_{kj}: \forall k, j \in \mathcal{K}\}$ remains constant for the duration of the communication session.	~\hfill\IEEEQED
\end{Assumption}

There is a bursty data source at each transmitter and let  $\mathbf{A}\left(t\right)=\left(A_1\left(t\right),\dots,A_K\left(t\right) \right)$ be the random new arrivals (number of packets per second) for the $K$ transmitters at the beginning of the $(t+1)$-th  epoch. We have the following assumption on the arrival process.
\begin{Assumption} [Bursty Source Model]	\label{assumeA}
	The arrival process $A_k\left(t\right)$ is i.i.d. over decision epochs according to a general distribution $\Pr(A_k)$ with finite average arrival rate $\mathbb{E}[A_k]=\lambda_k$ for all $k \in \mathcal{K}$. Furthermore, the arrival process $\{ A_k\left(t\right) \}$ is independent w.r.t. $k$.~\hfill\IEEEQED
\end{Assumption}

Let $\mathbf{Q}\left(t\right)=\left( Q_1\left(t\right), \dots, Q_K\left(t\right) \right)\in\boldsymbol{\mathcal{Q}} \triangleq \mathcal{Q}^K$ denote the global queue state information (QSI) at the beginning of the $t$-th  epoch, where $Q_k\left(t\right) \in \mathcal{Q}$ is the local QSI denoting the number of packets (pkts) at the data queue for the $k$-th transmitter and $\mathcal{Q}$ is the local QSI state space of each transmitter. $Q_k\left(t\right)$  corresponds to the state of the $k$-th sub-system in the general model in Section \ref{multidim_sys}.  

Treating interference as noise, the ergodic capacity (pkts/second) of the $k$-th Tx-Rx pair at the $t$-th epoch  is given by
\begin{equation}		\label{rate1}
	R_k\left(\mathbf{p}\left(t\right)\right) = \mathbb{E} \left[ \log\left( 1+\frac{p_k^{\mathbf{H}}\left(t\right)L_{kk}|H_{kk}|^2}{1 + \sum_{j \neq k} p_j^{\mathbf{H}}\left(t\right) L_{kj}|H_{kj}|^2}\right) \bigg| \mathbf{Q}\left(t\right)\right]
\end{equation}
where $\mathbb{E}[ \cdot | \mathbf{Q}\left(t\right)]$ denotes the conditional expectation given $\mathbf{Q}\left(t\right)$, and $p_k^{\mathbf{H}}\left(t\right)$ is the transmit power of the $k$-th transmitter at the $t$-th epoch when the global CSI realization is  $\mathbf{H}$. Since the transmitter cannot transmit more than $Q_k\left(t\right)$ at any decision epoch $t$, we require $ R_k\left(t\right)\tau \leq Q_k\left(t\right)$, i.e.,
\begin{align}	\label{constraintO}
	\mathbb{E} \left[ \log\left( 1+\frac{p_k^{\mathbf{H}}\left(t\right)L_{kk}|H_{kk}|^2}{1 + \sum_{j \neq k} p_j^{\mathbf{H}}\left(t\right) L_{kj}|H_{kj}|^2}\right) \bigg| \mathbf{Q}\left(t\right)\right]\tau \leq Q_k\left(t\right)
\end{align}
Hence, the queue dynamics for the $k$-th transmitter is given by\footnote{We assume that the transmitter of each Tx-Rx pair  is  causal so that new arrivals are observed after the transmitter's actions at each decision epoch.}
\begin{align}		
	Q_k(t+1) &= Q_k\left(t\right) -   R_k\left(t\right) \tau  + A_k\left(t\right)\tau	     	\label{Qdyn}
\end{align}	
The time epochs for the arrival and departure events  of the global system are  illustrated in Fig.~\ref{blockfading}.  

Comparing with the state dynamics of the general system framework in (\ref{DiscreteDym}), we have $\Delta=\tau$, $ \mathbf{f}_k\left(\mathbf{u}_k\left(t\right), \mathbf{u}_{-k}\left(t\right),{\boldsymbol{\epsilon}}_k\right)  = -   R_k\left(t\right) \tau = - \mathbb{E} \left[ \log\left( 1+\frac{p_k^{\mathbf{H}}\left(t\right)L_{kk}|H_{kk}|^2}{1 + \sum_{j \neq k} p_j^{\mathbf{H}}\left(t\right) L_{kj}|H_{kj}|^2}\right) \bigg| \mathbf{Q}\left(t\right) \right] \tau $, $\epsilon_{kj}=L_{kj}$ and $\mathbf{z}_k\left(t\right)=A_k\left(t\right)$ in the example. The control action generated by the $k$-th control agent ($k$-th transmitter) is $\mathbf{u}_k\left(t\right) = \mathbf{p}_k\left(t\right) = \{p_k^{\mathbf{H}}\left(t\right): \forall \mathbf{H}\}$. Furthermore, we have $\mathbf{u}_{-k}\left(t\right) =   \boldsymbol{\mathbf{p}}_{-k}\left(t\right) = \{p_j^{\mathbf{H}}\left(t\right): \forall \mathbf{H}, j \neq k\} $, and $\boldsymbol{\epsilon}_k=\mathbf{L}_k \triangleq \left\{L_{kj}: \forall j \neq k \right\}$.  Define $L = \max \{L_{kj}: \forall k, j\in \mathcal{K}, k \neq j  \}$, then we have $L_{kj} \leq L$ for all $k,j \in \mathcal{K}, k \neq j$. This corresponds to a physical interference network with weak interfering cross links due to the small cross-channel path gain. For simplicity, let $\boldsymbol{\mathbf{p}}\left(t\right)=\{\mathbf{p}_k\left(t\right): \forall k \}$ be the collection of power control actions of all the $K$ transmitters.  We define the system cost function as
\begin{equation}  \label{eg1:cost}
	c\left(\mathbf{Q}, \boldsymbol{\mathbf{p}}\right) = \sum_{k=1}^K c_k(Q_k, \mathbf{p}_k)
\end{equation}
where $c_k(Q_k, \mathbf{p}_k)$ is the cost function of the $k$-th Tx-Rx pair which is given by
\begin{equation}	\label{eg: per_pair_cost}
	c_k(Q_k, \mathbf{p}_k) =  \beta_k \frac{Q_k}{\lambda_k} + \gamma_k  \mathbb{E} [ p_k^{\mathbf{H}}|\mathbf{Q} ]
\end{equation}
where $\beta_k > 0$ is a positive weight of the delay cost\footnote{The delay cost is specifically defined in (\ref{average_util}).} and $\gamma_k>0$ is  the Lagrangian weight of the  transmit power cost of the $k$-th transmitter. 

The $K$-pair interference network with small cross-channel path gain  is a  weakly coupled multi-dimensional system according to Definition \ref{WCandDC}. The specific association  with the general model  in Section \ref{multidim_sys} is summarized in Table \ref{assoc_eg1}.
\begin{table} \label{Cores1}
	\centering
\begin{tabular}{|c|c|}
	\hline
		Multi-agent control systems & K-pair interference networks  \\
	\hline
		control agent & transmitter  			\\
		sub-system & Tx-Rx pair 	 \\
		$\Delta$	& $\tau$	\\
		$\mathbf{x}$ & $\mathbf{Q}$  		\\
		$\mathbf{x}_k$& $Q_k$	\\
		$\epsilon_{kj} $ &  $L_{kj} $	\\
		$\mathbf{u}$ & $\boldsymbol{\mathbf{p}}$	\\
		$\mathbf{u}_k$ & $\mathbf{p}_k$	\\
		$\mathbf{u}_{-k}$ & $\boldsymbol{\mathbf{p}}_{-k}$		\\	
		$\boldsymbol{\epsilon}_k$ &$\mathbf{L}_k$\\
		$\mathbf{z}_k$ &  $A_k$ 	\\
		$ \mathbf{f}_k\left(\mathbf{u}_k\left(t\right), \mathbf{u}_{-k}\left(t\right),\boldsymbol{\epsilon}_k\right) $ & $-    \mathbb{E} \left[ \log\left( 1+\frac{p_k^{\mathbf{H}}\left(t\right)L_{kk}|H_{kk}|^2}{1 + \sum_{j \neq k} p_j^{\mathbf{H}}\left(t\right) L_{kj}|H_{kj}|^2}\right) \bigg| \mathbf{Q}\left(t\right) \right]  $\\
		$c(\mathbf{x}, \mathbf{u}) = \sum_{k=1}^K  c_k(\mathbf{x}_k, \mathbf{u}_k)$ &  $c\left(\mathbf{Q}, \mathbf{p}\right) = \sum_{k=1}^K c_k(Q_k, \mathbf{p}_k)$	\\
		$c_k(\mathbf{x}_{k}, \mathbf{u}_{k}) = \| \mathbf{x}_k \| _{\mathbf{v}_k,1} + g_k\left(\mathbf{u}_k \right) $ &  $c_k(Q_k, \mathbf{p}_k) =  \beta_k \frac{Q_k}{\lambda_k} + \gamma_k  \mathbb{E} [ p_k^{\mathbf{H}}|\mathbf{Q} ]$ \\
	\hline
\end{tabular}
	\caption{Association between the general model in Section \ref{multidim_sys} and $K$-pair interference networks}
	\label{assoc_eg1}
\end{table}

Define a control policy $\Omega_k$ for the $k$-th Tx-Rx pair according to Definition \ref{StationDef}.  A policy $\Omega_k$ in this example is  feasible if the power allocation action $\mathbf{p}_k$ satisfies the constraint in (\ref{constraintO}). Denote $\Omega= \{ \Omega_k: \forall k \in \mathcal{K} \}$. For a given control policy $\Omega$, the average cost of the system starting from a given initial global QSI $\mathbf{Q}\left( 0\right)$ is given by
\begin{align}	
	\overline{L}^{\Omega}\left( \mathbf{Q}\left( 0\right) \right) &= \limsup_{T \rightarrow \infty} \frac{1}{T} \sum_{t=0}^{T-1} \mathbb{E}^{\Omega} \left[ c\left(\mathbf{Q}\left(t\right),\boldsymbol{\mathbf{p}}\left(t\right)\right)\right]	 \\
	&=\sum_{k=1}^K \left( \beta_k \overline{D}_k^{\Omega}\left( \mathbf{Q}\left( 0\right) \right)   + \gamma_k \overline{P}_k^{\Omega} \left( \mathbf{Q}\left( 0\right) \right)  \right)
  	\label{average_util}
\end{align} 
where the first term $\overline{D}_k^{\Omega}\left( \mathbf{Q}\left( 0\right) \right)  = \limsup_{T \rightarrow \infty} \frac{1}{T} \sum_{t=0}^{T-1} \mathbb{E}^{\Omega} \left[\beta_k\frac{Q_k\left(t\right)}{\lambda_k} \right]$ is the average delay of the $k$-th Tx-Rx pair according to \emph{Little's Law} \cite{littlelaw}, and  the second term $\overline{P}_k^{\Omega}\left( \mathbf{Q}\left( 0\right) \right) = \limsup_{T \rightarrow \infty} \frac{1}{T} \sum_{t=0}^{T-1} \mathbb{E}^{\Omega}  \left[ \mathbb{E}\left[p_k^{\mathbf{H}}\left(t\right)\right]   \right]$  is the average power consumption of the $k$-th transmitter. Similar to Problem \ref{IHAC_MDP}, the associated stochastic optimization problem for this example is given as follows:

\begin{Problem} [Delay-Optimal Control Problem for Interference Networks] \label{cmdp}
For some positive weight constants $\beta_k$, $\gamma_k$ ($\forall k$), the delay-optimal control problem for the $K$-pair interference networks is formulated as
\begin{align} \
	\min_{\Omega}  \overline{L}^{\Omega}\left( \mathbf{Q}\left( 0\right) \right) 
\end{align}
where $ \overline{L}^{\Omega}\left( \mathbf{Q}\left( 0\right) \right) $ is given in 	(\ref{average_util}).~\hfill\IEEEQED
\end{Problem}

The delay-optimal control problem in Problem \ref{cmdp} is an infinite horizon average cost MDP problem \cite{WirelessApp}, \cite{DP_Bertsekas}. Under the stationary unichain policy, the optimal control policy $\Omega^{\ast}$ can be obtained by solving the following Bellman equation w.r.t.  $(\theta, \{ V \left(\mathbf{Q} \right) \})$ according to Lemma \ref{LemBel}:
\begin{equation} 		\label{eq1 org bel}
 \theta \tau+ V \left(\mathbf{Q} \right) = \min_{\Omega(\mathbf{Q})} \left[ c\left( \mathbf{Q}, \Omega(\mathbf{Q})\right) \tau+  (\mathrm{\Gamma}_{\Omega}V)(\mathbf{Q})    \right], \quad \forall \mathbf{Q} \in \boldsymbol{\mathcal{Q}}
\end{equation}
Based on Corollary \ref{cor1}, the associated approximate optimal solution can be obtained by solving the following approximate Bellman equation:
\begin{align}
	 \tilde{\theta}^\ast= \min_{\Omega(\mathbf{Q})} \left[ c\left( \mathbf{Q}, \Omega(\mathbf{Q})\right) + \sum_{k=1}^K \frac{\partial V(\mathbf{Q})  }{\partial Q_k}\left(\lambda_k-R_k\left(\Omega(\mathbf{Q}) \right)\right)    \right], \quad \forall \mathbf{Q} \in \boldsymbol{\mathcal{Q}}
\end{align}

\section{Low Complexity  Distributed Solutions  under Virtual Continuous Time System}   \label{Fluid_Approach} 
In this section, we first define a virtual continuous time system (VCTS) using the fluid limit approximation approach. We  then establish the relationship between  the fluid value function of the VCTS and the relative value function of the ODTS. To address the distributed solution requirement and challenges due to the coupling in control variables, we model the weakly coupled system in  (\ref{DiscreteDym}) as a perturbation of a decoupled base system and derive per-flow fluid value functions to approximate the fluid value function of the multi-dimensional VCTS.  We also establish the associated approximation error using perturbation theory.  Finally, we show that solving the Bellman equation using  per-flow fluid value function approximation is equivalent to solving a deterministic network utility maximization (NUM) problem and we propose a distributed algorithm for solving the associated NUM problem.

\subsection{Virtual Continuous Time System (VCTS) and Total Cost Minimization Problem} \label{VCTS}
Given the multi-dimensional weakly coupled discrete time system in (\ref{DiscreteDym}) and the associated infinite horizon average cost minimization problem in Problem \ref{IHAC_MDP}, we can \emph{reverse-engineer} a  \emph{virtual continuous time system} and an associated total cost minimization problem. While the VCTS can be viewed as a characterization of the mean behavior of the ODTS in (\ref{DiscreteDym}), we will show that the total cost minimization problem of the VCTS has some interesting relationships with the original average cost minimization problem of the  discrete time  system and the solution to the VCTS problem can be used as an approximate solution to the original problem in Problem \ref{IHAC_MDP}. As a result, we can leverage the well-established theory of calculus in continuous time domain to solve the original stochastic optimization problem. 

The VCTS is characterized by a continuous  system state variable $\overline{\mathbf{x}}\left(t\right) = \left( \overline{\mathbf{x}}_1\left(t\right),\dots,\overline{\mathbf{x}}_K\left(t\right)\right) \in \overline{\boldsymbol{\mathcal{X}}} \triangleq \overline{\mathcal{X}}_1 \times \dots \times \overline{\mathcal{X}}_K$,  where  $\overline{\mathbf{x}}_k\left(t\right) \in \overline{\mathcal{X}}_k$ is the virtual state of the k-th sub-system at time $t$. $\overline{\mathcal{X}}_k$ is the virtual sub-system state space\footnote{The virtual sub-system state space $\overline{\mathcal{X}}_k$ is a continuous state space and has the same boundary as  the original sub-system state space $\mathcal{X}_k$.}, which contains the discrete time sub-system state space $\mathcal{X}_k$, i.e., $\overline{\mathcal{X}}_k \supseteq \mathcal{X}_k$.  $\overline{\boldsymbol{\mathcal{X}}} $ is the global virtual system state space\footnote{Because the virtual sub-system state space contains the original discrete time sub-system state space, i.e., $\overline{\mathcal{X}}_k \supseteq \mathcal{X}_k$, we have $\overline{\boldsymbol{\mathcal{X}}} \supseteq \boldsymbol{\mathcal{X}}$.}. Given an initial global virtual system state\footnote{Note that we focus on the initial states that satisfy $\overline{\mathbf{x}}(0)  \in \boldsymbol{\mathcal{X}}$, where $\boldsymbol{\mathcal{X}}$ is the global discrete time system state space.  } $\overline{\mathbf{x}}(0)\in \boldsymbol{\mathcal{X}}$,  the VCTS state trajectory of the $k$-th sub-system state is  described by the following  differential equation:
\begin{equation}		\label{VCTS_dyn}
	\frac {\mathrm{d}}{\mathrm{d}t} \overline{\mathbf{x}}_k\left(t\right)= \mathbf{f}_k\left(\mathbf{u}_k\left(t\right), \mathbf{u}_{-k}\left(t\right),\boldsymbol{\epsilon}_k\right)  + \overline{\mathbf{z}}_k, \quad \overline{\mathbf{x}}_k(0) \in \mathcal{X}_k 
\end{equation}
where $\overline{\mathbf{z}}_{k}$ is the mean of the disturbance process $\mathbf{z}_{k}$ as defined in Assumption \ref{assum_dis}.

For technicality, we have the following assumptions on $ \mathbf{f}_k\left(\mathbf{u}_k, \mathbf{u}_{-k},\boldsymbol{\epsilon}_k\right)$ for all $k \in \mathcal{K}$ in (\ref{VCTS_dyn}).
\begin{Assumption}	[Existence of  Steady State]	\label{steadystatec0}
	We assume that  the VCTS dynamics $\overline{\mathbf{x}}\left(t\right)$ has a steady state, i.e., there exists a control action  $\mathbf{u}^{\infty}=\big(\mathbf{u}_1^{\infty},\dots, \mathbf{u}_K^{\infty} \big)$ such that $\mathbf{f}_k(\mathbf{u}_k^{\infty}, \mathbf{u}_{-k}^{\infty},\boldsymbol{\epsilon}_k)+\overline{\mathbf{z}}_k=\mathbf{0}$ for all $k \in \mathcal{K}$.  Any control action $\mathbf{u}^{\infty}$ that satisfies the above equations is called the \emph{steady state control action}.~\hfill\IEEEQED
\end{Assumption}

Let  $\Omega^v= \{ \Omega_k^v: \forall k \in \mathcal{K}\}$ be the control policy  for the VCTS, where $\Omega_k^v$ is the control policy for the $k$-th sub-system of the VCTS which a mapping from the global virtual system state space $\overline{\boldsymbol{\mathcal{X}}}$ to the  action space $\mathcal{U}_k$.  Given a control policy $\Omega^v$, we define the total cost of the VCTS starting from a given initial global virtual system state $\overline{\mathbf{x}}\left(0 \right)$ as
\begin{equation}		\label{totalU}
	J^{\Omega^v} \left(\overline{\mathbf{x}}\left(0\right) \right)  = \int_0^{\infty} \widetilde{c}\left(\overline{\mathbf{x}}\left(t\right), \Omega^v\left(\overline{\mathbf{x}}\left(t\right)\right)\right) \  \mathrm{d}t, \quad \overline{\mathbf{x}}\left(0\right)  \in \boldsymbol{\mathcal{X}}
\end{equation}
where $\widetilde{c}\left(\overline{\mathbf{x}}, \mathbf{u}\right)=c\left(\overline{\mathbf{x}}, \mathbf{u}\right)-c^{\infty}$ is a \emph{modified} cost function for VCTS.  $c^{\infty} = \sum_{k=1}^K g_k\left(\mathbf{u}_k^{\infty}\right)$ where $\left\{\mathbf{u}_k^{\infty}: \forall k\right\}$ is a steady state control action, i.e., $\mathbf{f}_k(\mathbf{u}_k^{\infty}, \mathbf{u}_{-k}^{\infty},\boldsymbol{\epsilon}_k) + \overline{\mathbf{z}}_k= \mathbf{0}$ for all $k \in \mathcal{K}$.   Note that $c^{\infty}$ is chosen to guarantee that $J^{\Omega^v} \left(\overline{\mathbf{x}}\left(0\right) \right)$ is finite for some policy $\Omega^v$.  

We consider an infinite horizon total cost problem associated with the VCTS as below:
\begin{Problem}		[Infinite Horizon Total Cost Problem for  VCTS]  \label{fluid problem1}
For any initial global virtual system state $\overline{\mathbf{x}}(0)  \in \boldsymbol{\mathcal{X}}$, the infinite horizon total cost problem for the VCTS is formulated as
	\begin{align}
		\min_{\Omega^v} J^{\Omega^v} \left(\overline{\mathbf{x}}\left(0\right)  \right)
	\end{align}
	where $J^{\Omega^v} \left(\overline{\mathbf{x}}\left(0\right)  \right)$ is given in (\ref{totalU}).~\hfill\IEEEQED
\end{Problem}

The above total cost problem has been well-studied in the continuous time optimal control in \cite{DP_Bertsekas} and the solution can be obtained by solving the \emph{Hamilton-Jacobi-Bellman} (HJB) equation as summarized below. 
\begin{Lemma}	[Sufficient Conditions for Optimality  under VCTS]	\label{orghjblem}
If there exists a function $J\left( \mathbf{x}\right)$ of class\footnote{Class $\mathcal{C}^1$ function are those functions whose first order derivatives are continuous.} $\mathcal{C}^1$ that satisfies the following HJB equation:
	\begin{equation}	\label{cenHJB}
		\min_{\mathbf{u}} \left[\widetilde{c}\left(\mathbf{x}, \mathbf{u}\right) + \nabla_{\mathbf{x}} J\left(\mathbf{x}\right) \left[  \mathbf{f}\left( \mathbf{u}, \boldsymbol{\epsilon}\right)+\overline{\mathbf{z}}  \right]^T \right] =0, \quad \mathbf{x} \in \boldsymbol{\mathcal{X}}
	\end{equation}
	with  boundary condition $J(\mathbf{0})=0$, where\footnote{$\nabla_{\mathbf{x}_k} J\left(\mathbf{x}\right)$  is a row  vector with each element being the first order partial derivative of $J\left(\mathbf{x}\right)$ w.r.t. each component in vector $\mathbf{x}_k$.} $\nabla_{\mathbf{x}} J\left(\mathbf{x}\right) \triangleq \left( \nabla_{\mathbf{x}_1} J\left(\mathbf{x}\right), \dots, \nabla_{\mathbf{x}_K} J\left(\mathbf{x}\right) \right)$. For any  initial condition $\overline{\mathbf{x}}\left(0 \right)=\mathbf{x} \in \boldsymbol{\mathcal{X}}$, suppose that a given control $\Omega^{v \ast}$ and the corresponding state trajectory $\overline{\mathbf{x}}^{\ast}\left( t\right)$ satisfies 
	\begin{equation}	\label{suffcond}
 \left\{
	\begin{aligned}	
		& \lim_{t \rightarrow \infty} J\left( \overline{\mathbf{x}}^{\ast}\left(t\right)\right) =0	 \\
		&\Omega^{v \ast} \left(\overline{\mathbf{x}}^{\ast}\left( t\right) \right) = \mathbf{u}^{v \ast}\left( t\right) \in \arg\min_{\mathbf{u}} \left[\widetilde{c}\left(\overline{\mathbf{x}}^{\ast}\left( t\right), \mathbf{u}\right) + \nabla_{\mathbf{x}} J\left(\overline{\mathbf{x}}^{\ast}\left( t\right)\right) \left[  \mathbf{f}\left( \mathbf{u}, \boldsymbol{\epsilon}\right)+\overline{\mathbf{z}}  \right]^T \right], \quad  t \geq 0
	   \end{aligned}
   \right.
  \end{equation}
	Then $J\left(\mathbf{x}\right) =\min_{\Omega^v} J^{\Omega^v} \left( \mathbf{x} \right)$ is the optimal total cost and $\Omega^{v \ast}$ is the optimal control policy for Problem \ref{fluid problem1}. $ J\left(\mathbf{x}\right) $ is called the \emph{fluid value function}.~\hfill\IEEEQED
\end{Lemma}

\begin{proof}
	Please refer to \cite{DP_Bertsekas} for details.
\end{proof}

While the VCTS and the total cost minimization problem are not equivalent to the original weakly coupled discrete time system and the average cost minimization problem, it turns out that the relative value function $V\left(\mathbf{x} \right)$  in Problem \ref{IHAC_MDP}  is closely related to the fluid value function $J\left(\mathbf{x}\right)$ in Problem \ref{fluid problem1}. The following theorem establishes the relationship.  
\begin{Corollary}	[Relationship between VCTS and ODTS]	\label{orderoptlemma}
	If there is $J\left(\mathbf{x} \right)=\mathcal{O}\left(x^n \right)$ for some positive $n$ that satisfies the conditions in (\ref{cenHJB}) and (\ref{suffcond}), then ($c^{\infty}, \{ J\left(\mathbf{x} \right) \}$) is $\mathcal{O}\left(\Delta \right)$ optimal to the Bellman equation of the ODTS in (\ref{OrgBel}).
\end{Corollary}
\begin{proof}  
Please refer to Appendix	 B.	
\end{proof}

\begin{Theorem}	[General Relationship  between  VCTS and  ODTS]  	\label{RelashipVJ}
	For nonlinear system cost function in (\ref{syscost}), the difference between the fluid value function $J\left(\mathbf{x}\right)$ for the VCTS in (\ref{VCTS_dyn}) and the relative value function $V\left(\mathbf{x} \right)$ for the ODTS in (\ref{DiscreteDym}) can be expressed as
	\begin{equation}		\label{diff_VJ}
		 |V( \mathbf{x} ) - J( \mathbf{x})|   = \mathcal{O}\left( \| \mathbf{x}\|\sqrt{\| \mathbf{x}\| \log \log \| \mathbf{x}\|} \right), \quad \text{as }  \|\mathbf{x}\|\rightarrow \infty
	\end{equation}		
\end{Theorem}
where $\| \mathbf{x}\| $ denotes the Euclidean norm of  the system state $\mathbf{x}$.~\hfill\IEEEQED

\begin{proof}  
Please refer to Appendix	 A.	
\end{proof}

\begin{Remark} [Interpretation of Theorem \ref{RelashipVJ}]
Theorem \ref{RelashipVJ} suggests that as the norm of the system state vector  increases, the difference between $V\left(\mathbf{x} \right)$ and $J\left(\mathbf{x}\right)$ is $ \mathcal{O}(\| \mathbf{x}\|\sqrt{\| \mathbf{x}\| \log \log \| \mathbf{x}\|})$\footnote{Throughout the paper, $f\left(x\right)=o\left(g\left(x \right) \right)$ as $x \rightarrow \infty$ ($x \rightarrow 0$) means $\lim_{x \rightarrow \infty} \frac{f\left(x\right)}{g\left(x\right)}=0$ $(\lim_{x \rightarrow 0} \frac{f\left(x\right)}{g\left(x\right)}=0)$.}, i.e., $|V\left(\mathbf{x} \right) - J\left(\mathbf{x}\right)| =\mathcal{O}(\| \mathbf{x}\|\sqrt{\| \mathbf{x}\| \log \log \| \mathbf{x}\|})$, as $\|\mathbf{x}\|\rightarrow \infty$. Therefore, for large system states, the  fluid value function $J\left(\mathbf{x}\right)$ is a useful approximator for the relative value function $V\left(\mathbf{x} \right)$.~\hfill\IEEEQED
\end{Remark}		

As a result of Theorem \ref{RelashipVJ}, we can use $J\left(\mathbf{x}\right)$  to approximate $V\left(\mathbf{x} \right)$ and  the optimal control policy $\Omega^{\ast}$  in (\ref{OrgBel}) can be approximated by solving the following problem:
\begin{align}	\label{JapproxPolicy}
	\Omega^{\ast}\left(\mathbf{x}\right) \approx \arg\min_{ \Omega\left( \mathbf{x} \right)} \left[ c\left(\mathbf{x}, \Omega\left(\mathbf{x}\right)\right) + \left(\mathrm{\Gamma}_{\Omega}J\right)\left(\mathbf{x}\right)  \right]
\end{align}

Note that the result on the relationship between   the VCTS and the ODTS in Theorem \ref{RelashipVJ} holds for any given  epoch duration $\tau$. In the following lemma, we establish an asymptotic equivalence  between  VCTS and  ODTS for sufficiently small epoch duration $\tau$.
\begin{Lemma} 	[Asymptotic Equivalence  between  VCTS and  ODTS]	\label{reasmalltau}
	For sufficiently small epoch duration $\tau$,   the solution of Problem \ref{fluid problem1}  in the VCTS asymptotically solves  Problem \ref{IHAC_MDP} in the ODTS.  In other words,  ($c^{\infty}, \{ J\left(\mathbf{x} \right) \}$) obtained from the HJB equation in (\ref{cenHJB}) solves the simplified Bellman equation in (\ref{simbelman}).~\hfill\IEEEQED
\end{Lemma}

\begin{proof}  
Please refer to Appendix	 B.	
\end{proof}

Hence, to solve the original average cost problem in Problem \ref{IHAC_MDP}, we can solve  the associated total cost problem in Problem \ref{fluid problem1} for the VCTS leveraging the well-established theory of calculus and PDE.  However,  let $N_k$ be the dimension of the sub-system state $\mathbf{x}_k$, then deriving $J\left(\mathbf{x}\right)$ involves solving  a $\sum_{k=1}^K N_k$ dimensional non-linear PDE in (\ref{cenHJB}), which is in general challenging. Furthermore, the VCTS fluid value function $J\left(\mathbf{x}\right)$ will not have decomposable structure in general and hence, global system state information is needed to implement the control policy which solves Problem \ref{IHAC_MDP}. In  Section \ref{PFVCTS}, we introduce the  \emph{per-flow fluid value functions} to further approximate  $J\left(\mathbf{x}\right)$ and use perturbation theory to derive the associated approximation error. In Section \ref{KnTbased}, we derive distributed solutions  based on the per-flow fluid value function approximation.

\subsection{Per-Flow Fluid Value Function of Decoupled Base VCTS and  the Approximation Error}  \label{PFVCTS}
We first define a \emph{decoupled base VCTS} as below:
\begin{Definition} [Decoupled Base VCTS]	\label{base_sys}
	A decoupled base VCTS is the VCTS in (\ref{VCTS_dyn}) with $\epsilon_{kj}=0$ for all $k, j \in \mathcal{K}, k \neq j$.~\hfill\IEEEQED
\end{Definition}

For notation convenience, denote the fluid value function of the VCTS in (\ref{VCTS_dyn}) as $J(\mathbf{x}; \boldsymbol{\epsilon})$. Note that the decoupled base VCTS is a special case of the VCTS in (\ref{VCTS_dyn}) with $\boldsymbol{\epsilon}=\mathbf{0}$ and we have the following lemma summarizing the solution $J(\mathbf{x};\mathbf{0})$ for the decoupled base VCTS. 
\begin{Lemma}	[Sufficient Conditions for Optimality  under Decoupled Base VCTS]	\label{linearAp}
	If there exists a function $J_k\left( \mathbf{x}_k\right)$ of class\footnote{Class $\mathcal{C}^1$ function are those functions whose first order derivatives are continuous.} $\mathcal{C}^1$ that satisfies the following HJB equation:
	\begin{equation}	\label{perflowHJB}
		\min_{\mathbf{u}_k} \left[\widetilde{c}_k(\mathbf{x}_k, \mathbf{u}_k) + \nabla_{\mathbf{x}_k} J_k (\mathbf{x}_k)   \left[ \ \overline{\mathbf{f}}_k\left(\mathbf{u}_k, \mathbf{0}, \mathbf{0}\right) \right]^T \right] =0, \quad \mathbf{x}_k \in \mathcal{X}_k
	\end{equation}
	with boundary condition $J_k(\mathbf{0})=0$. $\overline{\mathbf{f}}_k\left(\mathbf{u}_k, \mathbf{0},\mathbf{0}\right) =  \mathbf{f}_k\left(\mathbf{u}_k, \mathbf{0},\mathbf{0}\right)\tau + \overline{\mathbf{z}}_k$. $\widetilde{c}_k(\mathbf{x}_k, \mathbf{u}_k) =c_k(\mathbf{x}_k, \mathbf{u}_k) -c^{\infty}_k$.   For any  initial condition $\overline{\mathbf{x}}_k\left(0 \right)=\mathbf{x}_k \in \mathcal{X}_k$, suppose that a given control $\Omega_k^{v \ast}$ and the corresponding state trajectory $\overline{\mathbf{x}}_k^{\ast}\left( t\right)$ satisfies 
	\begin{equation}	\label{suffcond11}
 \left\{
	\begin{aligned}	
		& \lim_{t \rightarrow \infty} J_k\left( \overline{\mathbf{x}}_k^{\ast}\left(t\right)\right) =0	 \\
		&\Omega_k^{v \ast} \left(\overline{\mathbf{x}}^{\ast}\left( t\right) \right) = \mathbf{u}_k^{v \ast}\left( t\right) \in \arg\min_{\mathbf{u}_k} \left[\widetilde{c}_k\left(\overline{\mathbf{x}}_k^{\ast}\left( t\right), \mathbf{u}_k\right) + \nabla_{\mathbf{x}_k} J_k\left(\overline{\mathbf{x}}_k^{\ast}\left( t\right)\right) \left[ \ \overline{\mathbf{f}}_k\left(\mathbf{u}_k, \mathbf{0}, \mathbf{0}\right)  \right]^T \right], \quad  t \geq 0
	   \end{aligned}
   \right.
  \end{equation}
	Then  $J_k\left( \mathbf{x}_k\right) $ is the optimal total cost and $\left\{\Omega_k^{v \ast}: \forall k\right\}$ is the optimal control policy for decoupled base VCTS. Therefore, the fluid value function for the decoupled VCTS can be obtained by:
	\begin{align}	\label{linearA}
		J(\mathbf{x}; \mathbf{0}) = \sum_{k=1}^K J_k \left(\mathbf{x}_k \right)
	\end{align}~\hfill\IEEEQED
\end{Lemma}			
\begin{proof} 
Please refer to Appendix C.		
\end{proof}

For the general coupled VCTS in (\ref{VCTS_dyn}), we would like to use the linear architecture in (\ref{linearA}) to approximate $J(\mathbf{x}; \boldsymbol{\epsilon})$, i.e.,
\begin{align}	\label{approxJsumJ}
	 J(\mathbf{x}; \boldsymbol{\epsilon}) \approx J(\mathbf{x}; \mathbf{0}) =\sum_{k=1}^K J_k \left(\mathbf{x}_k \right)
\end{align}
There are two motivations for such approximation:
\begin{itemize}
	\item \textbf{Low Complexity Solution:}
		Deriving $J(\mathbf{x}; \boldsymbol{\epsilon})$ requires solving  a $\sum_{k=1}^K N_k$ dimensional PDE  in (\ref{cenHJB}), while deriving $J_k \left(\mathbf{x}_k \right)$ requires solving  a lower ($N_k$) dimensional  PDE in (\ref{perflowHJB}), which is more manageable and will be illustrated in Section \ref{Example1}.
	\item \textbf{Distributed Control Policy:}
		Approximating $J(\mathbf{x}; \boldsymbol{\epsilon})$ using linear sum of per-flow value functions in (\ref{approxJsumJ}) may facilitate distributed control implementations, which will be illustrated in  Section \ref{Example1}.
\end{itemize} 
Using perturbation analysis, we obtain the approximation error of (\ref{approxJsumJ}) as below:
 
\begin{Theorem}	[Perturbation Analysis of Approximation Error]			\label{RelashipJsJ}
The approximation error of (\ref{approxJsumJ})  is given by\footnote{Throughout the paper, $f\left(x\right)=\mathcal{O}\left(g\left(x \right) \right)$ as $x \rightarrow \infty$ ($x \rightarrow 0$) means that for sufficiently large (small) $x$, there exist positive constants $k_1$ and $k_2$, such that $ k_1 \left|g\left(x \right)\right| \leq \left| f\left(x\right) \right| \leq k_2 \left| g\left(x \right) \right|$. }
	\begin{equation}	\label{GapJ}
		 J(\mathbf{x}; \boldsymbol{\epsilon}) - \sum_{k=1}^K J_k \left(\mathbf{x}_k \right)   =  \sum_{k=1}^K\sum_{j \neq k }\epsilon_{kj} \widetilde{J}_{kj}\left( \mathbf{x} \right) + \mathcal{O}\left( \epsilon^2 \right), \quad \text{as } \epsilon \rightarrow 0
	\end{equation}
	where $\widetilde{J}_{kj}\left( \mathbf{x} \right)$ is the solution of  the following first order PDE:
	\begin{align}		
		\sum_{i=1}^K \nabla_{\mathbf{x}_i} \widetilde{J}_{kj}\left( \mathbf{x} \right) \left[ \ \overline{\mathbf{f}}_i\left(\mathbf{u}_i^{v \ast }\left(\mathbf{x}_i \right), \mathbf{0},\mathbf{0}\right) \right]^T +\nabla_{\mathbf{x}_k}J_k \left(\mathbf{x}_k \right) \left[  \mathbf{u}_j^{v \ast } \left(\mathbf{x}_j \right)  \nabla_{\epsilon_{kj} \mathbf{u}_j}  \overline{\mathbf{f}}_k\left(\mathbf{u}_k^{v \ast }\left(\mathbf{x}_k \right), \mathbf{0} ,\mathbf{0}\right)  \right]^T =0		\label{PDE_Jij}
	\end{align}
with boundary condition 
\begin{equation}
	\widetilde{J}_{kj}\left(\mathbf{x} \right)\Big|_{\mathbf{x}_j=\mathbf{0}}=0
\end{equation}~\hfill\IEEEQED
\end{Theorem}	
\begin{proof} 
Please refer to Appendix D.		
\end{proof}

\begin{Remark} 	[Interpretation of Theorem \ref{RelashipJsJ}]
Theorem \ref{RelashipJsJ} suggests that the  approximation error between $J(\mathbf{x}; \boldsymbol{\epsilon})$ and $\sum_{i=1}^K J_k \left(\mathbf{x}_k \right)$ is  $\sum_{k=1}^K\sum_{j \neq k}\epsilon_{kj} \widetilde{J}_{kj}\left( \mathbf{x} \right) + \mathcal{O}\left( \epsilon^2 \right)$, which is small for weakly coupled  systems where the coupling parameters  are small. Note that the PDE defining $\widetilde{J}_{kj}\left( \mathbf{x} \right)$ in (\ref{PDE_Jij}) involves  $\{J_k \left(\mathbf{x}_k \right): \forall k\}$ and $\{\mathbf{u}_k^{v \ast}(\mathbf{x}_k): \forall k \}$, which can be obtained  by solving the per-flow HJB equation in (\ref{perflowHJB}).~\hfill\IEEEQED 
\end{Remark}

Finally, based on Theorem \ref{RelashipVJ} and Theorem \ref{RelashipJsJ}, we conclude that for all $\mathbf{x} \in \boldsymbol{\mathcal{X}}$, we have
\begin{align}
	V\left(\mathbf{x} \right)= \sum_{k=1}^K J_k \left(\mathbf{x}_k \right)   +  \sum_{k=1}^K\sum_{j \neq k }\epsilon_{kj} \widetilde{J}_{kj}\left( \mathbf{x} \right) + \mathcal{O}\left( \epsilon^2 \right) + \mathcal{O}\left( \| \mathbf{x}\|\sqrt{\| \mathbf{x}\| \log \log \| \mathbf{x}\|} \right)
\end{align}
where $\widetilde{J}_{kj}\left( \mathbf{x} \right)$ is defined in the PDE in (\ref{PDE_Jij}). As a result, we obtain the following per-flow fluid value function approximation:
\begin{equation}	\label{approxVJ}
		V\left(\mathbf{x} \right) \approx \sum_{k=1}^K J_k \left(\mathbf{x}_k \right), \quad    \mathbf{x} \in \boldsymbol{\mathcal{X}}
\end{equation}
We shall illustrate the quality of the approximation in the application example in Section \ref{Example1}.

\subsection{Distributed Solution Based on Per-Flow Fluid Value Function Approximation}   \label{KnTbased}
We first show that minimizing the R.H.S. of  the Bellman equation in (\ref{OrgBel}) using the per-flow fluid value function approximation in (\ref{approxVJ}) is equivalent to solving a deterministic \emph{network utility maximization} (NUM) problem with coupled objectives. We then  propose a distributed iterative algorithm for solving the associated NUM problem.

We have the following lemma on the equivalent NUM problem.
\begin{Lemma}	[Equivalent NUM Problem]	\label{NUMprob}
	Minimizing the R.H.S. of the Bellman equation in  (\ref{OrgBel}) using the per-flow fluid value function approximation in (\ref{approxVJ}) for all $\mathbf{x}\in \boldsymbol{\mathcal{X}}$  is equivalent to solving the following NUM problem:
	\begin{equation}	\label{utility}
		\max_{ \mathbf{u} \in \boldsymbol{\mathcal{U}}  } \sum_{k=1}^K  U_k\left(\mathbf{u}_k,   \mathbf{u}_{-k},\boldsymbol{\epsilon}_k  \right), \quad 
	\end{equation}
\end{Lemma}
where  $U_k\left(\mathbf{u}_k,   \mathbf{u}_{-k},\boldsymbol{\epsilon}_k \right)$ is the utility function of the $k$-th sub-system given by
	\begin{equation}	\label{UtilitySub}
		U_k\left(\mathbf{u}_k,   \mathbf{u}_{-k},\boldsymbol{\epsilon}_k \right) = -g_k\left(\mathbf{u}_k \right)   - \sum_{n=1}^\infty \frac{\nabla_{\mathbf{x}_k}^{\left(n\right)} J_k \left(\mathbf{x}_k\right)}{n!} \left[  \mathbb{E} \left[ \left[ \mathbf{f}_k\left(\mathbf{u}_k, \mathbf{u}_{-k},\boldsymbol{\epsilon}_k\right) \tau + \mathbf{z}_k \right]^{\left(n\right)} \big| \ \mathbf{x}_k, \mathbf{u}  \right] \right]^T
	\end{equation}
where $\nabla_{\mathbf{x}_k}^{\left(n\right)} J_k \left(\mathbf{x}_k\right)$ is a row vector with each element being the $n$-th partial derivative of $J_k \left(\mathbf{x}_k \right)$ w.r.t. each component in vector $\mathbf{x}_k$, and  $\left[ \mathbf{v} \right]^{\left(n\right)}$ is the element-wise power function\footnote{$\left[ \mathbf{v} \right]^{\left(n\right)}$ has the same dimension as $\mathbf{v}$.} with  each element  being the $n$-th power of each component in  $\mathbf{v}$.~\hfill\IEEEQED
\begin{proof}
	Please refer to Appendix E.
\end{proof}

We have the following assumption on the utility function in  (\ref{UtilitySub}).
\begin{Assumption} [Utility Function]	\label{assumfg}
We assume that  the utility function $U_k\left(\mathbf{u}_k,   \mathbf{u}_{-k},\boldsymbol{\epsilon}_k  \right)$ in (\ref{UtilitySub}) is a strictly concave function in $\mathbf{u}_k$ but not necessarily concave in $\mathbf{u}_{-k} $.~\hfill\IEEEQED
\end{Assumption}
\begin{Remark}	[Sufficient Condition for Assumption \ref{assumfg}]
A sufficient condition for Assumption \ref{assumfg} is that $\mathbf{f}_k\left(\mathbf{u}_k, \mathbf{u}_{-k},\boldsymbol{\epsilon}_k\right)$ and $g_k\left(\mathbf{u}_k \right)  $ are both strictly concave functions in $\mathbf{u}_k$.~\hfill\IEEEQED
\end{Remark}
	
Based on  Assumption \ref{assumfg},   the NUM problem in (\ref{utility}) is not necessarily a strictly concave maximization problem  in  control variable $\mathbf{u}$, and  might have several local/global optimal solutions. Solving such problem is difficult in general even for centralized computation. To obtain distributed solutions, the key idea (borrowed from \cite{HuangJW}) is to construct a local optimization problem for each sub-system based on \emph{local observation} and \emph{limited message passing among sub-systems}.   We summarized it in the following theorem.

\begin{Theorem}	[Local Optimization Problem based on Game Theoretical Formulation]	\label{conditionF}
	 For the $k$-th sub-system, there exists a local objective function $F_k(\mathbf{u}_k,   \mathbf{u}_{-k},\boldsymbol{\epsilon}_k , \mathbf{m}_{-k} )$, where $\mathbf{m}_{-k}=\{\mathbf{m}_j: \forall j \neq k \}$  and $\mathbf{m}_j$ is a locally computable message for the $j$-th sub-system. The message $\mathbf{m}_j$ is a function of $\mathbf{u}$, i.e., $\mathbf{m}_j=\mathbf{h}_j(\mathbf{u})$ for some function $\mathbf{h}_j$. We require that $F_k(\mathbf{u}_k,   \mathbf{u}_{-k}, \boldsymbol{\epsilon}_k, \mathbf{m}_{-k} )$ is strictly concave in $\mathbf{u}_k$ and satisfies the following condition:
	 \begin{equation}		\label{cond}
	\nabla_{\mathbf{u}_k} F_k(\mathbf{u}_k,   \mathbf{u}_{-k},\boldsymbol{\epsilon}_k, \mathbf{m}_{-k} ) = \nabla_{\mathbf{u}_k} U_k\left(\mathbf{u}_k,   \mathbf{u}_{-k}, \boldsymbol{\epsilon}_k\right) + \sum_{j \neq k}  \nabla_{\mathbf{u}_k} U_j\left(\mathbf{u}_j,   \mathbf{u}_{-j}, \boldsymbol{\epsilon}_j  \right)
\end{equation}
Define a non-cooperative game \cite{Game} where the players are the agents of each  sub-system and the payoff function for each sub-system is $F_k(\mathbf{u}_k,   \mathbf{u}_{-k}, \boldsymbol{\epsilon}_k, \mathbf{m}_{-k} )$.  Specifically, the game has the following structure:
	\begin{equation}	\label{local_opt}
		(\mathcal{G}): \quad \max_{\mathbf{u}_k \in \mathcal{U}_k} F_k(\mathbf{u}_k,   \mathbf{u}_{-k}, \boldsymbol{\epsilon}_k,\mathbf{m}_{-k} ), 
		\quad \forall k	\in \mathcal{K}
	\end{equation}
\end{Theorem}
We conclude that a Nash Equilibrium\footnote{$\mathbf{u}^{\ast} = \{  \mathbf{u}_1^{\ast}, \dots,  \mathbf{u}_K^{\ast} \}$ is a NE if and only if $F_k(\mathbf{u}_k^{\ast},   \mathbf{u}_{-k}^{\ast} ,\boldsymbol{\epsilon}_k, \mathbf{m}_{-k}^{\ast} ) \geq F_k(\mathbf{u}_k,   \mathbf{u}_{-k}^{\ast}, \boldsymbol{\epsilon}_k, \mathbf{m}_{-k}^{\ast} )$,  $\forall \mathbf{u}_k \in \mathcal{U}_k$, $\forall k$, where $\mathbf{u}_{-k}^{\ast} = \{\mathbf{u}_j^{\ast}: \forall j \neq k\}$  and $\mathbf{m}_{-k}^{\ast}=\{m_j^{\ast}: \forall j \neq k, m_j^{\ast} = h_j(\mathbf{u}^{\ast})  \}.$} (NE) of the game $\mathcal{G}$ is a stationary point\footnote{A stationary point satisfies the KKT conditions of  the NUM problem  in (\ref{utility}).}  of the NUM problem in (\ref{utility}).~\hfill\IEEEQED
\begin{proof}
	Please refer to \cite{HuangJW} for details.
\end{proof}

Based on the game structure in Theorem \ref{conditionF}, we propose the following distributed iterative algorithm to achieve a NE of the game $\mathcal{G}$ in (\ref{local_opt}).
\begin{Algorithm}	[Distributed Iterative Algorithm] \label{distgen}	\
	\begin{itemize}
	\item \textbf{Step 1 (Initialization):} Let $n=0$. Initialize $\mathbf{u}_k(0) \in \mathcal{U}_k$ for each sub-system $k$.
	\item \textbf{Step 2 (Message Update and Passing):}  Each sub-system k updates message $\mathbf{m}_k\left(n\right)$ according to the following equation:
		\begin{align}
			\mathbf{m}_k\left(n\right) = \mathbf{h}_k\left(\mathbf{u}\left(n\right) \right) 	\label{updateSMS}
		\end{align}
		and announces it to the other sub-systems.
	\item \textbf{Step 3 (Control Action Update): } Based on $\mathbf{m}_{-k}\left(n\right)$, each sub-system $k$ updates the control action $\mathbf{u}_k\left(n+1\right)$ according to
		\begin{align}
			\mathbf{u}_k\left(n+1\right)= \arg \max_{\mathbf{u}_k \in \mathcal{U}_k } F_k(\mathbf{u}_k,   \mathbf{u}_{-k}\left(n\right) , \boldsymbol{\epsilon}_k, \mathbf{m}_{-k}\left(n\right) )		\label{updateAct}   
		\end{align}	
	\item \textbf{Step 4 (Termination):} Set $n=n+1$ and go to Step 2 until a certain termination condition\footnote{For example, the termination condition can be chosen as $\| \mathbf{u}_k\left(n+1 \right) - \mathbf{u}_k\left(n\right)\| < \delta_k$ for some threshold $\delta_k$.} is satisfied.
	\end{itemize}
\end{Algorithm}
\begin{Remark}	[Convergence Property of Algorithm \ref{distgen}]
	The proof of convergence for Algorithm \ref{distgen} is shown in \cite{HuangJW}. The limiting point $\mathbf{u}(\infty)$ is a NE of the game $\mathcal{G}$ in (\ref{local_opt}), and thus is a stationary point of the NUM problem in (\ref{utility}) according to Theorem \ref{conditionF}.~\hfill\IEEEQED
\end{Remark}

While the NUM problem in (\ref{utility}) is not convex in general, the following corollary states that the limiting point $\mathbf{u}\left( \infty\right)$ of Algorithm\ref{distgen}  is asymptotically optimal for sufficiently small coupling parameter $\boldsymbol{\epsilon}$. 

\begin{Corollary}  [Asymptotic Optimality of Algorithm \ref{distgen}]  	\label{collaryAlg}
As the coupling parameter $\boldsymbol{\epsilon}$  goes to zero,  Algorithm \ref{distgen} converges to the unique global optimal point of the NUM problem in (\ref{utility}).~\hfill\IEEEQED
\end{Corollary}
\begin{proof}
Please refer to Appendix F.
\end{proof}

In the next section, we shall elaborate on the low complexity distributed solutions for the application example introduced  in Section \ref{ExampleSec} based on the analysis in this section.

\section{Low Complexity Distributed Solutions  for Interference Networks}  \label{Example1}
In this section, we  apply  the low complexity distributed solutions  in Section \ref{Fluid_Approach} to  the  application example  introduced in Section \ref{ExampleSec}. We  first obtain the associated VCTS and  derive the per-flow fluid value function. We then discuss  the associated approximation error using Theorem \ref{RelashipJsJ}.  Based on  Algorithm \ref{distgen}, we propose a distributed control algorithm using  per-flow fluid value function approximation. Finally, we  compare the delay performance gain of the proposed algorithm with several baseline schemes using numerical simulations.

\subsection{Per-Flow Fluid Value Function}	\label{eg1perflow}
 We first consider the associated VCTS for the interference networks in the application example. Let $\mathbf{q}\left(t\right)=\left( q_1\left(t\right), \dots, q_K\left(t\right) \right) \in \overline{\boldsymbol{\mathcal{Q}}} \triangleq  \overline{\mathcal{Q}}^K$ be the global virtual queue state at time $t$, where $q_k\left(t\right) \in  \overline{\mathcal{Q}}$ is the virtual queue state of the $k$-th transmitter and $\overline{\mathcal{Q}}$ is the virtual queue state space\footnote{Here the virtual queue state space $\overline{\mathcal{Q}}$  is the set of nonnegative real numbers, while the original discrete time queue state space $\mathcal{Q}$ is the set of nonnegative integer numbers.} which contains $\mathcal{Q}$, i.e., $\overline{\mathcal{Q}} \supseteq \mathcal{Q}$. $\overline{\boldsymbol{\mathcal{Q}}} $  is the global virtual queue state space.  $q_k\left(t\right)$ in this example corresponds to the virtual state of the $k$-th sub-system of the VCTS in Section \ref{VCTS}. Therefore,  for a given initial global virtual queue state $\mathbf{q}(0) \in \boldsymbol{\mathcal{Q}}$, the VCTS queue state trajectory of the $k$-th transmitter is given by
\begin{align}	\label{VCTSeg1}
	\frac{ \mathrm{d}} {\mathrm{d}t}  q_{k}\left(t\right) = - R_k\left(t\right) \tau  + \lambda_k, \quad q_k(0) \in \mathcal{Q}
\end{align}
where $\lambda_k$ is the average data arrival rate of the $k$-th transmitter as defined in Assumption \ref{assumeA} and $R_k\left(t\right)$ is the ergodic data rate in (\ref{rate1}).  Starting from a  global virtual queue state $\mathbf{q}(0)= \mathbf{q}  \in \boldsymbol{\mathcal{Q}}$, we denote the optimal total cost of the VCTS, i.e., the VCTS fluid value function as $J(\mathbf{q}; \mathbf{L})$, where $\mathbf{L}=\{ L_{kj}: \forall k,j \in \mathcal{K}, k \neq j \}$ is the collection of the cross-channel path gains that correspond to the coupling parameter $\boldsymbol{\epsilon}$ as shown in Table 1.  According to Definition \ref{base_sys},  the associated decoupled base VCTS  is obtained by setting $L_{kj}=0$ for all $k,j \in \mathcal{K}, k \neq j$ in (\ref{VCTSeg1}).  Using Lemma \ref{linearAp}, for all $\mathbf{q} = (q_1, \dots, q_K) \in \boldsymbol{\mathcal{Q}}$, the fluid value function of the decoupled based VCTS $J(\mathbf{q}; \mathbf{0})$  has a linear architecture $J(\mathbf{q}; \mathbf{0}) = \sum_{k=1}^K J_k \left( q_k \right)$, where  $J_k \left( q_k \right)$ is the per-flow fluid value function  of the $k$-th Tx-Rx pair. Before obtaining the closed-form per-flow fluid value function $J_k \left( q_k \right)$, we calculate $c_k^{\infty}$ based on the sufficient conditions for optimality under decoupled base VCTS in Lemma \ref{linearAp} and $c_k^{\infty}$ is given by
\begin{align}
	c_k^{\infty}=v_k\tau e^{- \frac{\gamma_k}{ L_{kk} v_k\tau}} - \frac{\gamma_k}{L_{kk}}  E_1\left(\frac{\gamma_k}{ L_{kk} v_k\tau}\right)
\end{align}
where $v_k$ satisfies $ E_1\left(\frac{\gamma_k}{\tau L_{kk} v_k}\right)\tau=\lambda_k$. Therefore, the modified cost function for the decoupled base VCTS in this example is given by
\begin{align}
	\widetilde{c}_k(Q_k, \mathbf{p}_k) =  \beta_k \frac{Q_k}{\lambda_k} + \gamma_k  \mathbb{E} [ p_k^{\mathbf{H}}|\mathbf{Q} ] - c_k^{\infty}
\end{align}
By solving the associated per-flow  HJB equation, we can obtain the per-flow fluid value function  $J_k \left( q_k \right)$ which is given in the following lemma:
\begin{Lemma}[Per-Flow Fluid Value Function for Interference Networks]	\label{Eg1Per}
	The per-flow fluid value function $J_k \left( q_k \right)$  is given below in a parametric form w.r.t. $y$:
\begin{equation}
 \left\{
	\begin{aligned}	 \label{pareJkEg1}	
		q_k(y) &=  \frac{\lambda_k \tau}{\beta_k} \left(  \left(\frac{1}{a_k}+y\right)E_1\left(\frac{1}{a_k y}\right) - \frac{\lambda_k }{\tau}y   -  e^{- \frac{1}{a_k y}} y + \frac{c_k^{\infty}}{\tau}\right)  		 \\
		J_k(y) &=   \frac{\lambda_k \tau}{\beta_k} \left( \frac{\left(1-a_k y\right)}{4a_k}ye^{-\frac{1}{a_k y}} -\frac{\lambda_k }{2 \tau}y^2 +\left(\frac{y^2}{2}-\frac{1}{4a_k^2}\right) E_1\left(\frac{1}{a_k y}\right) \right)  + b_k
	   \end{aligned}
   \right.
  \end{equation}
\end{Lemma}
where $a_k=\frac{\tau L_{kk}}{\gamma_k}$, $E_1(x)= \int_1^{\infty} \frac{e^{-tx}}{t}\mathrm{d}t$ is the exponential integral function, and $b_k$ is chosen such that the boundary condition $J_k(0)=0$ is satisfied\footnote{To find $b_k$, first solve $q_k(y_k^0)=0$ using one dimensional search techniques (e.g., bisection method). Then $b_k$ is chosen such that $J_k(y_k^0)=0$.}.~\hfill\IEEEQED

\begin{proof} 
Please refer to Appendix G  for the proof of Lemma \ref{Eg1Per} and the derivation of $c_k^{\infty}$.
\end{proof}

The following corollary summarizes the asymptotic behavior of $J_k \left( q_k \right)$.

\begin{Corollary}	[Asymptotic Behavior of $J_k \left( q_k \right)$]   	\label{PropJk}
	\begin{align}	\label{asympoticJ_k}
		J_k \left( q_k \right) = \frac{\beta_k}{\lambda_k \tau} \mathcal{O} \left( \frac{q_k^2}{\log \left( q_k \right)} \right), \quad \text{as } q_k \rightarrow \infty
	\end{align}~\hfill\IEEEQED
\end{Corollary}
\begin{proof} 
Please refer to Appendix H.
\end{proof}

\begin{figure}
  \centering
  \includegraphics[width=4in]{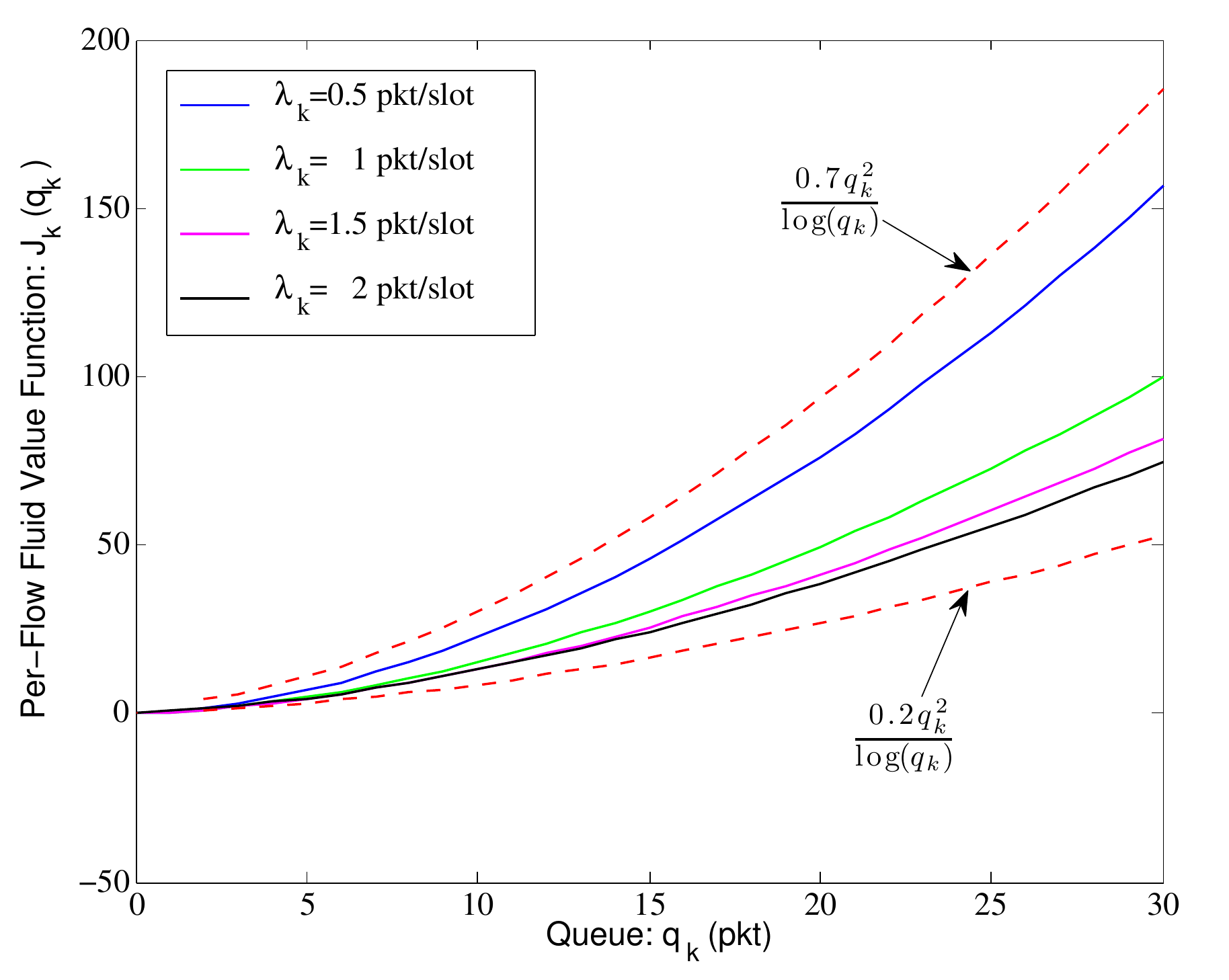}
  \caption{Per-Flow fluid value function $J_k\left(q_k \right)$ versus queue state $q_k$ with $\tau=5$ms, $\gamma_k=0.05$ and $L_{kk}=1$  for all $k \in \mathcal{K}$. The two dashed lines represent the functions $\frac{0.2 q_k^2}{\log\left(q_k \right)}$ and $\frac{0.7q_k^2}{\log\left(q_k \right)}$.}
  \label{sim2}
\end{figure}

Corollary \ref{PropJk} suggests that for large queue state $q_k$, the per-flow fluid value function increases at the order of $\frac{q_k^2}{\log \left( q_k \right)}$ and is a decreasing function of the average arrival rate $\lambda_k$. The analytical result in (\ref{asympoticJ_k}) is also verified in Fig.~\ref{sim2}.

\subsection{Analysis of Approximation Error}
Based on Theorem \ref{RelashipJsJ}, the approximation error of using the linear architecture $\sum_{k=1}^K J_k \left( q_k \right)$ to approximate  $J(\mathbf{q}; \mathbf{L})$ is given in the following lemma.
\begin{Lemma}	[Analysis of Approximation Error for Interference Networks]	\label{ErrorEg1}
	 The approximation error between $J(\mathbf{q}; \mathbf{L})$ and $\sum_{k=1}^K J_k \left( q_k \right)$ is given by
	 \begin{equation}	\label{Ex1Error}
		\big| J(\mathbf{q}; \mathbf{L}) - \sum_{k=1}^K J_k \left( q_k \right)  \big| =  \sum_{k=1}^K \sum_{j \neq k} L_{kj} D_{kj} \mathcal{O} \left(\frac{q_k q_j^2+q_k^2 q_j}{\log q_k \log q_j} \right)  + \mathcal{O}(L^2), \quad \text{as } q_j, q_k \rightarrow \infty,   L \rightarrow 0
	\end{equation}
	where the coefficient $D_{kj}$ is given by
	\begin{equation}	\label{Dexpr}
		D_{kj} =  \frac{\beta_k \beta_j}{\gamma_j \lambda_k \lambda_j \tau^2}
	\end{equation}~\hfill\IEEEQED
\end{Lemma}
\begin{proof} 
Please refer to Appendix I.		
\end{proof}

We have the following remark discussing the approximation error in (\ref{Ex1Error}).

\begin{Remark} [Approximation Error w.r.t. System Parameters]	\label{link_sim}
The dependence between the  approximation error in (\ref{Ex1Error}) and  system parameters are given below:
\begin{itemize}
	\item \textbf{Approximation Error w.r.t. Traffic Loading:}	
		 the approximation error is a decreasing function of  the average arrival rate $\lambda_k$.
	\item \textbf{Approximation Error w.r.t. SNR:} 
		the approximation error is an increasing function of  the SNR per Tx-Rx pair (which is a decreasing function of $\gamma_k$).~\hfill\IEEEQED
	\end{itemize}
\end{Remark}

\subsection{Distributed Power Control Algorithm}

As discussed in Section \ref{KnTbased}, we can use $\sum_{k=1}^K J_k \left( q_k \right)$ to approximate $V \left(\mathbf{Q} \right)$. Fig.~\ref{approxqua} illustrates the quality of the approximation. It can be observed that the sum of the per-flow fluid value functions is a good approximator  to the relative value function for both high and low transmit SNR. Using the per-flow fluid value function approximation, i.e.,
\begin{equation}	\label{approxEg1}
	V \left(\mathbf{Q} \right) \approx \sum_{k=1}^K J_k \left( q_k \right), \quad \forall \mathbf{Q}  \in \boldsymbol{\mathcal{Q}}
\end{equation}
the distributed power control for the interference network can be obtained by solving the following NUM problem according to Lemma \ref{NUMprob}.

\begin{figure}
\centering
\subfigure[Low Tx SNR regimes with $\gamma_k = 200$, $\gamma_k = 300$, $\gamma_k = 400$, which corresponds to the average Tx SNR per pair being 3.5dB, 3.1dB, 2.9dB (under optimal policy), respectively.]{
\includegraphics[width=2.9in]{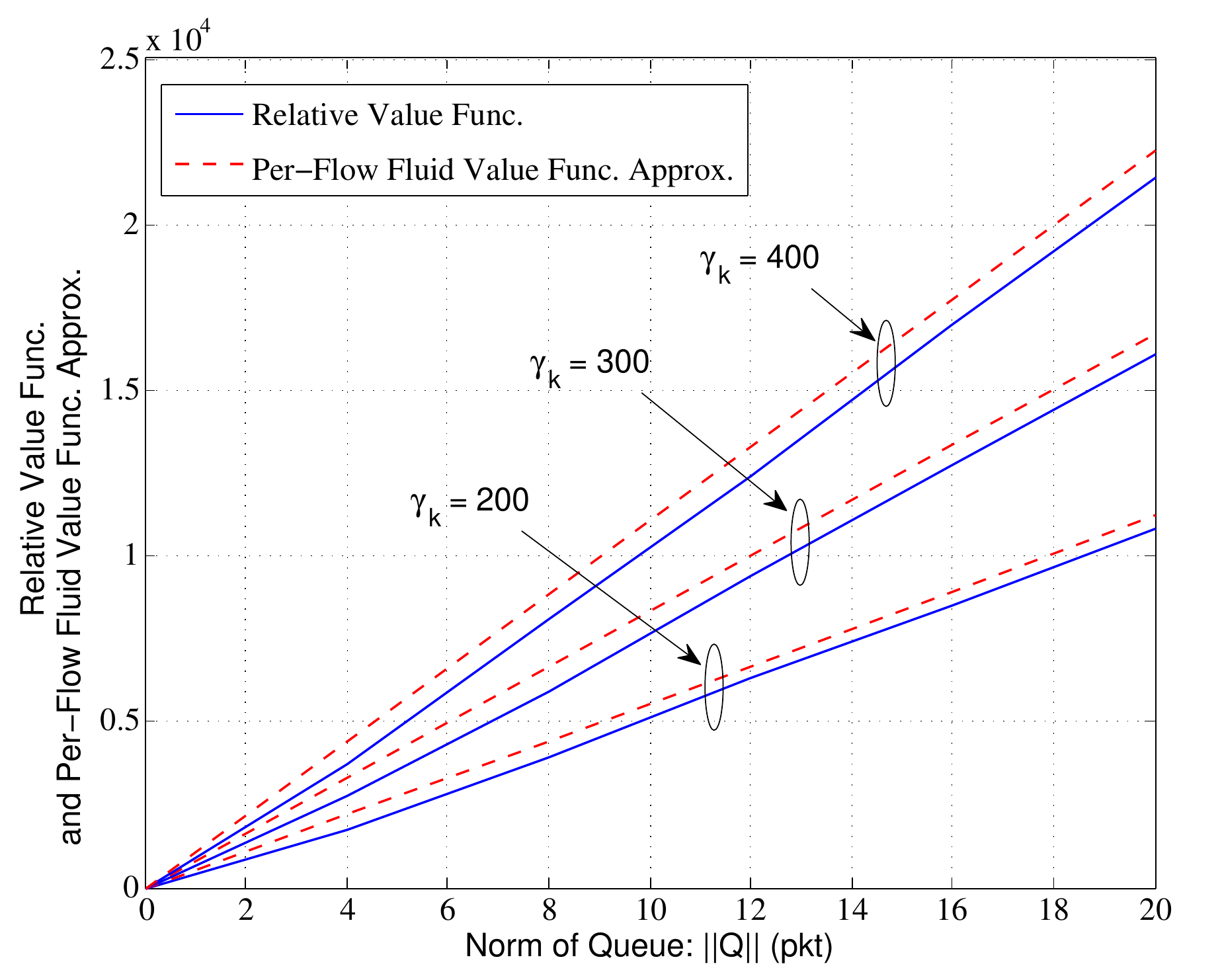}}
\hspace{0.5cm}
\subfigure[High Tx SNR regimes with  $\gamma_k = 0.1$, $\gamma_k = 0.3$, $\gamma_k = 0.6$,  which corresponds to the average Tx SNR per pair being 9.7dB, 9.1dB, 8.7dB (under optimal policy), respectively.]{
\includegraphics[width=3in]{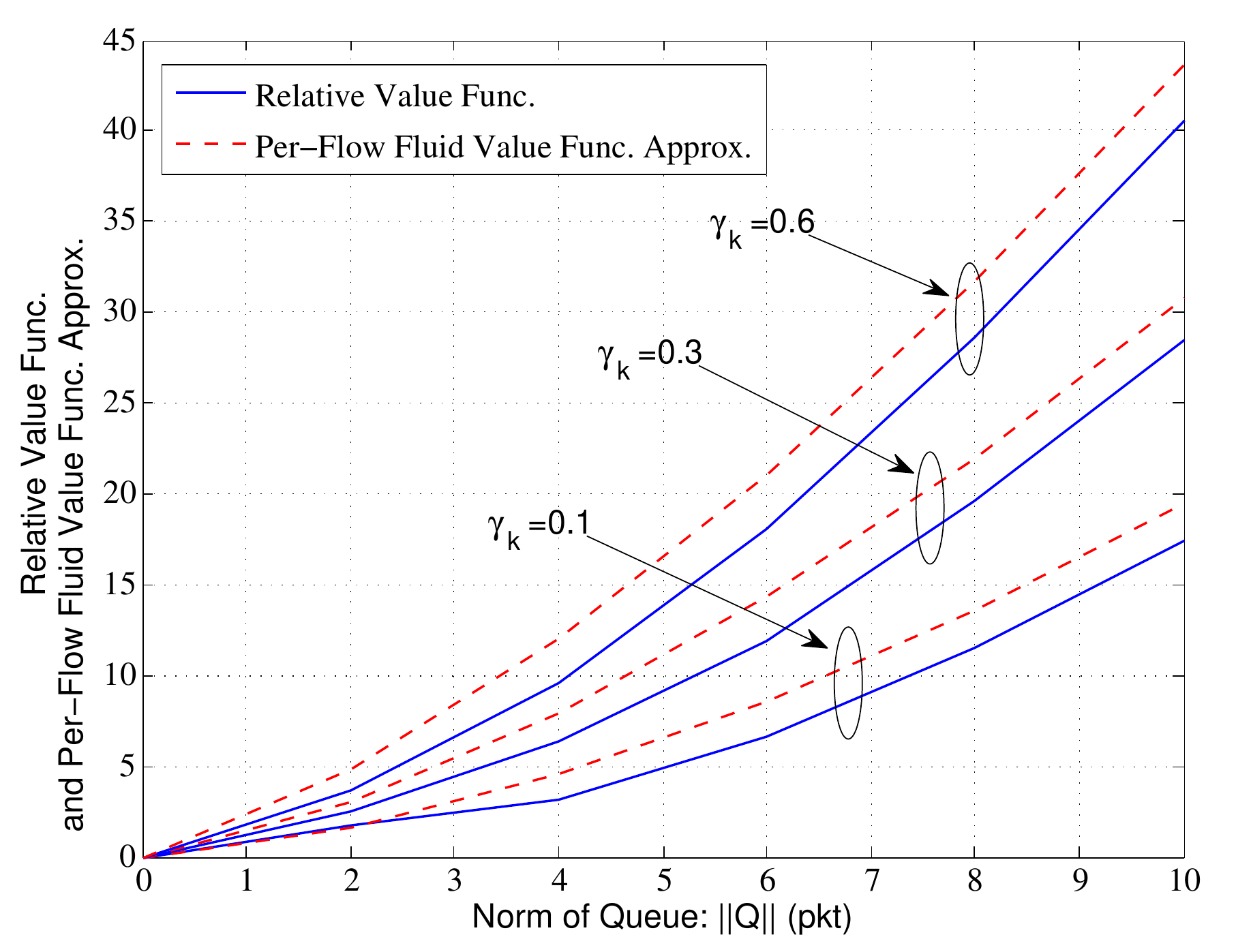}}
\caption{Relative value function $V\left(\mathbf{Q} \right)$ and per-flow fluid value function approximation $\sum_{k=1}^K J_k \left( q_k\right)$ versus the norm of the global queue state $\| \mathbf{Q} \|$ with $\mathbf{Q}=\{q_1, 0, \dots, 0 \}$. The system parameters are configured as in Section \ref{simulation}. Note that the relative value functions are calculated using value iteration \cite{DP_Bertsekas}.}
\label{approxqua}
\end{figure}

\begin{Lemma}  [Equivalent NUM Problem for Interference Networks]
	Minimizing the R.H.S. of  the  Bellman equation in (\ref{eq1 org bel}) using the per-flow fluid value function approximation in (\ref{approxEg1}) for all $\mathbf{Q}\in \boldsymbol{\mathcal{Q}}$ is equivalent to  the following  NUM problem:
	\begin{equation}	\label{ApproxPol}
		 \max_{\boldsymbol{\mathbf{p}}}   \ \sum_{k=1}^K U_k(\mathbf{p}_k, \boldsymbol{\mathbf{p}}_{-k},\mathbf{L}_k)
	\end{equation}
where $\boldsymbol{\mathbf{p}}$, $\mathbf{p}_k$, $\boldsymbol{\mathbf{p}}_{-k}$ and $\mathbf{L}_k$ are defined in Section \ref{ExampleSec}.  $U_k(\mathbf{p}_k, \boldsymbol{\mathbf{p}}_{-k},\mathbf{L}_k)$ is  the utility function of the $k$-th Tx-Rx pair given by\footnote{Note that $J_k'\left(Q_k \right)$ in the utility function in (\ref{UtilityEg1}) can be calculated as follows: $J_k'\left(Q_k \right) = \left(\frac{\mathrm{d} J_k\left(y\right)}{\mathrm{d} y } \big/\frac{\mathrm{d} q_k\left(y\right)}{\mathrm{d} y }\right)\Big|_{y=y\left(Q_k \right)} = y\left(Q_k \right)$, where $ y\left(Q_k \right)$ satisfies $q_k\left(  y\left(Q_k \right) \right)=Q_k$ in (\ref{pareJkEg1}).}
	\begin{equation}	\label{UtilityEg1}
		U_k(\mathbf{p}_k, \boldsymbol{\mathbf{p}}_{-k},\mathbf{L}_k) =  \mathbb{E} \left[ J_k'\left(Q_k \right) \tau \log\left( 1+\frac{p_k^{\mathbf{H}}L_{kk} \left|H_{kk}\right|^2}{1 + \sum_{j \neq k} p_j^{\mathbf{H}} L_{kj} \left|H_{kj}\right|^2 }\right) - \gamma_k p_k^{\mathbf{H}} \Bigg|   \mathbf{Q} \right] + \mathcal{O}(\tau^2), \text{as  }\tau \rightarrow 0
	\end{equation}~\hfill\IEEEQED
\end{Lemma}

For sufficiently small epoch duration $\tau$,  the term $\mathcal{O}(\tau^2)$ is negligible. Note that the utility function  $U_k(\mathbf{p}_k, \boldsymbol{\mathbf{p}}_{-k},\mathbf{L}_k)$ in (\ref{UtilityEg1}) is  strictly concave  in $\mathbf{p}_k$ but  convex  in $\boldsymbol{\mathbf{p}}_{-k}$.  Hence, Assumption \ref{assumfg} holds for $U_k(\mathbf{p}_k, \boldsymbol{\mathbf{p}}_{-k},\mathbf{L}_k)$ in (\ref{UtilityEg1}) in this example. 

Based on Theorem \ref{conditionF}, we choose a local objective function  $F_k(\mathbf{p}_k, \boldsymbol{\mathbf{p}}_{-k},\mathbf{L}_k, \mathbf{m}_{-k})$ for the $k$-th sub-system as follows\footnote{The condition in (\ref{cond}) can be easily verified by substituting the expressions of $U_k$ in (\ref{UtilityEg1}) and  $F_k$ in (\ref{exprF}) into (\ref{cond}).}:
\begin{align}	
	F_k(\mathbf{p}_k, \boldsymbol{\mathbf{p}}_{-k},\mathbf{L}_k, \mathbf{m}_{-k})  =  U_k(\mathbf{P}_k, \boldsymbol{\mathbf{P}}_{-k},\mathbf{L}_k) - \mathbb{E} \left[p_k^{\mathbf{H}}  \sum_{j \neq k} m_j^{\mathbf{H}} |L_{jk}H_{jk}|^2  \Bigg| \mathbf{Q} \right]	\label{exprF}
\end{align} 
where $\mathbf{m}_{-k} =\{\mathbf{m}_j: \forall j \neq k \}$ and $\mathbf{m}_j$ is the message of the $j$-th Tx-Rx pair given by
\begin{align}	
	\mathbf{m}_j=\left\{m_j^{\mathbf{H}}=\frac{J_j'(Q_j)\tau \Upsilon_j^{\mathbf{H}}}{\boldsymbol{\Psi}_j^{\mathbf{H}} }: \forall \mathbf{H} \right\}	\label{exprM}
\end{align}
where $\Upsilon_j^{\mathbf{H}} = \frac{p_j^{\mathbf{H}}L_{jj} |H_{jj}|^2}{1+\sum_{k \neq j} p_k^{\mathbf{H}} L_{jk}|H_{jk}|^2}$ and $\boldsymbol{\Psi}_j^{\mathbf{H}}=1+ \sum_{k=1}^K p_k^{\mathbf{H}} L_{jk} |H_{jk}|^2$.  Note that  $\Upsilon_j^{\mathbf{H}}$  and $\boldsymbol{\Psi}_j^{\mathbf{H}}$ are the  SINR and the total received signal plus noise at receiver $j$  for a given CSI realization $\mathbf{H}$, respectively and both are locally measurable.  The associated game for this example is given by
\begin{equation}	\label{gameEg1}
	(\mathcal{G}1): \quad \max_{\mathbf{p}_k} F_k(\mathbf{p}_k, \boldsymbol{\mathbf{p}}_{-k},\mathbf{L}_k, \mathbf{m}_{-k}), \quad \forall k \in \mathcal{K}
\end{equation}
The distributed iterative algorithm solving  the game in (\ref{gameEg1}) can be obtained from Algorithm \ref{distgen} by replacing variable $\mathbf{u}_k$ with $\mathbf{p}_k$, message $\mathbf{m}_k$ in (\ref{updateSMS}) with (\ref{exprM}), respectively.   Furthermore, according to Corollary \ref{collaryAlg}, as $\mathbf{L}$ goes to zero, the algorithm converges to the unique global optimal point of the NUM problem in (\ref{ApproxPol}) for this example.

Define $\boldsymbol{\mathbf{p}}_{-k}^{\mathbf{H}}=\left\{p_j^{\mathbf{H}}: \forall j \neq k \right\}$ and $\mathbf{m}_{-k}^{\mathbf{H}} =\left\{m_j^{\mathbf{H}}: \forall j \neq k  \right\}$ as the collection of coupling effects and messages for a given CSI realization $\mathbf{H}$. Then, the  objective function in (\ref{gameEg1}) can be written as
\begin{align}
	F_k\left(\mathbf{p}_k, \boldsymbol{\mathbf{p}}_{-k},\mathbf{L}_k, \mathbf{m}_{-k} \right)=  \mathbb{E}\left[  f_k\left(p_k^{\mathbf{H}}, \boldsymbol{\mathbf{p}}_{-k}^{\mathbf{H}},\mathbf{L}_k, \mathbf{m}_{-k}^{\mathbf{H}}  \right) \Big| \mathbf{Q}\right]		\label{structureF}
\end{align}
where $f_k\left(p_k^{\mathbf{H}}, \boldsymbol{\mathbf{p}}_{-k}^{\mathbf{H}},\mathbf{L}_k,\mathbf{m}_{-k}^{\mathbf{H}}  \right) = J_k'\left(Q_k \right) \tau \log\left( 1+\frac{p_k^{\mathbf{H}}L_{kk} \left|H_{kk}\right|^2}{1 + \sum_{j \neq k} p_j^{\mathbf{H}} L_{kj} \left|H_{kj}\right|^2 }\right) - p_k^{\mathbf{H}} \left( \sum_{j \neq k} m_j^{\mathbf{H}} L_{jk} |H_{jk}|^2+\gamma_k \right)$.

Based on the structure of $F_k\left(\mathbf{p}_k, \boldsymbol{\mathbf{p}}_{-k},\mathbf{L}_k, \mathbf{m}_{-k} \right)$ in (\ref{structureF}), the solution of the game in (\ref{gameEg1}) can be further decomposed into per-CSI control  as illustrated in the following Algorithm \ref{distCAeg1}.  
\begin{Algorithm}	[Distributed Power Control Algorithm for Interference Networks]	\label{distCAeg1}\
	\begin{itemize}
		\item \textbf{Step 1 [Information Passing within Each Tx-Rx Pair]}:	 at the beginning of the $t$-th epoch, the $k$-th transmitter notifies the value of $J_k'\left(Q_k\left(t\right)\right)$ to the $k$-th receiver. 
		\item \textbf{Step 2  [Calculation of Control Actions]}: at the beginning of each time slot within the $t$-th epoch with the CSI realization being $\mathbf{H}$, each transmitter determines the transmit power $p_k^{\mathbf{H}}$  according to the following per-CSI distributed power allocation  algorithm:
		\begin{Algorithm}		[Per-CSI Distributed Power Allocation Algorithm]	\label{PCSIDPA}\
		\begin{itemize}	
			\item \textbf{Step 1 [Initialization]:}  Set $n=0$. Each transmitter initializes $p_k^{\mathbf{H}}(0)$. 
			\item \textbf{Step 2 [Message Update and Passing]:}  Each receiver $k$ locally estimates the SNR $\Upsilon_m^{\mathbf{H}}\left(n\right)$  and the total received signal plus noise $\boldsymbol{\Psi}_k^{\mathbf{H}}\left(n\right)$. Then, each receiver $k$ calculates $m_k^{\mathbf{H}}\left(n\right)$ according to (\ref{exprM}) and broadcasts  $m_k^{\mathbf{H}}\left(n\right)$ to all the transmitters.
			\item \textbf{Step 3 [Power Action Update]:} After receiving   messages $\{ m_k^{\mathbf{H}}\left(n\right) \}$ from all the receivers, each transmitter locally updates $p_k^{\mathbf{H}}\left(n+1\right)$ according to
			\begin{align}	\
			p_k^{\mathbf{H}}\left(n+1\right) &=  \arg \max_{p_k^{\mathbf{H}}} f_k\left(p_k^{\mathbf{H}}, \boldsymbol{\mathbf{p}}_{-k}^{\mathbf{H}} \left(n \right),\mathbf{L}_k, \mathbf{m}_{-k}^{\mathbf{H}} \left(n \right) \right) \notag \\ 
	 		& = \min\left\{ \left( \frac{J_k'\left(Q_k \left( t\right)\right)\tau}{\sum_{j \neq k} m_j^{\mathbf{H}}\left(n \right)  L_{jk} |H_{jk}\left(n \right) |^2  + \gamma_k} -\frac{1+I_k\left(n \right) }{L_{kk} |H_{kk}\left(n \right) |^2} \right)^+, \  p_k^{up}\left(t_n \right)	\right\}  \label{controlDPA}
			\end{align}
			where $I_k\left(n \right)=\sum_{j \neq k}	p_j^{\mathbf{H}}\left(n \right) L_{kj} |H_{kj}\left(n \right)|^2$ is the total interference\footnote{Note that $I_k$ can be calculated based on the received message $m_k$.  Specifically, we write $m_k$ in (\ref{exprM}) as $m_k = \frac{J_k'\left(Q_k \right)\tau p_k^{\mathbf{H}}L_{kk}|H_{kk}|^2  }{(1+ I_k + p_k^{\mathbf{H}} L_{kk}|H_{kk}|^2)(1+ I_k)}$. Then, $I_k$ can be easily calculated based on the local knowledge of $J_k \left( q_k \right)$. $p_k^{\mathbf{H}}$ and $L_{kk}|H_{kk}|^2$. } at receiver $k$. $p_k^{up}\left(t_n \right)$ satisfies $\log \left(1+\frac{p_k^{up}\left(t_n \right)L_{kk} |H_{kk}\left(n\right)|^2} {1+I_k\left(n\right)} \right) \tau' = Q\left(t_n \right)$, where $\tau'$ is the slot duration and $Q\left(t_n \right)$ is the QSI at the $n$-th time slot of the $t$-th epoch\footnote{The constraint in (\ref{constraintO}) is equivalent to the requirement  that the transmitter cannot transmit more than the unfinished work left in the queue at the each time slot. Therefore, $p_k^{up}\left(t_n \right)$ is maximum that $p_k^{\mathbf{H}}\left(n+1 \right)$ can take at the $n$-th time slot of the $t$-th epoch.}.
			\item \textbf{Step 4 [Termination]:} If a certain termination condition\footnote{For example, the termination condition can be chosen as $|p_k^{\mathbf{H}}\left(n+1\right)-p_k^{\mathbf{H}}\left(n\right)|<\delta_k$ for some threshold $\delta_k$.} is satisfied, stop. Otherwise,  $n=n+1$ and go to Step 2 of Algorithm \ref{PCSIDPA}.~\hfill\IEEEQED
	\end{itemize}
		\end{Algorithm}
	\end{itemize}
\end{Algorithm}

Fig.~\ref{flow_eg1}  illustrates the above procedure in a flow chart. 
\begin{figure}
  \centering
  \includegraphics[width=5.5in]{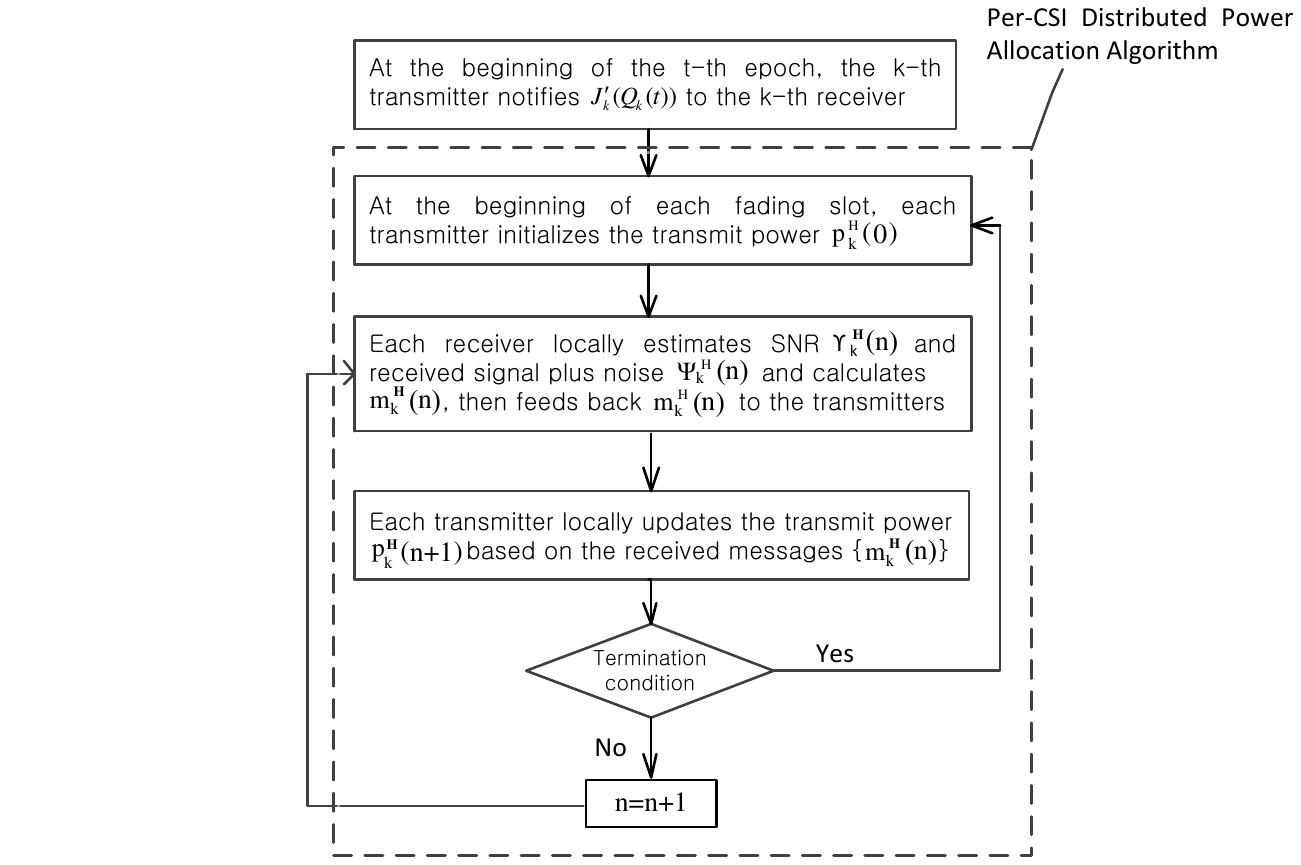}
  \caption{Algorithm flow of the proposed distributive power control algorithm for the interference networks.}
  \label{flow_eg1}
\end{figure}

\begin{Remark}	[Multi-level Water-filling Structure of the Power Action Update]
The power action update in (\ref{controlDPA}) in Algorithm \ref{PCSIDPA} has the \emph{multi-level water-filling} structure where the power is allocated according to the CSI but the water-level is adaptive  to the QSI indirectly via the per-flow fluid value function $J_k\left(Q_k \right)$.~\hfill\IEEEQED
\end{Remark}

\subsection{Simulation Results and Discussions}		\label{simulation}	
In this subsection, we  compare the delay performance gain of the proposed distributed power control scheme in Algorithm \ref{distCAeg1} for the interference networks example with the following three baseline schemes using numerical simulations.
\begin{itemize}
	\item \textbf{Baseline 1 [Orthogonal Transmission]:}  The transmissions among the $K$  Tx-Rx pairs are coordinated using TDMA.  At each  time slot, the Tx-Rx pair with the largest channel gain is select to transmit and the resulting power allocation  is adaptive to the CSI only.
 	\item \textbf{Baseline 2 [CSI Only  Scheme]:} CSI Only scheme solves the problem with the objective function given in (\ref{UtilityEg1}) replacing $J_k'\left(Q_k \right)$  with constant 1. The corresponding power control algorithm can be  obtained by replacing $J_k'\left(Q_k \right)$ with constant 1 in Algorithm \ref{distCAeg1} and the resulting power allocation is adaptive to the  CSI only.
	\item \textbf{Baseline 3 [Queue-Weighted  Throughput-Optimal (QWTO) Scheme\footnote{Baseline 3 is similar to the \emph{Modified Largest Weighted Delay First} algorithm \cite{DBPalg}  but with a modified objective function.}]:} QWTO scheme solves the problem with the objective function given in (\ref{UtilityEg1}) replacing  $J_k'\left(Q_k \right)$  with $Q_k$. The corresponding power control algorithm can be  obtained by replacing $J_k'\left(Q_k \right)$ with $Q_k$ in Algorithm \ref{distCAeg1} and the resulting power allocation is  adaptive to the  CSI and the QSI.
\end{itemize}

In the simulations, we consider a symmetric system with $K$ Tx-Rx pairs in the fast fading environment, where the microscopic fading coefficient and the channel noise are $\mathcal{CN}\left(0,1 \right)$ distributed. The direct channel long term path gain  is $L_{kk}=1$ for all $k \in \mathcal{K}$ and  the cross-channel path gain is $L_{kj}=0.1$ for all $k, j \in \mathcal{K}, k \neq j$ as in \cite{simtopo}. We consider Poisson packet arrival with average  arrival rate $\lambda_k$ (pkts/epoch). The packet size  is exponentially distributed  with mean size equal to  $30$K bits. The decision epoch duration $\tau$ is $5$ms. The total bandwidth is $10$MHz. Furthermore, $\gamma_k$ is the same  and  $\beta_k=1$ for all $k \in \mathcal{K}$.

Fig.~\ref{sim1} illustrates the average delay per pair  versus the average transmit SNR. The average delay of all the schemes decreases as the average transmit SNR increases. It can be observed that  there is significant performance gain of the proposed scheme compared with all the baselines.  It also verifies that the sum of the per-flow fluid value functions is a good approximator to the relative value function.

Fig.~\ref{sim3} illustrates the average delay per pair versus the traffic loading (average data arrival rate $\lambda_k$).  The proposed scheme achieves significant performance gain over all the baselines across a wide range of the input traffic loading. In addition, as  $\lambda_k$ increases, the performance gain of the proposed scheme also increases compared with all the baselines. This verifies Theorem \ref{RelashipVJ} and Lemma \ref{ErrorEg1}. Specifically, it is because as the traffic loading increases, the chance for the queue state at large values increases, which means that $J\left(\mathbf{q}; \mathbf{L} \right)$ becomes a good approximator for $V\left(\mathbf{Q} \right)$ according to Remark \ref{RelashipVJ}. Furthermore, the approximation error between $J\left(\mathbf{q}; \mathbf{L} \right)$ and $\sum_{k=1}^K J\left(q_k \right)$ also decreases according to Remark \ref{link_sim}. Therefore, the per-flow fluid value function approximation in (\ref{approxEg1}) becomes more accurate as $\lambda_k$ increases.

Fig.~\ref{sim4}  illustrates the average delay per pair versus the number of the Tx-Rx pairs.  The average delay of all the schemes increases as the number of the Tx-Rx pairs increases.  This is due to the increasing of the total interference for each Tx-Rx pair. It can  be observed that there is  significant performance gain of the proposed scheme compared with all the baselines across a wide range of the number of the Tx-Rx pairs.

\begin{figure}
  \centering
  \includegraphics[width=4in]{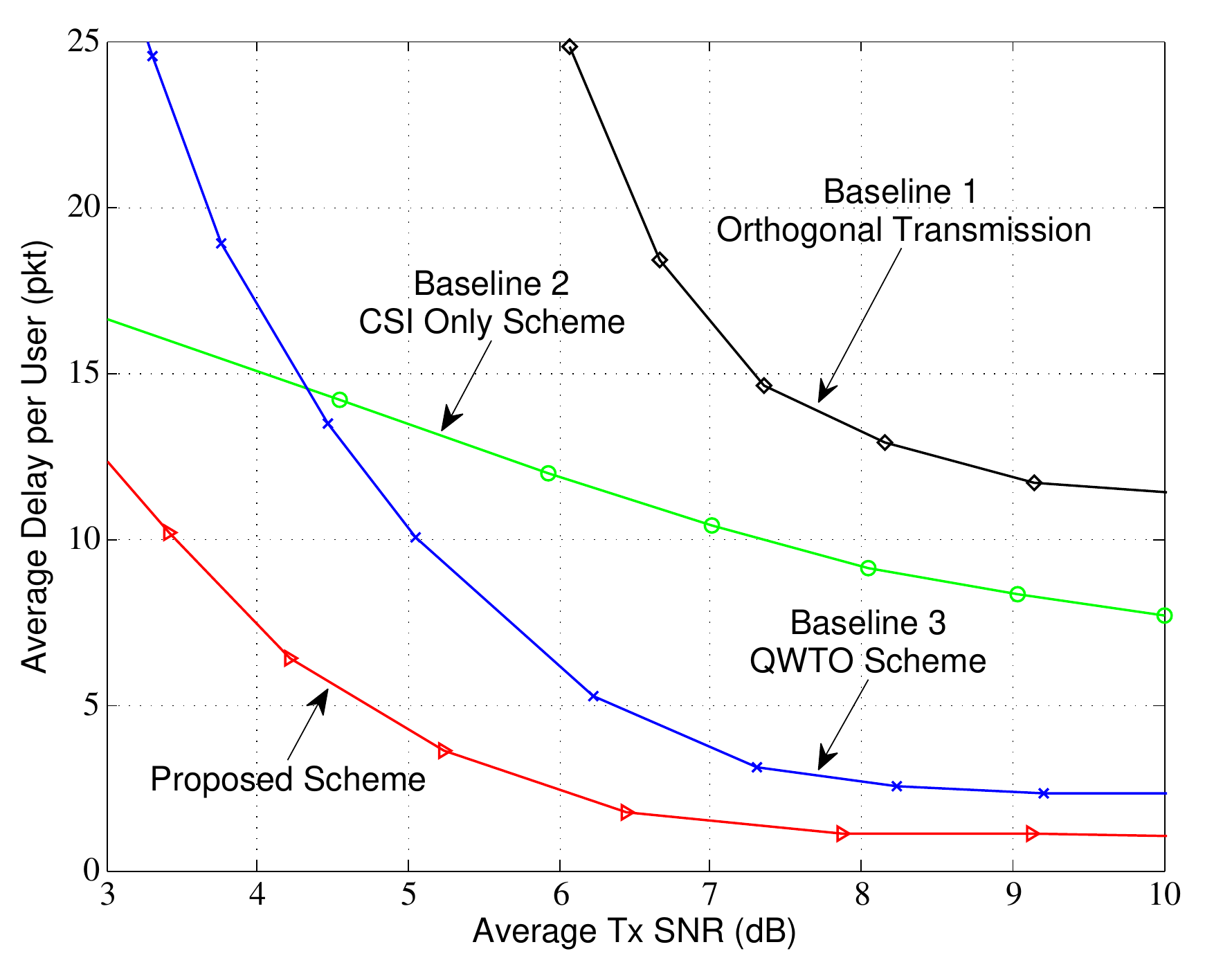}
  \caption{Average delay per pair  versus average transmit SNR. The number of the Tx-Rx pair is $K=5$ and the average data arrival rate is $\lambda_k =1$ pkt/epoch  for all $k \in \mathcal{K}$.}
  \label{sim1}
\end{figure}

\begin{figure}
  \centering
  \includegraphics[width=4in]{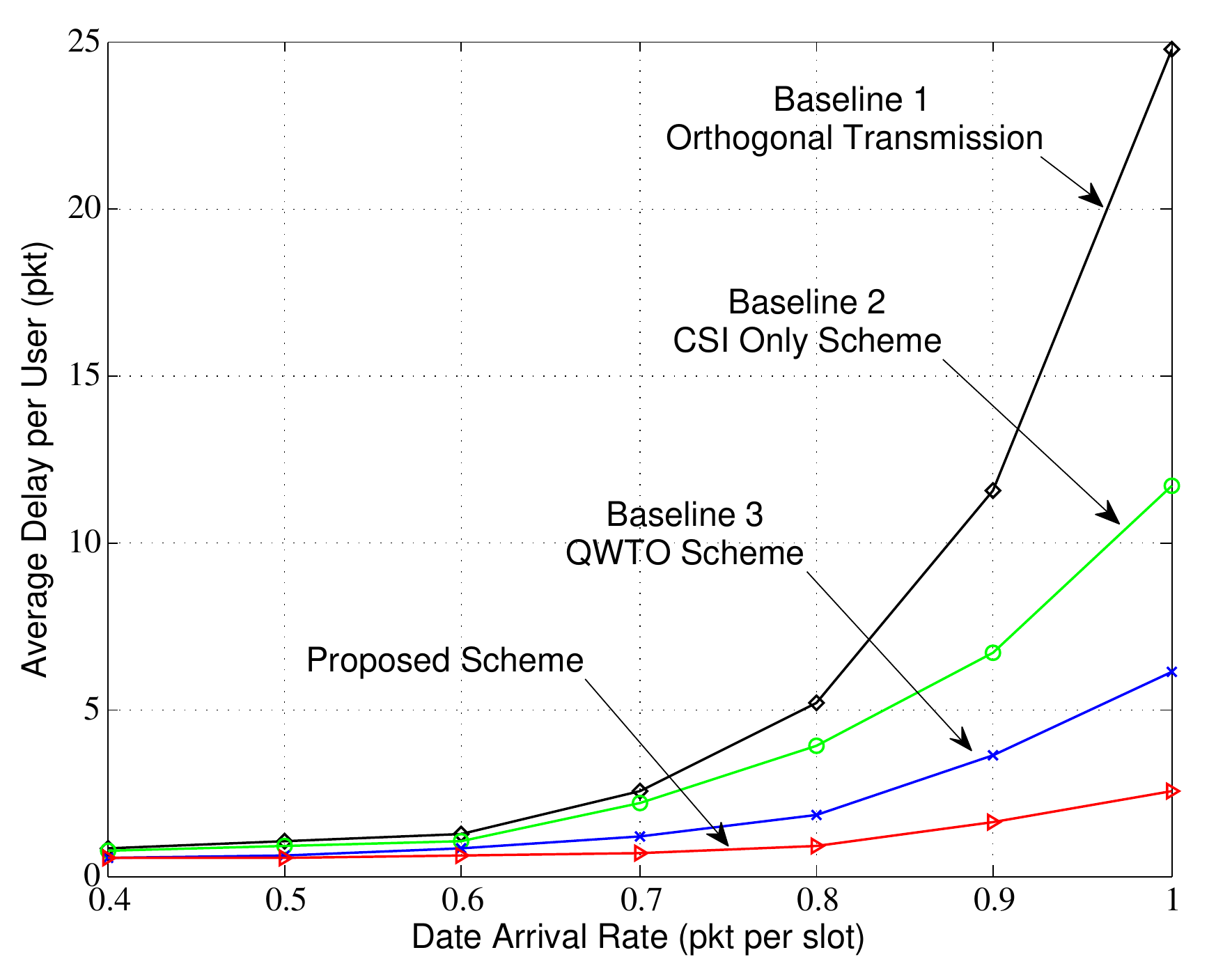}
  \caption{Average delay per pair  versus average data arrival rate at  average transmit SNR $6$ dB. The number of the Tx-Rx pair is $K=5$.}
  \label{sim3}
\end{figure}

\begin{figure}
  \centering
  \includegraphics[width=4in]{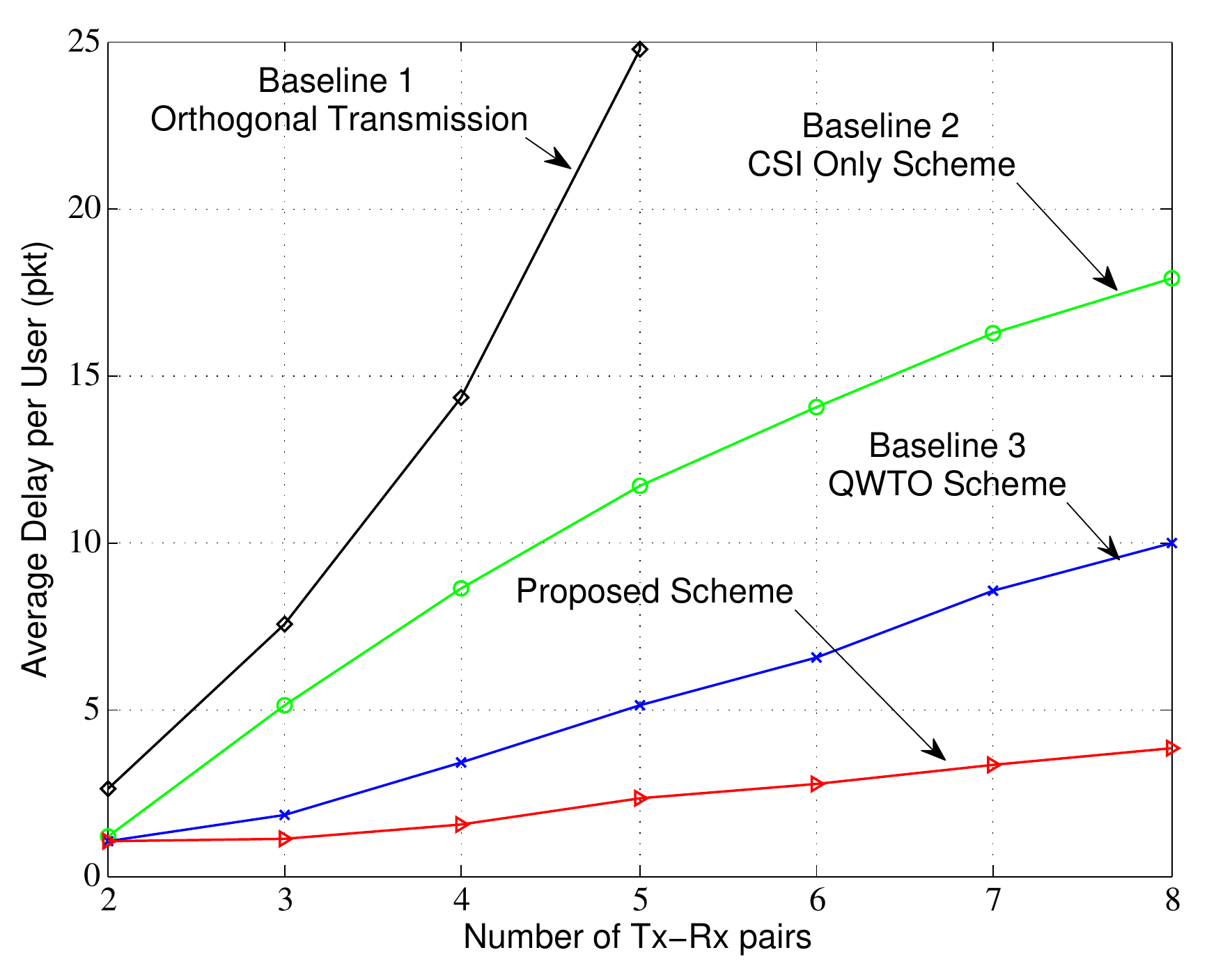}
  \caption{Average delay per pair versus number of Tx-Rx pairs at  average transmit SNR $6$ dB. The average data arrival rate is $\lambda_k=1$ pkt/epoch for all $k \in \mathcal{K}$.}
  \label{sim4}
\end{figure}

\section{Summary}
In this paper, we propose a framework of solving  the infinite horizon average cost problem for the weakly coupled multi-dimensional systems. To reduce the computational complexity,  we first introduce the VCTS and obtain the associated fluid value function to approximate the relative value function of the ODTS.  To further address the low complexity distributed solution requirement and the coupling challenge, we model the weakly coupled system as a perturbation of a decoupled base system. We then use the sum of the per-flow fluid value functions, which are obtained by solving the per-flow HJB equations under each sub-system, to approximate the fluid value function of the VCTS. Finally,  using per-flow fluid value function approximation, we obtain the distributed solution by solving an equivalent deterministic NUM problem. Moreover, we also elaborate on how to use this framework in the interference networks example. It is shown by simulations that the proposed distributed power control algorithm  has much better delay performance than the other three baseline schemes.

\section*{Appendix A: Proof of Corollary \ref{cor1}}	
We first write the state dynamics in (\ref{DiscreteDym}) in the following form: 
\[
\mathbf{x}_k(t+1) = \mathbf{x}_k\left(t\right) +  \left(\mathbf{f}_k\left(\mathbf{u}_k\left(t\right), \mathbf{u}_{-k}\left(t\right),\boldsymbol{\epsilon}_k\right) + \mathbf{z}_k\left(t\right)\right)\Delta
\]
Assume  $V\left( \mathbf{x}\right)$ is of class $\mathcal{C}^1$, we have the following Taylor expansion on $V\left( \mathbf{x}(t+1)\right)$ in (\ref{OrgBel}):
\[	
	\mathbb{E}\left[ V\left( \mathbf{x}(t+1)\right) \big| \mathbf{x}\left(t\right)=\mathbf{x}, \Omega\left(\mathbf{x}\right) \right]=V\left( \mathbf{x}\right)+\nabla_{\mathbf{x}}V \left(\mathbf{x} \right) \left[  \mathbf{f}\left( \Omega\left( \mathbf{x}\right), \boldsymbol{\epsilon}\right)+\overline{\mathbf{z}}  \right]^T \Delta + \mathcal{O}\left( \Delta^2\right)
\]
Hence, the Bellman equation in (\ref{OrgBel}) becomes:
\begin{align}	\notag 
	\theta^\ast \Delta  = \min_{ \Omega\left( \mathbf{x} \right)} \left[ c\left(\mathbf{x}, \Omega\left(\mathbf{x}\right)\right) \Delta +\nabla_{\mathbf{x}}V \left(\mathbf{x} \right) \left[  \mathbf{f}\left( \Omega\left( \mathbf{x}\right), \boldsymbol{\epsilon}\right)+\overline{\mathbf{z}}  \right]^T \Delta + \mathcal{O}\left( \Delta^2\right) \right]
\end{align}
Suppose $\left(\theta^\ast, V^\ast, \mathbf{u}^\ast \right)$ satisfies the Bellman equation in (\ref{OrgBel}), we have
\begin{align}
	-\theta^\ast   +  c\left(\mathbf{x}, \mathbf{u}^\ast \right)  +\nabla_{\mathbf{x}}V^\ast \left(\mathbf{x} \right) \left[  \mathbf{f}\left( \mathbf{u}^\ast , \boldsymbol{\epsilon}\right)+\overline{\mathbf{z}}  \right]^T  + \mathcal{O}\left( \Delta\right)&=0, \quad  	\forall \mathbf{x} \in \boldsymbol{\mathcal{X}} 	\label{fix11}	\\
	\nabla_{\mathbf{u}_k} c_k\left(\mathbf{x}_k, \mathbf{u}_k^\ast\right) + \nabla_{\mathbf{x}}V^\ast \left(\mathbf{x} \right) \left[ \nabla_{\mathbf{u}_k}  \mathbf{f} \left(\mathbf{u}^\ast, \boldsymbol{\epsilon}\right)+\overline{\mathbf{z}}  \right]^T+ \mathcal{O}\left( \Delta\right)&=\mathbf{0}, \quad \forall k\in \mathcal{K}	\label{fix12}
\end{align}
where $\nabla_{\mathbf{u}_k}  \mathbf{f} \left(\mathbf{u}, \boldsymbol{\epsilon}\right)=\left(\nabla_{\mathbf{u}_k} \mathbf{f}_1\left(\mathbf{u}_1, \mathbf{u}_{-1},\boldsymbol{\epsilon}_1 \right),\dots,\nabla_{\mathbf{u}_k} \mathbf{f}_K\left(\mathbf{u}_K, \mathbf{u}_{-K},\boldsymbol{\epsilon}_K \right)  \right)$. (\ref{fix11}) and (\ref{fix12}) can be expressed as a fixed point equation in $\left(\theta^\ast, V^\ast, \mathbf{u}^\ast \right)$:
\begin{align}
	\mathbf{F}\left( \theta^\ast, V^\ast, \mathbf{u}^\ast \right)= \mathbf{0}
\end{align}
Suppose $\big(\tilde{\theta}^\ast, V, \mathbf{u}\big)$ is the solution of the approximate Bellman equation in (\ref{simbelman}), we have
\begin{align}
	&-\tilde{\theta}^\ast   +  c\left(\mathbf{x}, \mathbf{u} \right)  +\nabla_{\mathbf{x}}V \left(\mathbf{x} \right) \left[  \mathbf{f}\left( \mathbf{u} , \boldsymbol{\epsilon}\right)+\overline{\mathbf{z}}  \right]^T=0, \quad  	\forall \mathbf{x} \in \boldsymbol{\mathcal{X}} 	\label{fix111}	\\
	&\nabla_{\mathbf{u}_k} c_k\left(\mathbf{x}_k, \mathbf{u}_k\right) + \nabla_{\mathbf{x}}V \left(\mathbf{x} \right) \left[ \nabla_{\mathbf{u}_k}  \mathbf{f} \left(\mathbf{u}, \boldsymbol{\epsilon}\right)+\overline{\mathbf{z}}  \right]^T=0, \quad \forall k\in \mathcal{K}	\label{fix121}
\end{align}
Comparing (\ref{fix111}) and (\ref{fix121}) with (\ref{fix11}) and (\ref{fix12}), $\big(\tilde{\theta}^\ast, V, \mathbf{u}\big)$ can be visualized as a solution of the  perturbed fixed point equation:
\begin{align}	\label{comp111a}
	\mathbf{F}\big( \tilde{\theta}^\ast, V, \mathbf{u} \big)=\mathcal{O}(\Delta)
\end{align}
Hence, we have $\theta^\ast=\tilde{\theta}^\ast+\delta_{\theta}$, $V^\ast(\mathbf{Q})=V(\mathbf{Q})+\delta_V$ and $\mathbf{u}^\ast=\mathbf{u}+\boldsymbol{\delta}_{\mathbf{u}}$ and $(\delta_{\theta}, \delta_V, \boldsymbol{\delta}_{\mathbf{u}})$ satisfies
\begin{align}	\label{interderr}
	\mathrm{d}\mathbf{F}\left( \theta^\ast, V^\ast, \mathbf{u}^\ast \right)=\frac{\partial \mathbf{F}\left( \theta^\ast, V^\ast, \mathbf{u}^\ast \right) }{\partial \theta }\delta_{\theta}+\frac{\partial \mathbf{F}\left( \theta^\ast, V^\ast, \mathbf{u}^\ast \right) }{\partial V }\delta_V+ \nabla_{\mathbf{u}}\mathbf{F}\left( \theta^\ast, V^\ast, \mathbf{u}^\ast \right) \boldsymbol{\delta}_{\mathbf{u}}
\end{align}
Comparing (\ref{comp111a}) with (\ref{fix111}), we have $\mathrm{d}\mathbf{F}\left( \theta^\ast, V^\ast, \mathbf{u}^\ast \right)=\mathcal{O}(\Delta)$. Hence, we have
\[
	\frac{\partial \mathbf{F}\left( \theta^\ast, V^\ast, \mathbf{u}^\ast \right) }{\partial \theta }\delta_{\theta}+\frac{\partial \mathbf{F}\left( \theta^\ast, V^\ast, \mathbf{u}^\ast \right) }{\partial V }\delta_V+ \nabla_{\mathbf{u}}\mathbf{F}\left( \theta^\ast, V^\ast, \mathbf{u}^\ast \right) \boldsymbol{\delta}_{\mathbf{u}}=\mathcal{O}(\Delta)
\]
Therefore, $(\delta_{\theta}, \delta_V, \boldsymbol{\delta}_{\mathbf{u}})=\mathcal{O}(\Delta)$.

\section*{Appendix B: Proof of Corollary \ref{orderoptlemma}}	
First, It can be observed that if ($c^{\infty}, \{ J\left(\mathbf{x} \right) \}$) satisfies the HJB equation in (\ref{cenHJB}), then it also satisfies the approximate HJB equation in (\ref{simbelman}). Second, if $J\left(\mathbf{x} \right)=\mathcal{O}\left(|\mathbf{x}|^n \right)$ is polynomial growth at order $n$, we have $\mathbb{E}^{\Omega^v}\left[J\left(\mathbf{x} \right) \right] < \infty$ for any admissible policy $\Omega^v$. Hence, $J\left(\mathbf{x} \right)$ satisfies the transversality condition of the approximate Bellman equation. Therefore, we have $\theta^\ast=\tilde{\theta}^\ast+\mathcal{O}\left(\Delta \right)=c^\infty+\mathcal{O}\left(\Delta \right)$.

\section*{Appendix A: Proof of Theorem \ref{RelashipVJ}}	
In the following proof, we first establish  three important equalities in (\ref{impotequ1}), (\ref{impoEqu}) and (\ref{impoEqu2}). We then prove Theorem \ref{RelashipVJ}  based on the three equalities.

First, we establish the following equality:
\begin{align}	
	\hspace{-0.28cm}\frac{1}{n^2}\mathbb{E} \left[ V\left(\mathbf{x}^{\ast}\left(\lfloor nT \rfloor;  n \mathbf{x}^n  \right)\right) \right]=\frac{1}{n^2} \Big[V\left( n \mathbf{x}^n \right)-\sum_{i=0}^{\lfloor nT \rfloor-1} \mathbb{E} \left[ \widetilde{c}\left(\mathbf{x}^{\ast} \left( i;  n \mathbf{x}^n   \right),\mathbf{u}^{\ast}\left(\mathbf{x}^{\ast} \left( i;  n \mathbf{x}^n  \right) \right) \right) \right]	\Big] + \mathcal{O}\left( \frac{1}{n} \right)	\label{impotequ1}
\end{align}
Here we define $\mathbf{x} \left(t;\mathbf{x}_0 \right)$ and $\overline{\mathbf{x}}(t;\mathbf{x}_0)$ to be the system states at time $t$ which evolve according to the dynamics in (\ref{DiscreteDym}) and (\ref{VCTS_dyn}), respectively with initial state  $\mathbf{x}_0$.  Let $N$ be the dimension of $\mathbf{x}$ and  let $\mathbf{u}^{\ast}$ be the optimal policy solving Problem \ref{IHAC_MDP}. Then, we write the  Bellman equation in (\ref{OrgBel}) in a vector form as: $\theta \mathbf{e}+ \mathbf{V}= \mathbf{c} \left(\mathbf{u}^{\ast}\right)+ \mathbf{P} \left( \mathbf{u}^{\ast} \right) \mathbf{V}$ where  $\mathbf{e}$ is a $N \times 1$ vector with each element being $1$, $\mathbf{P} \left( \mathbf{u}^{\ast} \right) $ is  a $N \times N$ transition matrix, $\mathbf{c} \left(\mathbf{u}^{\ast}\right)$ and  $\mathbf{V}$ are $N \times 1$ cost and value function vectors. We iterate the vector form Bellman equation as follows: 
\begin{align}
	\mathbf{P} \left( \mathbf{u}^{\ast} \right) \mathbf{V} &= \mathbf{V}- \left( \mathbf{c} \left(\mathbf{u}^{\ast}\right) -\theta \mathbf{e} \right)\notag \\
	\mathbf{P}^2 \left( \mathbf{u}^{\ast} \right) \mathbf{V} &= \mathbf{P} \left( \mathbf{u}^{\ast} \right) \mathbf{V}- \mathbf{P} \left( \mathbf{u}^{\ast} \right) \left( \mathbf{c} \left(\mathbf{u}^{\ast}\right) -\theta \mathbf{e} \right)		\notag \\
	& =   \mathbf{V}- \left( \mathbf{c} \left(\mathbf{u}^{\ast}\right) -\theta \mathbf{e} \right) - \mathbf{P} \left( \mathbf{u}^{\ast} \right) \left( \mathbf{c} \left(\mathbf{u}^{\ast}\right) -\theta \mathbf{e} \right)	\notag \\
	& \vdots	\notag \\
	\mathbf{P}^m \left(\mathbf{u}^{\ast} \right) \mathbf{V} &= \mathbf{V}-\sum_{i=0}^{m-1}\mathbf{P}^i \left( \mathbf{u}^{\ast} \right) \left(\mathbf{c}\left(\mathbf{u}^{\ast} \right)-\theta \mathbf{e} \right)	\label{iterativeequ}
\end{align}
Considering the row corresponding to a given system state $\mathbf{x}_0$, we have
\begin{equation}		\label{impotequu}
	  \mathbb{E} \left[ V\left(\mathbf{x}^{\ast}\left( m; \mathbf{x}_0\right)\right) \right]=V\left(\mathbf{x}_0\right)-\sum_{i=0}^{m-1} \left( \mathbb{E} \left[ c\left(\mathbf{x}^{\ast} \left( i; \mathbf{x}_0 \right),\mathbf{u}^{\ast}\left(\mathbf{x}^{\ast}\left(i; \mathbf{x}_0\right)  \right)\right) \right]-\theta \right)
\end{equation}
where $\mathbf{x}^{\ast}\left(t; \mathbf{x}_0 \right)$ is the system state under optimal policy $\mathbf{u}^{\ast} $ with initial state $\mathbf{x}_0$. Dividing $n^2$ on both size of (\ref{impotequu}), choosing $m= \lfloor nT \rfloor$ and $\mathbf{x}_0= n \mathbf{x}^n$, we have 
\begin{align}	
	\frac{1}{n^2}\mathbb{E} \left[ V\left(\mathbf{x}^{\ast}\left( \lfloor nT \rfloor; n \mathbf{x}^n \right)\right) \right]=\frac{1}{n^2} \Big[V\left(n \mathbf{x}^n \right)-\sum_{i=0}^{\lfloor nT \rfloor-1} \mathbb{E} \left[ c\left(\mathbf{x}^{\ast} \left( i; n \mathbf{x}^n \right),\mathbf{u}^{\ast}\left(\mathbf{x}^{\ast} \left( i; n \mathbf{x}^n \right) \right) \right) \right]	\Big]+\mathcal{O}\left( \frac{1}{n} \right)	\notag \\
	 =\frac{1}{n^2} \Big[V\left(n \mathbf{x}^n \right)-\sum_{i=0}^{\lfloor nT \rfloor-1} \mathbb{E} \left[\widetilde{ c}\left(\mathbf{x}^{\ast} \left( i; n \mathbf{x}^n \right),\mathbf{u}^{\ast}\left(\mathbf{x}^{\ast} \left( i; n \mathbf{x}^n \right) \right) \right) \right]\Big]+\mathcal{O}\left( \frac{1}{n} \right)	
\end{align}
where the first equality is due to $\frac{1}{n^2}\sum_{i=0}^{\lfloor nT \rfloor-1} \theta=\mathcal{O}\left(\frac{1}{n} \right) $ last equality is due to $\widetilde{c}\left(\overline{\mathbf{x}}, \mathbf{u}\right)=c\left(\overline{\mathbf{x}}, \mathbf{u}\right)-c^{\infty}$ and $\frac{1}{n^2}\sum_{i=0}^{\lfloor nT \rfloor-1} c^{\infty}=\left(\frac{1}{n} \right)  $.

Second, we establish the following equality which holds under any unichain stationary policy $\mathbf{u}$:  
\begin{align}	\label{impoEqu}
	& \frac{1}{n^2} \sum_{i=0}^{n T^n} \widetilde{c} \left( \mathbf{x} \left( i; n \mathbf{x}^n \right), \mathbf{u}\left(\mathbf{x} \left( i; n \mathbf{x}^n \right) \right)  \right)	\notag \\
	=&\int_0^{T^n} \widetilde{c}\left(  \overline{\mathbf{x}}^n \left(t;  \mathbf{x}^n\right), \mathbf{u}\left( \overline{\mathbf{x}}^n \left(t;  \mathbf{x}^n\right)  \right)\right) \mathrm{d} t 	 -      \int_0^{T^n} \sum_{k=1}^K \widetilde{g}_k\left( \mathbf{u}_k\left( \overline{\mathbf{x}}^n \left(t;  \mathbf{x}^n\right)  \right)\right) \mathrm{d} t 	+ \mathcal{O}\left( \frac{1}{n}\right)
\end{align} 
where $T^n=\frac{1}{n} \lfloor nT \rfloor$. Here we define a scaled process w.r.t. $\mathbf{x}\left(t\right)$ as 
\begin{equation}
	\overline{\mathbf{x}}^n(t;\mathbf{x}_0^n)=\frac{1}{n} \mathbf{x}(nt; n\mathbf{x}_0^n), \quad \text{where \ }  \mathbf{x}_0^n=  \frac{1}{n} \lfloor n\mathbf{x}_0 \rfloor	\label{scaledprocess}
\end{equation}
$\lfloor \mathbf{x} \rfloor$ is the floor function that maps each element of $\mathbf{x}$ to the integer not greater than it. According to \emph{Prop.3.2.3} of \cite{CompApp}, we have $\lim_{n \rightarrow \infty } \overline{\mathbf{x}}^n(t;\mathbf{x}_0^n) = \overline{\mathbf{x}}^{vt}(t;\mathbf{x}_0)$, where $\overline{\mathbf{x}}^{vt}$ is some fluid process w.r.t. $\mathbf{x}$.  In the fluid control problem in Problem \ref{fluid problem1}, for each initial state $\mathbf{x}_0 \in \boldsymbol{\mathcal{X}}$, there is a finite time horizon $T$ such that $\overline{\mathbf{x}}^{vt}(t; \mathbf{x}_0)=\mathbf{0}$ for all $t \geq T$ \cite{StateAggr_2}.  Furthermore, we can write the above convergence result based on the functional law of large numbers as $\overline{\mathbf{x}}^n(t;\mathbf{x}_0^n) = \overline{\mathbf{x}}^{vt}(t;\mathbf{x}_0) + \mathcal{O}\left(\sqrt{\frac{\log\log n}{n}} \right)$. Before proving (\ref{impoEqu}), we show the following lemma. 
\begin{Lemma}	\label{twoequ}
	For the continuously differentiable function $g_k$ in the cost function in (\ref{costfunc}) with $\left| g_k \left(\mathbf{u}_k \right) \right| < \infty$ $\left( \forall \mathbf{u}_k \in \mathcal{U}_k\right)$, there exist a finite constant $C$, such that for all $k$, $t$, we have
\begin{align}		
		\left|\int_t^{t+1} \widetilde{g}_k\left(\mathbf{u}_k\left(\widetilde{\mathbf{x}}\left(t; \mathbf{x}_0\right)\right)\right) \mathrm{d}s   -   \frac{1}{2} \left[ \widetilde{g}_k\left(\mathbf{u}_k\left(\mathbf{x}\left(t;\mathbf{x}_0\right)\right)\right) + \widetilde{g}_k\left(\mathbf{u}_k\left(\mathbf{x}\left(t+1;\mathbf{x}_0 \right)\right)\right)  \right]   \right| \leq C, \quad \forall \mathbf{x}_0 \in \boldsymbol{\mathcal{X}}   \label{propG1} 
\end{align}	
where $\widetilde{g}_k\left(\mathbf{u}_k\right) \triangleq g_k\left(\mathbf{u}_k\right) -\frac{c^{\infty}}{K}$,  $\widetilde{\mathbf{x}} \left(t;\mathbf{x}_0 \right)$ is a piecewise linear function and  satisfies   $\widetilde{\mathbf{x}} \left(t;\mathbf{x}_0 \right) = \mathbf{x} \left(t;\mathbf{x}_0 \right)$ for all $t \in \mathbb{Z}^{\ast}$.  Furthermore, for a given finite positive real number $\kappa$,  we have
\begin{align}
	\left|\widetilde{g}_k\left(\mathbf{u}_k\left(\kappa \mathbf{x}_0 \right)\right) - \kappa \widetilde{g}_k\left(\mathbf{u}_k\left(\mathbf{x}_0\right)\right) \right| = \mathcal{O}\left(\kappa\right), \quad \forall  \mathbf{x}_0 \in \boldsymbol{\mathcal{X}}	\label{propG2}
\end{align}~\hfill\IEEEQED
\end{Lemma}
Since $\mathbf{u}_k \in \mathcal{U}_k$ and  $\left| \widetilde{g}_k \left(\mathbf{u}_k \right) \right| < \infty$ $\left( \forall \mathbf{u}_k \in \mathcal{U}_k\right)$, the two inequalities in (\ref{propG1}) and (\ref{propG2}) can be easily verified  using the compactness property of the sub-system action space $\mathcal{U}_k$.  We establish the  proof of  (\ref{impoEqu}) in the following two steps. 
\begin{enumerate} [1)]
	\item  First, we prove the following equality:
	\begin{align}	\label{step111}
	 \frac{1}{n^2}\int_0^{nT^n} \widetilde{c}\left(\widetilde{\mathbf{x}}\left(t;n\mathbf{x}^n\right), \mathbf{u}\left( \widetilde{\mathbf{x}}\left(t;n\mathbf{x}^n\right)\right) \right) \mathrm{d}t= \frac{1}{n^2} \sum_{i=0}^{nT^n} \widetilde{c} \left( \mathbf{x} \left(i;n\mathbf{x}^n\right), \mathbf{u}\left(\mathbf{x} \left(i;n\mathbf{x}^n\right) \right) \right)+\mathcal{O}\left(\frac{1}{n} \right)
	 \end{align}
We calculate $\int_0^{nT^n} \widetilde{c}\left(\widetilde{\mathbf{x}}\left(t; n\mathbf{x}^n\right), \mathbf{u}\left( \widetilde{\mathbf{x}}\left(t; n\mathbf{x}^n\right)\right) \right) \mathrm{d}t$ and $\sum_{i=0}^{n T^n} \widetilde{c} \left( \mathbf{x} \left(i ; n\mathbf{x}^n\right), \mathbf{u}\left(\mathbf{x}\left(i ; n\mathbf{x}^n\right) \right) \right)$ in (\ref{step111}) as follows
	\begin{align}	
		&\int_0^{nT^n} \widetilde{c}\left(\widetilde{\mathbf{x}}\left(t; n\mathbf{x}^n\right), \mathbf{u}\left( \widetilde{\mathbf{x}}\left(t; n\mathbf{x}^n\right)\right) \right) \mathrm{d}t	\notag \\
		=& \sum_{i=0}^{nT^n-1} \sum_{k=1}^K\frac{1}{2} \left[ \alpha_k \|\mathbf{x}_k(t; n\mathbf{x}^n)  \|_{\mathbf{v}_k,1}+\alpha_k \|\mathbf{x}_k(t+1; n\mathbf{x}^n)  \|_{\mathbf{v}_k,1}\right] + \int_0^{nT^n} \sum_{k=1}^K   \widetilde{g}_k\left(\mathbf{u}_k\left(\widetilde{\mathbf{x}}\left(t; n\mathbf{x}^n\right)\right)\right) \mathrm{d}t	\label{impp111}  \\
		&\sum_{i=0}^{n T^n} \widetilde{c} \left( \mathbf{x} \left(i ; n\mathbf{x}^n\right), \mathbf{u}\left(\mathbf{x}\left(i ; n\mathbf{x}^n\right) \right) \right)	\notag \\
		=& \sum_{i=0}^{nT^n-1} \sum_{k=1}^K\frac{1}{2} \left[ \alpha_k \|\mathbf{x}_k(t; n\mathbf{x}^n)  \|_{\mathbf{v}_k,1}+\alpha_k \|\mathbf{x}_k(t+1; n\mathbf{x}^n)  \|_{\mathbf{v}_k,1} \right. \notag \\
		& \hspace{4cm} \left. +\widetilde{g}_k\left(\mathbf{u}_k\left(\mathbf{x}\left(t;n\mathbf{x}^n\right)\right)\right)+\widetilde{g}_k\left(\mathbf{u}_k\left(\mathbf{x}\left(t+1;n\mathbf{x}^n\right)\right)\right)	  \right]+ n^2 \mathcal{E}^n	\label{impp222}
	\end{align}
where we denote  $\mathcal{E}^n=\frac{1}{2n^2}\left[ \widetilde{c}\left(\mathbf{x}\left(0;n\mathbf{x}^n\right), \mathbf{u}\left( \mathbf{x}\left(0;n\mathbf{x}^n\right) \right)\right)+ \widetilde{c}\left(\mathbf{x}\left(nT^n;n\mathbf{x}^n\right), \mathbf{u} \left( \mathbf{x}\left(nT^n;n\mathbf{x}^n\right)\right)\right) \right]$. Then based on (\ref{impp111}) and (\ref{impp222}), we have
	\begin{align}
		         &\frac{\left|(\ref{impp111})- (\ref{impp222}) \right|}{n^2}=\left| \frac{1}{n^2}\int_0^{nT^n} \widetilde{c}\left(\widetilde{\mathbf{x}}\left(t; n\mathbf{x}^n\right), \mathbf{u}\left( \widetilde{\mathbf{x}}\left(t; n\mathbf{x}^n\right)\right) \right) \mathrm{d}t- \frac{1}{n^2} \sum_{i=0}^{n T^n}\widetilde{c} \left( \mathbf{x} \left(i ; n\mathbf{x}^n\right), \mathbf{u}\left(\mathbf{x}\left(i ; n\mathbf{x}^n\right) \right) \right)\right| \notag \\
		= & \frac{1}{n^2} \left| \int_0^{nT^n} \sum_{k=1}^K  \widetilde{g}_k\left(\mathbf{u}_k\left(\widetilde{\mathbf{x}}\left(t; n\mathbf{x}^n\right)\right)\right)	 \right. \left. - \sum_{i=0}^{nT^n-1} \sum_{k=1}^K\frac{1}{2} \left[\widetilde{g}_k\left(\mathbf{u}_k\left(\mathbf{x}\left(t;n\mathbf{x}^n\right)\right)\right)+\widetilde{g}_k\left(\mathbf{u}_k\left(\mathbf{x}\left(t+1;n\mathbf{x}^n\right)\right)\right)  \right]+ n^2 \mathcal{E}^n\right|  \notag \\
		  \overset{(a)}= &  \mathcal{O}\left(\frac{1}{n}\right) + \left|\mathcal{E}^n\right|		\label{firsteq111}
	\end{align}
	where $(a)$ is due to the triangle inequality and (\ref{propG1}). Next, we prove $\mathcal{E}^n  =\mathcal{O}\left(\frac{1}{n}\right) $ as follows
	\begin{align}
		&  \mathcal{E}^n  =\frac{1}{2n^2} \sum_{k=1}^K \left[ \left( \alpha_k \|\mathbf{x}_k\left(0;n\mathbf{x}^n\right)  \|_{\mathbf{v}_k,1} +  \widetilde{g}_k \left( \mathbf{u}_k\left(\mathbf{x}\left(0;n\mathbf{x}^n\right) \right)  \right) \right)  \right. \notag \\
		& \left. + \left(  \alpha_k \|\mathbf{x}_k\left(nT^n;n\mathbf{x}^n\right)  \|_{\mathbf{v}_k,1} +  \widetilde{g}_k \left( \mathbf{u}_k\left(\mathbf{x}\left(nT^n;n\mathbf{x}^n\right) \right)  \right) \right)  \right]\overset{(b)} =  \frac{\mathcal{O}\left(n \right)}{2n^2} = \mathcal{O}\left(\frac{1}{n} \right)	\label{firsteq222}
	\end{align}
where $(b)$ is due to $\| \mathbf{x}_k\left(0;n\mathbf{x}^n\right) \|_{\mathbf{v}_k,1}=\|n\mathbf{x}_k^n \|_{\mathbf{v}_k,1}=\mathcal{O}\left(n \right)$, $\|\mathbf{x}_k\left(nT^n;n\mathbf{x}^n\right)  \|_{\mathbf{v}_k,1}=\|n\mathbf{x}_k^n + \sum_{i=1}^{nT^n-1}  \left(\mathbf{f}_k\left(\mathbf{u}_k(i), \mathbf{u}_{-k}(i),\boldsymbol{\epsilon}\right) + \mathbf{z}_k(i)  \right)   \|_{\mathbf{v}_k,1} = \mathcal{O}\left(n \right)$, and $\\ \lim_{n \rightarrow \infty} \frac{1}{2 n^2} \left[g_k \left( \mathbf{u}_k\left(\mathbf{x}\left(0;n\mathbf{x}^n\right) \right)  \right) \right.  \left.+  g_k \left( \mathbf{u}_k\left(\mathbf{x}\left(nT^n;n\mathbf{x}^n\right) \right)  \right)  \right]=0$.  

Combining (\ref{firsteq111}) and (\ref{firsteq222}), we have (\ref{step111}).
	
	\item Second, we prove the following equality:
		\begin{align}	\label{step222}
			& \frac{1}{n^2} \int_0^{nT^n} \widetilde{c}\left(\widetilde{\mathbf{x}}\left(t; n \mathbf{x}^n\right), \mathbf{u}\left(\widetilde{\mathbf{x}} \left(t; n \mathbf{x}^n\right) \right)\right) \mathrm{d}t \notag \\
			=& \int_0^{T^n} \widetilde{c}\left(  \overline{\mathbf{x}}^n \left(t;  \mathbf{x}^n\right), \mathbf{u}\left( \overline{\mathbf{x}}^n \left(t;  \mathbf{x}^n\right)  \right)\right) \mathrm{d} t  -  \int_0^{T^n} \sum_{k=1}^K \widetilde{g}_k\left( \mathbf{u}_k\left( \overline{\mathbf{x}}^n \left(t;  \mathbf{x}^n\right)  \right)\right) \mathrm{d} t + \mathcal{O}\left( \frac{1}{n}\right)
		\end{align}
	We have
	\begin{align}
		 &\frac{1}{n^2} \int_0^{nT^n} \widetilde{c}\left(\widetilde{\mathbf{x}}\left(t; n \mathbf{x}^n\right), \mathbf{u}\left(\widetilde{\mathbf{x}} \left(t; n \mathbf{x}^n\right) \right)\right) \mathrm{d}t \overset{(d)}= \frac{1}{n} \int_0^{T^n} \widetilde{c}\left(\widetilde{\mathbf{x}}\left(nt; n \mathbf{x}^n\right), \mathbf{u}\left( \widetilde{ \mathbf{x}} \left(nt; n \mathbf{x}^n\right)  \right)\right) \mathrm{d} t \notag \\	
		= & \frac{1}{n} \int_0^{T^n} \widetilde{c}\left( n \overline{\mathbf{x}}^n \left(t;  \mathbf{x}^n\right), \mathbf{u}\left( n \overline{\mathbf{x}}^n \left(t;  \mathbf{x}^n\right) \right)\right) \mathrm{d} t	\notag \\
		\overset{(e)}=&  \int_0^{T^n} \widetilde{c}\left(  \overline{\mathbf{x}}^n \left(t;  \mathbf{x}^n\right), \mathbf{u}\left( \overline{\mathbf{x}}^n \left(t;  \mathbf{x}^n\right)  \right)\right) \mathrm{d} t - \int_0^{T^n} \sum_{k=1}^K \widetilde{g}_k\left( \mathbf{u}_k\left( \overline{\mathbf{x}}^n \left(t;  \mathbf{x}^n\right)  \right)\right) \mathrm{d} t+ \mathcal{O}\left( \frac{1}{n}\right)	\label{th1step2}
	\end{align}	
where $(d)$ is due to the change of variable from $t$ to $nt$ and $(e)$ is due to $\frac{1}{n} \int_0^{T^n} \sum_{k=1}^K\widetilde{g}_k\left( \mathbf{u}\left( n \overline{\mathbf{x}}^n \left(t;  \mathbf{x}^n\right) \right)\right) \mathrm{d} t \\=\mathcal{O}\left( \frac{1}{n}\right)$. This proves (\ref{step222}). 
\end{enumerate}
Combining (\ref{step111}) and (\ref{step222}), we can prove (\ref{impoEqu}).

Third, we establish the following equality:
\begin{align}	\label{impoEqu2}
	&\int_0^{T} \widetilde{c}\left(\overline{\mathbf{x}}^{\ast} \left(t; \mathbf{x} \right), \overline{\mathbf{u}}^{\ast} \left(\overline{\mathbf{x}}^{\ast} \left(t; \mathbf{x} \right) \right) \right)\mathrm{d} t -  \lim_{n \rightarrow \infty}\int_0^{T} \sum_{k=1}^K \widetilde{g}_k\left(\overline{\mathbf{u}}_k^{\ast} \left(\overline{\mathbf{x}}^{\ast} \left(t; \mathbf{x} \right) \right) \right) \mathrm{d} t 		\notag \\
	=&\frac{1}{n^2}\int_0^{T} \widetilde{c}\big(\overline{\mathbf{x}}^{ \dagger} \left(t; n\mathbf{x} \right),  \overline{\mathbf{u}}^{\dagger}\big(\overline{\mathbf{x}}^{\dagger} \left(t; n\mathbf{x} \right)\big) \big)\mathrm{d} t  + \mathcal{O}\left(\sqrt{\frac{\log \log n}{n}} \right)
\end{align} 
where $\overline{\mathbf{u}}^{\ast}$ is the optimal control trajectory solving the fluid control problem when the initial state is $\mathbf{x}$ with corresponding state trajectory $\overline{\mathbf{x}}^{\ast} $, while $\overline{\mathbf{u}}^{\dagger}$  is the optimal control trajectory when the initial state is $n\mathbf{x}$ with corresponding state trajectory $\overline{\mathbf{x}}^{\dagger} $. We define a scaled process w.r.t. $\overline{\mathbf{x}} \left(t\right)$ as follows:
\begin{align}
	\overline{\mathbf{x}}^{(n)}(t; \overline{\mathbf{x}}_0)=\frac{1}{n} \overline{\mathbf{x}}(nt; n\overline{\mathbf{x}}_0)
\end{align}
We establish the  proof of  (\ref{impoEqu2}) in the following two steps:
\begin{enumerate} [1)]
\item  First, we prove the following inequality:
	\begin{align}	\label{step1111}
	& \int_0^{T} \widetilde{c}\left(\overline{\mathbf{x}}^{\ast} \left(t; \mathbf{x} \right), \overline{\mathbf{u}}^{\ast} \left(\overline{\mathbf{x}}^{\ast} \left(t; \mathbf{x} \right) \right) \right)\mathrm{d} t - \int_0^{T} \sum_{k=1}^K \widetilde{g}_k\left(\overline{\mathbf{u}}_k^{\ast} \left(\overline{\mathbf{x}}^{\ast} \left(t; \mathbf{x} \right) \right) \right) \mathrm{d} t 	\notag \\
	 \leq & \frac{1}{n^2}\int_0^{T} \widetilde{c}\big(\overline{\mathbf{x}}^{ \dagger} \left(t; n\mathbf{x} \right), \overline{\mathbf{u}}^{\dagger}\big(\overline{\mathbf{x}}^{\dagger} \left(t; n\mathbf{x} \right) \big)  \big)\mathrm{d} t+ \mathcal{O}\left(\sqrt{\frac{\log \log n}{n}} \right)
	\end{align}
	We have
	\begin{align}
		& \int_0^{T} \widetilde{c}\left(\overline{\mathbf{x}}^{\ast} \left(t; \mathbf{x} \right), \overline{\mathbf{u}}^{\ast} \left(\overline{\mathbf{x}}^{\ast} \left(t; \mathbf{x} \right) \right) \right)\mathrm{d} t -  \sum_{k=1}^K \widetilde{g}_k\left(\overline{\mathbf{u}}_k^{\ast} \left(\overline{\mathbf{x}}^{\ast} \left(t; \mathbf{x} \right) \right) \right) \mathrm{d} t 	\notag \\		
		\overset{(f)}\leq & \int_0^{T} \widetilde{c}\big(\overline{\mathbf{x}}^{\dagger} \left(t; \mathbf{x} \right), \overline{\mathbf{u}}^{\dagger} \big(\overline{\mathbf{x}}^{\dagger} \left(t; \mathbf{x} \right) \big) \big)\mathrm{d} t -  \int_0^{T}\sum_{k=1}^K \widetilde{g}_k\big(\overline{\mathbf{u}}_k^{\dagger} \big(\overline{\mathbf{x}}^{\dagger} \left(t; \mathbf{x} \right) \big) \big) + \mathcal{E}(t) \ \mathrm{d} t\notag \\	
		\overset{(g)}  =& \frac{1}{n} \int_0^{T^n} \widetilde{c}\big( n \overline{\mathbf{x}}^{(n) \dagger} \left(t;  \mathbf{x}^n\right), \overline{\mathbf{u}}^{\dagger} \big( n \overline{\mathbf{x}}^{(n) \dagger} \left(t;  \mathbf{x}^n\right) \big)\big) + n \mathcal{E}(t) \ \mathrm{d} t + \mathcal{O}\left(\sqrt{\frac{\log \log n}{n}} \right)	\notag \\
	\overset{(h)} = & \frac{1}{n^2} \int_0^{nT^n} \widetilde{c}\big( n \overline{\mathbf{x}}^{(n) \dagger} \left(t/n;  \mathbf{x}^n\right), \overline{\mathbf{u}}^{\dagger} \big( n \overline{\mathbf{x}}^{(n) \dagger} \left(t/n;  \mathbf{x}^n\right) \big)\big) + \mathcal{O}\left(\sqrt{\frac{\log \log n}{n}} \right)	\notag  \\
	  =&  \frac{1}{n^2} \int_0^{nT^n} \widetilde{c}\big(\overline{\mathbf{x}}^{\dagger} \left(t;  n\mathbf{x}^n\right), \overline{\mathbf{u}}^{\dagger}\big(\mathbf{x}^{\dagger} \left(t;  n\mathbf{x}^n\right)\big)\big)  + \mathcal{O}\left(\sqrt{\frac{\log \log n}{n}} \right)\notag \\
	 	= & \frac{1}{n^2}\int_0^{T} \widetilde{c}\big(\overline{\mathbf{x}}^{ \dagger} \left(t; n\mathbf{x} \right), \overline{\mathbf{u}}^{\dagger}\big(\overline{\mathbf{x}}^{\dagger} \left(t; n\mathbf{x} \right) \big)  \big)\mathrm{d} t+ \mathcal{O}\left(\sqrt{\frac{\log \log n}{n}} \right)
	\end{align}
	where $\mathcal{E}(t)=\sum_{k=1}^K \widetilde{g}_k\left(\overline{\mathbf{u}}_k^{\ast} \left(\overline{\mathbf{x}}^{\ast} \left(t; \mathbf{x} \right) \right) \right) \mathrm{d} t - \sum_{k=1}^K \widetilde{g}_k\big(\overline{\mathbf{u}}_k^{\dagger} \left(\overline{\mathbf{x}}^{\dagger} \left(t; \mathbf{x} \right) \right) \big)$, $(g)$ is due to $\overline{\mathbf{x}}^{(n) \dagger} \left(t;  \mathbf{x}^n\right)=\overline{\mathbf{x}}^{\dagger} \left(t; \mathbf{x} \right)+\mathcal{O}\left(\sqrt{\frac{\log \log n}{n}} \right) $, $\frac{1}{n} \int_0^{T^n} \sum_{k=1}^K\widetilde{g}_k\big(\overline{\mathbf{u}}_k^{\dagger} \big( n \overline{\mathbf{x}}^{(n) \dagger} \left(t;  \mathbf{x}^n\right) \big)\big)  \ \mathrm{d} t=\mathcal{O}\left(\frac{1}{n} \right)$ and $\mathcal{O}\left(\sqrt{\frac{\log \log n}{n}} \right) +\mathcal{O}\left(\frac{1}{n} \right)=\mathcal{O}\left(\sqrt{\frac{\log \log n}{n}} \right)  $, and $(h)$ is due to the fact that $\mathbf{u}^{\ast}$ achieves the optimal total cost when initial state is $\mathbf{x}$, $(g)$ is due to the change of variable from $t$ to $\frac{t}{n}$ and  $\frac{1}{n^2} \int_0^{nT^n}  n\mathcal{E}(t/n) \mathrm{d} t =\frac{1}{n} \int_0^{T}  \mathcal{E}(t/n) \mathrm{d} t =\mathcal{O}\left(\frac{1}{n}\right)$.
\item Second, we prove the following inequality:
	\begin{align}	\label{step2222}
	& \int_0^{T} \widetilde{c}\left(\overline{\mathbf{x}}^{\ast} \left(t; \mathbf{x} \right), \overline{\mathbf{u}}^{\ast} \left(\overline{\mathbf{x}}^{\ast} \left(t; \mathbf{x} \right) \right) \right)\mathrm{d} t - \int_0^{T} \sum_{k=1}^K \widetilde{g}_k\left(\overline{\mathbf{u}}_k^{\ast} \left(\overline{\mathbf{x}}^{\ast} \left(t; \mathbf{x} \right) \right) \right) \mathrm{d} t 	\notag \\
	 \geq & \frac{1}{n^2}\int_0^{T} \widetilde{c}\big(\overline{\mathbf{x}}^{ \dagger} \left(t; n\mathbf{x} \right), \overline{\mathbf{u}}^{\dagger}\big(\overline{\mathbf{x}}^{\dagger} \left(t; n\mathbf{x} \right) \big)  \big)\mathrm{d} t	+ \mathcal{O}\left(\sqrt{\frac{\log \log n}{n}} \right)
	 \end{align}
	We have
\begin{align}
	& \frac{1}{n^2}\int_0^{T} \widetilde{c}\big(\overline{\mathbf{x}}^{ \dagger} \left(t; n\mathbf{x} \right), \overline{\mathbf{u}}^{\dagger}\big(\overline{\mathbf{x}}^{\dagger} \left(t; n\mathbf{x} \right) \big)  \big)\mathrm{d} t \overset{(i)}\leq \frac{1}{n^2}\int_0^{T} \widetilde{c}\big(\overline{\mathbf{x}}^{ \ast} \left(t; n\mathbf{x}^n \right), \overline{\mathbf{u}}^{\ast}\big(\overline{\mathbf{x}}^{\ast} \left(t; n\mathbf{x}^n \right) \big)  \big)\mathrm{d} t	\notag \\
	\overset{(j)}=& \frac{1}{n}\int_0^{T} \widetilde{c}\big(\frac{1}{n}\overline{\mathbf{x}}^{ \ast} \left(t; n\mathbf{x}^n \right), \overline{\mathbf{u}}^{\ast}\big(\frac{1}{n}\overline{\mathbf{x}}^{\ast} \left(t; n\mathbf{x}^n \right) \big)  \big)\mathrm{d} t + \mathcal{O}\left(\frac{1}{n} \right)	\notag \\
	\overset{(k)}= &\frac{1}{n}\int_0^{T} \widetilde{c}\big(\overline{\mathbf{x}}^{(n) \ast} \left(t/n; \mathbf{x}^n \right), \overline{\mathbf{u}}^{\ast}\big(\overline{\mathbf{x}}^{(n)\ast} \left(t/n; \mathbf{x}^n \right) \big)  \big)- \sum_{k=1}^K \widetilde{g}_k\big(\overline{\mathbf{u}}_k^{\ast}\big( \overline{\mathbf{x}}^{(n) \ast} \left(t/n;  \mathbf{x}^n\right)  \big)\big) \mathrm{d} t + \mathcal{O}\left(\sqrt{\frac{\log \log n}{n}} \right)\notag \\
	\overset{(l)} \leq & \int_0^{T^n} \widetilde{c}\big(  \overline{\mathbf{x}}^{(n) \ast} \left(t;  \mathbf{x}^{n}\right), \overline{\mathbf{u}}^{\ast}\big( \overline{\mathbf{x}}^{\ast} \left(t;  \mathbf{x}^n\right)  \big)\big) \mathrm{d} t - \int_0^{T^n} \sum_{k=1}^K \widetilde{g}_k\big(\overline{\mathbf{u}}_k^{\ast}\big( \overline{\mathbf{x}}^{(n) \ast} \left(t;  \mathbf{x}^n\right)  \big)\big) \mathrm{d} t 	+ \mathcal{O}\left(\sqrt{\frac{\log \log n}{n}} \right)\notag \\
	=& \int_0^{T} \widetilde{c}\left(\overline{\mathbf{x}}^{\ast} \left(t; \mathbf{x} \right), \overline{\mathbf{u}}^{\ast} \left(\overline{\mathbf{x}}^{\ast} \left(t; \mathbf{x} \right) \right) \right)\mathrm{d} t - \int_0^{T} \sum_{k=1}^K \widetilde{g}_k\left(\overline{\mathbf{u}}_k^{\ast} \left(\overline{\mathbf{x}}^{\ast} \left(t; \mathbf{x} \right) \right) \right) \mathrm{d} t + \mathcal{O}\left(\sqrt{\frac{\log \log n}{n}} \right)
\end{align}
where $(i)$ is due to the fact that $\mathbf{u}^{\dagger}$ achieves the optimal total cost when initial state is $n\mathbf{x}$, $(j)$ is due to $\frac{1}{n^2}\int_0^{T} \sum_{k=1}^K\widetilde{g}_k\big(\overline{\mathbf{u}}_k^{\ast}\big(\overline{\mathbf{x}}^{\ast} \left(t; n\mathbf{x}^n \right) \big)  \big)\mathrm{d} t - \frac{1}{n}\int_0^{T} \sum_{k=1}^K \widetilde{g}_k\big(\overline{\mathbf{u}}_k^{\ast}\big(\frac{1}{n}\overline{\mathbf{x}}^{\ast} \left(t; n\mathbf{x}^n \right) \big)  \big)\mathrm{d} t  =\mathcal{O}\left(\frac{1}{n}\right)$, $k$ is due to $\frac{1}{n}\int_0^{T} \sum_{k=1}^K \widetilde{g}_k\big(\overline{\mathbf{u}}_k^{\ast}\big( \overline{\mathbf{x}}^{(n) \ast} \left(t/n;  \mathbf{x}^n\right)  \big)\big) \mathrm{d} t=\mathcal{O}\left( \frac{1}{n}\right)$, $(l)$ is due to the change of variable from $t$ to $nt$.
\end{enumerate}
Combining (\ref{step1111}) and (\ref{step2222}), we can prove (\ref{impoEqu2}).   
 
Finally, we prove  Theorem \ref{RelashipVJ}  based on the three equalities  in (\ref{impotequ1}), (\ref{impoEqu}) and (\ref{impoEqu2}). We first prove the following inequality:   $\frac{1}{n^2}V(n \mathbf{x} ) - \frac{1}{n^2}J(n\mathbf{x})\geq\mathcal{O}\left(\sqrt{\frac{\log \log n}{n}} \right)$. Specifically, it is proves as
\begin{align}
	&\frac{1}{n^2} V\left( n \mathbf{x}^n \right) \overset{(q)}=\frac{1}{n^2}\mathbb{E} \Big[ V\left(\mathbf{x}^{\ast}\left(\lfloor nT \rfloor;  n \mathbf{x}^n  \right)\right) +\sum_{i=0}^{\lfloor nT \rfloor-1} \mathbb{E} \left[ \widetilde{c}\left(\mathbf{x}^{\ast} \left( i;  n \mathbf{x}^n   \right),\mathbf{u}^{\ast}\left(\mathbf{x}^{\ast} \left( i;  n \mathbf{x}^n  \right) \right) \right) \right] \Big]+ \mathcal{O}\left(\frac{1}{n} \right)	\notag \\
	\overset{(r)}\geq & \int_0^{T^n} \widetilde{c}\left(\overline{\mathbf{x}}^{n\ast} \left( t; \mathbf{x}^n \right), \mathbf{u}\left(\overline{\mathbf{x}}^{n \ast} \left( t; \mathbf{x}^n \right) \right) \right)\mathrm{d} t -  \lim_{n \rightarrow \infty}\int_0^{T^n} \sum_{k=1}^K \widetilde{g}_k\left( \mathbf{u}\left( \overline{\mathbf{x}}^{n\ast} \left(t;  \mathbf{x}^n\right)  \right)\right) \mathrm{d} t+ \mathcal{O}\left(\frac{1}{n} \right)	\notag \\
	\overset{(s)}=& \frac{1}{n^2}\int_0^{T} \widetilde{c}\left(\overline{\mathbf{x}}^{dt \ast} \left(t; n\mathbf{x} \right),  \overline{\mathbf{u}}^{dt \ast}\left(\overline{\mathbf{x}}^{dt\ast} \left(t; n\mathbf{x} \right)\right) \right)\mathrm{d} t + \mathcal{O}\left(\sqrt{\frac{\log \log n}{n}} \right)	\notag \\
	\overset{(t)}\geq& \frac{1}{n^2} \inf_{\overline{\mathbf{x}}} \int_0^{T} \widetilde{c}\left(\overline{\mathbf{x}} \left(t; n\mathbf{x} \right),  \overline{\mathbf{u}}\left(\overline{\mathbf{x}} \left(t; n\mathbf{x} \right)\right) \right)\mathrm{d} t = \frac{1}{n^2} J\left( n \mathbf{x} \right)	+ \mathcal{O}\left(\sqrt{\frac{\log \log n}{n}} \right)\label{lastimpo1}
\end{align}
where $(q)$ is due to (\ref{impotequ1}), $(r)$ is due to the positive property of $V$ and (\ref{impoEqu}), (s) is due to  \\ $\int_0^{T^n} \widetilde{c}\left(\overline{\mathbf{x}}^{n\ast} \left( t; \mathbf{x}^n \right), \mathbf{u}\left(\overline{\mathbf{x}}^{n \ast} \left( t; \mathbf{x}^n \right) \right) \right)-  \sum_{k=1}^K \widetilde{g}_k\left( \mathbf{u}\left( \overline{\mathbf{x}}^{n \ast} \left(t;  \mathbf{x}^n\right)  \right)\right) \mathrm{d} t \\ = \frac{1}{n}\int_0^{T^n} \widetilde{c}\big( n \overline{\mathbf{x}}^{n, dt \ast} \left(t;  \mathbf{x}^n\right), \overline{\mathbf{u}}^{dt\ast} \big( n \overline{\mathbf{x}}^{n,dt\ast} \left(t;  \mathbf{x}^n\right) \big)\big) \mathrm{d} t+\mathcal{O}\left(\sqrt{\frac{\log \log n}{n}} \right)\\
=\frac{1}{n}\int_0^{T^n} \widetilde{c}\big(  \overline{\mathbf{x}}^{n, dt \ast} \left(nt;  n\mathbf{x}^n\right), \overline{\mathbf{u}}^{dt\ast} \big(  \overline{\mathbf{x}}^{n,dt\ast} \left(nt;  n\mathbf{x}^n\right) \big)\big) \mathrm{d} t + \mathcal{O}\left(\sqrt{\frac{\log \log n}{n}} \right)\\
= \frac{1}{n^2}\int_0^{T^n} \widetilde{c}\big(  \overline{\mathbf{x}}^{n, dt \ast} \left(t;  n\mathbf{x}^n\right), \overline{\mathbf{u}}^{dt\ast} \big(  \overline{\mathbf{x}}^{n,dt\ast} \left(t;  n\mathbf{x}^n\right) \big)\big) \mathrm{d} t +\mathcal{O}\left(\sqrt{\frac{\log \log n}{n}} \right)\\ =\frac{1}{n^2}\int_0^{T} \widetilde{c}\left(\overline{\mathbf{x}}^{dt \ast} \left(t; n\mathbf{x} \right),  \overline{\mathbf{u}}^{dt \ast}\left(\overline{\mathbf{x}}^{dt\ast} \left(t; n\mathbf{x} \right)\right) \right)\mathrm{d} t + \mathcal{O}\left(\sqrt{\frac{\log \log n}{n}} \right)$, $(t)$ is due the infimum over all fluid trajectories starting from $n \mathbf{x} $. We next prove the following  inequality:   $\frac{1}{n^2} V(n \mathbf{x} ) - \frac{1}{n^2}J(n\mathbf{x}) \leq \mathcal{O}\left(\sqrt{\frac{\log \log n}{n}} \right)$. Based (\ref{impotequ1}), if $V$ solve the Bellman equation in (\ref{OrgBel}), then for any policy $\mathbf{u} \in \boldsymbol{\mathcal{U}} $, we have
\begin{align}	\label{lasrimpo2}
	 & \frac{1}{n^2} V\left( n \mathbf{x}^n \right)  \notag \\
	 \leq & \frac{1}{n^2}\mathbb{E} \Big[ V\left(\mathbf{x}_{\mathbf{u}}\left(\lfloor nT \rfloor;  n \mathbf{x}^n  \right)\right) \Big] +\int_0^{T^n} \widetilde{c}\left(\overline{\mathbf{x}}^{n}_\mathbf{u} \left( t; \mathbf{x}^n \right), \mathbf{u}\left(\overline{\mathbf{x}}^{n}_\mathbf{u} \left( t; \mathbf{x}^n \right) \right) \right)\mathrm{d} t -      \int_0^{T^n} \sum_{k=1}^K \widetilde{g}_k\left( \mathbf{u}_k\left( \overline{\mathbf{x}}_{\mathbf{u}}^n \left(t;  \mathbf{x}^n\right)  \right)\right) \mathrm{d} t
\end{align}

According to \emph{Lemma 10.6.6} of \cite{CompApp}, fixing $\epsilon_0 \in (0, \epsilon)$, we can choose a piecewise linear trajectory $\overline{\mathbf{x}}^{\epsilon}$ satisfying $\| \overline{\mathbf{x}}^{\epsilon} \left(t; \mathbf{x} \right) - \overline{\mathbf{x}}^{\ast}\left(t; \mathbf{x} \right) \| \leq \epsilon_0$ for all $t \geq 0$, $\mathbf{x} \in \boldsymbol{\mathcal{X}}$, $\overline{\mathbf{x}}^{\epsilon} \left(t; \mathbf{x} \right)=0$ for $t \geq T$, and the control trajectory satisfies the requirements in (10.54) and (10.55) in \cite{CompApp}. Then according to \emph{Proposition 10.5.3}, we can construct a randomized policy $\widetilde{\mathbf{u}}$ so that $\overline{\mathbf{x}}^n_{\widetilde{\mathbf{u}}}\left(t ;  \mathbf{x}^n  \right)$ converges to $\overline{\mathbf{x}}^{\epsilon}\left(t; \mathbf{x} \right)$. Therefore, we have
\begin{align}	\label{epsipol}
	\int_0^{T^n} \widetilde{c}\left(\overline{\mathbf{x}}^{n}_{\widetilde{\mathbf{u}}} \left( t; \mathbf{x}^n \right), \widetilde{\mathbf{u}}\left(\overline{\mathbf{x}}^{n}_{\widetilde{\mathbf{u}}} \left( t; \mathbf{x}^n \right) \right) \right)\mathrm{d} t \leq \int_0^{T^n} \widetilde{c}\left(\overline{\mathbf{x}}^{\ast}\left(t; \mathbf{x} \right), {\mathbf{u}}^{\ast}\left(\overline{\mathbf{x}}^{\ast}\left(t; \mathbf{x} \right) \right) \right)\mathrm{d} t  + \mathcal{O}\left(\sqrt{\frac{\log \log n}{n}} \right)+\epsilon
\end{align}
for any $\epsilon >0$. Using $\widetilde{\mathbf{u}}$ in (\ref{lasrimpo2}), we have 
\begin{align}
	\frac{1}{n^2}\mathbb{E}\left[ V\left(\mathbf{x}_{\widetilde{\mathbf{u}}}\left(\lfloor nT \rfloor;  n \mathbf{x}^n  \right)\right) \right]= \frac{1}{n^2}\mathbb{E}\left[ V\big(n \overline{\mathbf{x}}^{\epsilon}\left(T; \mathbf{x} \right)+ n\mathcal{O}\left(\sqrt{\frac{\log \log n}{n}} \right)\big) \right]\overset{(u)}=\mathcal{O}\left(\frac{\log \log n}{n}\right)
\end{align}
where $(u)$ is due to $V\left(\mathbf{x} \right)=\mathcal{O}\big( \left\|\mathbf{x}\right\|^2\big)$ \cite{CompApp}. Then, using similar steps in (\ref{lastimpo1}), we can obtain that $\frac{1}{n^2} V(n \mathbf{x} ) - \frac{1}{n^2}J(n\mathbf{x}) \leq \mathcal{O}\left(\sqrt{\frac{\log \log n}{n}} \right)$. Combining the result that $\frac{1}{n^2} V(n \mathbf{x} ) - \frac{1}{n^2}J(n\mathbf{x}) \geq \mathcal{O}\left(\sqrt{\frac{\log \log n}{n}} \right)$, we have $\frac{1}{n^2} V(n \mathbf{x} ) - \frac{1}{n^2}J(n\mathbf{x}) = \mathcal{O}\left(\sqrt{\frac{\log \log n}{n}} \right)$.  Then, by changing variable from $n \mathbf{x}$ to $\mathbf{x}$, we can  obtain that 
\begin{equation}		
	  |V( \mathbf{x} ) - J( \mathbf{x})|   = \mathcal{O}\left( \| \mathbf{x}\|\sqrt{\| \mathbf{x}\| \log \log \| \mathbf{x}\|} \right)
\end{equation}
This completes the proof.

\section*{Appendix B: Proof of Lemma \ref{reasmalltau}}	
For sufficiently small epoch duration $\tau$, the original Bellman equation can be written in the form as the simplified Bellman equation as in (\ref{simbelman}). We then write the HJB equation in  (\ref{cenHJB}) in the following form:
\begin{equation}	\label{cenHJBapp}
		c^\infty = \min_{\mathbf{u}} \left[c\left(\mathbf{x}, \mathbf{u}\right) + \nabla_{\mathbf{x}} J\left(\mathbf{x}\right) \left[ \ \overline{\mathbf{f}}\left(\mathbf{u}, \boldsymbol{\epsilon}\right)  \right]^T \right] 
	\end{equation} 
Comparing (\ref{cenHJBapp}) and (\ref{simbelman}), the following relationship between these two equations can be obtained: $c^\infty=\theta$ and $V(\mathbf{x})=J(\mathbf{x})$.

\section*{Appendix C: Proof of Lemma \ref{linearAp}}	
The HJB equation for the VCTS in (\ref{cenHJB}) can be written as
\begin{equation}		\label{equivalentHJB}
		\min_{\mathbf{u}} \left[ \sum_{k=1}^K \left(c_k\left(\mathbf{x}_k, \mathbf{u}_k\right)+ \nabla_{\mathbf{x}_k} J\left(\mathbf{x}; \boldsymbol{\epsilon}\right)   \left[ \ \overline{\mathbf{f}}_k\left(\mathbf{u}_k, \mathbf{u}_{-k},\boldsymbol{\epsilon}_k \right) \right]^T \right)\right] =0
\end{equation}
Setting the coupling parameters equal to zero in the above equation, we could obtain the associated HJB equation for the base VCTS as follows:
\begin{equation}		\label{veryequ1}
		\min_{\mathbf{u}} \left[ \sum_{k=1}^K \left(c_k\left(\mathbf{x}_k, \mathbf{u}_k\right)+ \nabla_{\mathbf{x}_k} J\left(\mathbf{x}; \mathbf{0}\right)   \left[ \ \overline{\mathbf{f}}_k\left(\mathbf{u}_k, \mathbf{0},\mathbf{0} \right) \right]^T \right)\right] =0
\end{equation}
where $\overline{\mathbf{f}}_k\left(\mathbf{u}_k, \mathbf{0}, \mathbf{0}\right) =  \overline{\mathbf{f}}_k\left(\mathbf{u}_k\left(t\right), \mathbf{u}_{-k}\left(t\right), \boldsymbol{\epsilon}_k\right) \big|_{\boldsymbol{\epsilon}_{k}=\mathbf{0}, \mathbf{u}_{-k}=\mathbf{0}}$.  Suppose $J(\mathbf{x}; \mathbf{0}) = \sum_{k=1}^K J_k \left(\mathbf{x}_k \right)$, where $J_k \left(\mathbf{x}_k \right)$ is the per-flow fluid value function, i.e.,  the solution of  the per-flow HJB equation in (\ref{perflowHJB}). The L.H.S. of  (\ref{veryequ1}) becomes 
\begin{align}		
		\text{L.H.S. of } (\ref{veryequ1})  &\overset{(a)}= \min_{\mathbf{u}} \left[ \sum_{k=1}^K \left(c_k\left(\mathbf{x}_k, \mathbf{u}_k\right)+ \nabla_{\mathbf{x}_k} J_k\left(\mathbf{x}_k\right)   \left[ \ \overline{\mathbf{f}}_k\left(\mathbf{u}_k, \mathbf{0},\mathbf{0} \right) \right]^T \right)\right]  	\notag \\
		&=   \sum_{k=1}^K \min_{\mathbf{u}_k} \left[ c_k\left(\mathbf{x}_k, \mathbf{u}_k\right)+ \nabla_{\mathbf{x}_k} J_k\left(\mathbf{x}_k\right)   \left[ \ \overline{\mathbf{f}}_k\left(\mathbf{u}_k, \mathbf{0},\mathbf{0} \right) \right]^T \right] \overset{(b)}= 0
\end{align}
where $(a)$ is due to $\nabla_{\mathbf{x}_k} J(\mathbf{x}; \mathbf{0})  = \nabla_{\mathbf{x}_k} J_k\left(\mathbf{x}_k\right)$ and $(b)$ is due to the fact that $J_k\left(\mathbf{x}_k\right)$ is the solution of the per-flow HJB equation in (\ref{perflowHJB}). Therefore, we  show that $J(\mathbf{x}; \mathbf{0}) = \sum_{k=1}^K J_k \left(\mathbf{x}_k \right)$ is the solution of (\ref{veryequ1}). This completes the proof.

\section*{Appendix D: Proof of Theorem \ref{RelashipJsJ}}	
First, we obtain the first order Taylor expansion of the L.H.S. of the HJB equation in (\ref{cenHJB}) at $\epsilon_{kj} \mathbf{u}_j=\mathbf{0}$ $(\forall k, j \in \mathcal{K}, k \neq j)$. Taking the first order Taylor expansion of $\overline{\mathbf{f}}_k\left(\mathbf{u}_k, \mathbf{u}_{-k}, \boldsymbol{\epsilon}_k \right)$ at $\epsilon_{kj} \mathbf{u}_j=\mathbf{0}$ $(\forall k, j \in \mathcal{K}, k \neq j)$,  we have
\begin{align}	\label{approxF}
	\overline{\mathbf{f}}_k\left(\mathbf{u}_k, \mathbf{u}_{-k}, \boldsymbol{\epsilon}\right) = \overline{\mathbf{f}}_k\left(\mathbf{u}_k, \mathbf{0},\mathbf{0}\right)  + \sum_{j \neq k} \epsilon_{kj} \mathbf{u}_j   \nabla_{\epsilon_{kj} \mathbf{u}_j}  \overline{\mathbf{f}}_k\left(\mathbf{u}_k, \mathbf{0}, \mathbf{0} \right) + \mathcal{O}\left( \epsilon^2 \right), \quad \text{as } \epsilon \rightarrow 0
\end{align} 
$\nabla_{\epsilon_{kj} \mathbf{u}_j}  \overline{\mathbf{f}}_k$ means taking partial derivative w.r.t. to each element of vector $\epsilon_{kj} \mathbf{u}_j$.   We use $\mathcal{O}\left( \epsilon^2 \right)$ to characterize the growth rate of a function as $\epsilon$ goes to zero and in the following proof  we will not mention  `$\text{as } \epsilon \rightarrow 0$' for simplicity. Let $N_k$ and $N_{\mathbf{u}_k}$ be the dimensions of row vectors $\mathbf{f}_k$   and $\mathbf{u}_k$. Then,  $\nabla_{\epsilon_{kj} \mathbf{u}_j}  \overline{\mathbf{f}}_k$  is a $N_{\mathbf{u}_j} \times N_k$ dimensional matrix.  Taking the first order taylor expansion of $J\left(\mathbf{x}; \boldsymbol{\epsilon}\right)$ at $\boldsymbol{\epsilon} = \mathbf{0}$, we have
\begin{align}	
	& J(\mathbf{x}; \boldsymbol{\epsilon}) = J(\mathbf{x}; \mathbf{0}) + \sum_{i=1}^K \sum_{j \neq i} \epsilon_{ij}   \widetilde{J}_{ij}\left( \mathbf{x} \right)+ \mathcal{O}\left( \epsilon^2 \right) \overset{(a)}= \sum_{i=1}^K J_i \left(\mathbf{x}_i \right) + \sum_{i=1}^K \sum_{j \neq i} \epsilon_{ij} \widetilde{J}_{ij}\left( \mathbf{x} \right)+\mathcal{O}\left( \epsilon^2 \right)	\label{approxJxx}	\\
	\Rightarrow & \left| J(\mathbf{x}; \boldsymbol{\epsilon}) - \sum_{i=1}^K J_i \left(\mathbf{x}_i \right) \right|  = \left| \sum_{i=1}^K\sum_{j \neq i }\epsilon_{ij} \widetilde{J}_{ij}\left( \mathbf{x} \right) \right|  + \mathcal{O}\left( \epsilon^2 \right)		\label{approxJthm}
\end{align}
where $ \widetilde{J}_{ij}\left( \mathbf{x} \right) \triangleq \frac{\partial J(\mathbf{x}; \boldsymbol{\epsilon})}{\partial \epsilon_{ij}}\big|_{\boldsymbol{\epsilon}=\mathbf{0}} $ and $(a)$ is due to $\nabla_{\mathbf{x}_k} J(\mathbf{x}; \mathbf{0})  = \nabla_{\mathbf{x}_k} J_k\left(\mathbf{x}_k\right)$. Note that the equation in (\ref{approxJthm}) quantifies the difference $\left| J(\mathbf{x}; \boldsymbol{\epsilon}) - \sum_{k=1}^K J_i \left(\mathbf{x}_i \right) \right|$ in terms of the coupling parameters $\epsilon_{ij}$  and $\widetilde{J}_{ij}\left( \mathbf{x} \right)$ $(\forall i,j \in \mathcal{K}, i\neq j)$. Substituting (\ref{approxF}) and (\ref{approxJxx}) into (\ref{equivalentHJB}), which is an equivalent form of the HJB equation in (\ref{cenHJB}), we have
\begin{align}
	&\min_{\mathbf{u}} \left[ \sum_{k=1}^K \left(c_k\left(\mathbf{x}_k, \mathbf{u}_k\right)+ \nabla_{\mathbf{x}_k} \left( \sum_{i=1}^K J_i \left(\mathbf{x}_i \right) + \sum_{i=1}^K \sum_{j \neq i} \epsilon_{ij} \widetilde{J}_{ij}\left( \mathbf{x} \right)+\mathcal{O}\left( \epsilon^2 \right) \right) \cdot    \right. \right. \notag \\
	 &\hspace{4cm}  \left.  \left.   \left[  \overline{\mathbf{f}}_k\left(\mathbf{u}_k, \mathbf{0},  \mathbf{0} \right) + \sum_{j \neq k} \epsilon_{kj} \mathbf{u}_j   \nabla_{\epsilon_{kj} \mathbf{u}_j}  \overline{\mathbf{f}}_k\left(\mathbf{u}_k, \mathbf{0}, \mathbf{0} \right) + \mathcal{O}\left( \epsilon^2 \right) \right]^T \right)\right] =0 \notag		\\	
	\overset{(b)}\Rightarrow	&\min_{\mathbf{u}} \left[ \sum_{k=1}^K \left(c_k\left(\mathbf{x}_k, \mathbf{u}_k\right) + \nabla_{\mathbf{x}_k}J_k\left(\mathbf{x}_k\right)\left[ \ \overline{\mathbf{f}}_k\left(\mathbf{u}_k, \mathbf{0}, \mathbf{0} \right)  \right]^T  + \sum_{i=1}^K \sum_{j \neq i} \epsilon_{ij} \nabla_{\mathbf{x}_k} \widetilde{J}_{ij}\left( \mathbf{x} \right) \left[ \ \overline{\mathbf{f}}_k\left(\mathbf{u}_k, \mathbf{0}, \mathbf{0} \right) \right]^T   	\right. \right.	\qquad  \notag \\
	& \hspace{4cm}  \left. \left.  +  \nabla_{\mathbf{x}_k}J_k \left(\mathbf{x}_k \right) \left[   \sum_{j \neq k} \epsilon_{kj} \mathbf{u}_j   \nabla_{\epsilon_{kj} \mathbf{u}_j}  \overline{\mathbf{f}}_k\left(\mathbf{u}_k, \mathbf{0}, \mathbf{0} \right)  \right]^T \right) + \mathcal{O}\left( \epsilon^2 \right)\right] =0		\label{transformHJB}
\end{align}
where $(b)$ is due to $\nabla_{\mathbf{x}_k} \left( \sum_{i=1}^K J_i \left(\mathbf{x}_i \right)   \right)= \nabla_{\mathbf{x}_k} J_k\left(\mathbf{x}_k\right)$.

Second, we compare the difference between the optimal control policy under the general coupled VCTS (denoted as $\mathbf{u}_k^{c \ast}$)  and  the optimal policy under the decoupled base VCTS (denoted as $\mathbf{u}_k^{v \ast }$). Before proceeding, we show the following  lemma.
\begin{Lemma}		\label{applemma}
	Consider the following two convex optimization problems:
	\begin{align}
		(\mathcal{P}_1): \ \min_{x,y \in \mathbb{R}} f_1\left(x\right)+f_2\left(y\right)	\qquad \qquad 	(\mathcal{P}_2): \ \min_{x,y \in \mathbb{R}} f_1\left(x\right)+f_2\left(y\right) + \epsilon g\left(x,y\right)	+ \mathcal{O}\left(\epsilon^2 \right)	\label{p1p2lamma}
	\end{align}
	where $\mathcal{P}_2$ is a perturbed problem w.r.t. $\mathcal{P}_1$. Let $\left( x^{\ast}, y^{\ast}\right)$ be the optimal solution of $\mathcal{P}_1$ and  $\left( x^{\ast}\left(\epsilon \right), y^{\ast}\left(\epsilon \right) \right)$ be the optimal solution of $\mathcal{P}_2$. Then, we have
	\begin{align}
		&x^{\ast}\left(\epsilon \right)- x^{\ast}	= -\epsilon \frac{g_x' \left(x^{\ast},y^{\ast}\right)}{f_1^{''}\left(x^{\ast}\right)} + \mathcal{O}\left(\epsilon^2 \right)	\label{proofxstar} \\
		&y^{\ast}\left(\epsilon \right)- y^{\ast}  = - \epsilon	\frac{g_y'\left(x^{\ast},y^{\ast}\right)}{f_2^{''}\left(y^{\ast}\right)} + \mathcal{O}\left(\epsilon^2 \right)	\label{proofystar}
	\end{align}
	where $g_x' (x,y) =\frac{\partial g (x,y)}{\partial x}$, $g_y' (x,y) =\frac{\partial g (x,y)}{\partial y}$, $f_1^{''}(x)=\frac{\mathrm{d}^2 f_1(x)}{\mathrm{d}x^2}$, and $f_2^{''}(y)=\frac{\mathrm{d}^2 f_2(y)}{\mathrm{d} y^2}$.~\hfill\IEEEQED
\end{Lemma} 
\begin{proof} [Proof of Lemma \ref{applemma}]
	According to the first order optimality condition of $\mathcal{P}_1$ and $\mathcal{P}_2$, we have
	\begin{equation}	\label{baseequ}
		\left\{  \begin{array}{cc}
					f_1'\left(x^{\ast}\right)=0		\\
					f_2'\left(y^{\ast}\right)=0
			\end{array}\right.
		\qquad  \text{and}\qquad 
		\left\{  \begin{array}{cc}
					f_1'\left(x^{\ast}\left(\epsilon \right)\right) + \epsilon g_x' \left(x^{\ast}\left(\epsilon \right),y^{\ast}\left(\epsilon \right)\right) + \mathcal{O}\left(\epsilon^2 \right) =0		\\
					f_2'\left(y^{\ast}\left(\epsilon \right)\right) + \epsilon g_y' \left(x^{\ast}\left(\epsilon \right),y^{\ast}\left(\epsilon \right)\right) + \mathcal{O}\left(\epsilon^2 \right) =0
			\end{array}\right.
	\end{equation}
	where $f_1^{'}(x)=\frac{\mathrm{d} f_1(x)}{\mathrm{d}x}$ and $f_2^{'}(y)=\frac{\mathrm{d} f_2(y)}{\mathrm{d} y}$. Taking the first order Taylor expansion of $x^{\ast}\left(\epsilon \right)$  at $\epsilon =0 $, we have
\begin{align}
	& x^{\ast}\left(\epsilon \right) = x^{\ast}\left(0 \right) + \epsilon \widetilde{x}  + \mathcal{O}\left(\epsilon^2 \right)\overset{(c)}=x^{\ast}+ \epsilon \widetilde{x} + \mathcal{O}\left(\epsilon^2 \right)		\label{taylorexpx}		\\
	\Rightarrow & 	x^{\ast}\left(\epsilon \right) - x^{\ast}= \epsilon \widetilde{x} + \mathcal{O}\left(\epsilon^2 \right)		\label{taylorexpxxx}
\end{align}
where $ \widetilde{x} \triangleq  \frac{\mathrm{d} x^{\ast}\left(\epsilon \right) }{ \mathrm{d} \epsilon } \Big|_{\epsilon=0}$. $(c)$ is due to the equivalence between $\mathcal{P}_1$ and $\mathcal{P}_2$ when $\epsilon=0$, i.e., $ x^{\ast}\left(0 \right)= x^{\ast}$.  

Similarly, we have the following relationship between $y^{\ast}\left(\epsilon \right)$ and $y^{\ast}$:
\begin{align}	\label{taylorexpy}
	y^{\ast}\left(\epsilon \right) - y^{\ast} = \epsilon \widetilde{y} + \mathcal{O}\left(\epsilon^2 \right)	
\end{align}
where $\widetilde{y} \triangleq  \frac{\mathrm{d} y^{\ast}\left(\epsilon \right) }{ \mathrm{d} \epsilon }\Big|_{\epsilon=0}$. Taking the first order Taylor expansion of the L.H.S. of the  first equation of the second term in (\ref{baseequ}) at $\left( x^{\ast}\left(\epsilon \right), y^{\ast}\left(\epsilon \right)  \right)= \left(x^{\ast}, y^{\ast} \right)$, we have
\begin{align}
	&f_1'\left(x^{\ast}\left(\epsilon \right)\right) + \epsilon g_x' \left(x^{\ast}\left(\epsilon \right),y^{\ast}\left(\epsilon \right)\right)+ \mathcal{O}\left(\epsilon^2 \right)  \notag \\
	=& f_1'\left(x^{\ast}\right)  + f_1^{''} \left( x^{\ast} \right) \left(x^{\ast}\left(\epsilon \right)- x^{\ast}  \right) + \epsilon \left( g_x' \left(x^{\ast},y^{\ast}\right) + g_{xx}^{''} \left( x^{\ast}, y^{\ast} \right) \left(x^{\ast}\left(\epsilon \right)- x^{\ast}  \right)+ g_{xy}^{''} \left(x^{\ast},y^{\ast}\right)   \left(y^{\ast}\left(\epsilon \right)- y^{\ast}  \right) \right. \notag \\
	&\hspace{7.5cm}\left.+ \mathcal{O}\left( \left(x^{\ast}\left(\epsilon \right)- x^{\ast}  \right)^2 \right)+\mathcal{O}\left( \left(y^{\ast}\left(\epsilon \right)- y^{\ast}  \right)^2 \right)\right)+ \mathcal{O}\left(\epsilon^2 \right)	\notag \\
	\overset{(d)}=& f_1^{''} \left( x^{\ast} \right) \left(x^{\ast}\left(\epsilon \right)- x^{\ast}  \right) + \epsilon  g_x' \left(x^{\ast},y^{\ast}\right) + \mathcal{O}\left(\epsilon^2 \right)	\label{taylor1}
\end{align}
where $(d)$ is due to the first equation of the first term in (\ref{baseequ}), (\ref{taylorexpxxx}) and (\ref{taylorexpy}). Substituting (\ref{taylorexpxxx}) into (\ref{taylor1}) and by the definition of $\widetilde{x}$, we  have
\begin{align}	\label{appended}
	\widetilde{x} = - \frac{g_x' \left(x^{\ast},y^{\ast}\right)}{ f_1^{''} \left( x^{\ast} \right)}
\end{align}
Similarly, we could obtain
\begin{align}	\label{appendedyy}
	\widetilde{y} = -\frac{g_y'\left(x^{\ast},y^{\ast}\right)}{f_2^{''}\left(y^{\ast}\right)}
\end{align}
Therefore, substituting (\ref{appended}) into (\ref{taylorexpxxx}) and (\ref{appendedyy}) into (\ref{taylorexpy}), we obtain (\ref{proofxstar}) and (\ref{proofystar}).
\end{proof}
\begin{Corollary} [Extension of Lemma \ref{applemma}]		\label{corrAppd}
	Consider the two convex optimization problems in (\ref{p1p2lamma}) with $x, y \in G$, where $G=\left[ G_{min}, G_{max} \right] $ is a subset of $R$, i.e., $G \subset R$.  Let $\left(x^{\dagger}, y^{\dagger} \right)$ and  $\left(x^{\dagger}\left(\epsilon \right), y^{\dagger}\left(\epsilon \right) \right)$ be the optimal solutions of the corresponding two problems. Then, we have $\left| x^{\dagger}\left(\epsilon \right)- x^{\dagger} \right| \leq \mathcal{O}\left(\epsilon \right)$, $\left| y^{\dagger}\left(\epsilon \right)- y^{\dagger} \right| \leq \mathcal{O}\left( \epsilon \right)$.  Furthermore, we conclude that either of the following equalities holds:
	\begin{align}
		& \left(x^{\dagger}\left(\epsilon \right) -  x^{\dagger} \right) f_1'\left(x^{\dagger} \right)=0		\\
		& \left(x^{\dagger}\left(\epsilon \right) -  x^{\dagger} \right) f_1'\left(x^{\dagger} \right)= \sum_{i \geq 0} \epsilon^{1+\delta_i} \tilde{f}_i \left(x^{\ast}, y^{\ast}, G_{min},G_{max} \right)
	\end{align}	
	for some function $\tilde{f}_i $ and positive constants $\{\delta_i: i \geq 0\}$, where $0 < \delta_0 \leq \delta_1 \leq \delta_2 \leq \dots \leq \infty $.~\hfill\IEEEQED
\end{Corollary}
\begin{proof}	[Proof of Corollary \ref{corrAppd}]
	The optimal solutions $\left(x^{\dagger}, y^{\dagger} \right)$ and  $\left(x^{\dagger}\left(\epsilon \right), y^{\dagger}\left(\epsilon \right) \right)$ of the two new convex optimization problems  can be obtained by mapping each element of $\left(x^{\ast}, y^{\ast} \right)$ and  $\left(x^{\ast}\left(\epsilon \right), y^{\ast}\left(\epsilon \right) \right)$  to the set $G$.  Specifically, if $x^{\ast} \in G$,  $x^{\dagger}=x^{\ast}$. if $x^{\ast} > G_{max}$,  $x^{\dagger}= G_{max}$. if  $x^{\ast} < G_{min}$,  $x^{\dagger}= G_{min}$. Therefore, we have 
	\begin{align}
		\left|x^{\dagger}\left(\epsilon \right)- x^{\dagger} \right| \leq \left| x^{\ast}\left(\epsilon \right)- x^{\ast} \right|  =  \left|\epsilon \frac{g_x' \left(x^{\ast},y^{\ast}\right)}{f_1^{''}\left(x^{\ast}\right)} \right|+ \mathcal{O}\left(\epsilon^2 \right) = \mathcal{O}\left( \epsilon \right)
	\end{align}
	Similarly, we could obtain
	\begin{align}
		\left| y^{\dagger}\left(\epsilon \right)- y^{\dagger} \right| \leq \mathcal{O}\left( \epsilon \right)
	\end{align}
	where the equality  is achieved when $x^{\dagger}=x^{\ast}$ and $y^{\dagger}=y^{\ast}$. 
	
	Next, we prove the property of the expression $\left(x^{\dagger}\left(\epsilon \right) -  x^{\dagger} \right) f_1'\left(x^{\dagger} \right)$. Based on the above analysis, when $x^{\ast} \leq G_{min}, x^{\ast}\left(\epsilon \right) \leq G_{min}$ or $x^{\ast}\geq G_{max}, x^{\ast}\left(\epsilon \right) \geq G_{max}$, we have $x^{\dagger}\left(\epsilon \right) -  x^{\dagger} =0$. When $x^{\ast} \in G$, we have $f_1'\left(x^{\dagger} \right)=0$. Thus, at these cases, we have $\left(x^{\dagger}\left(\epsilon \right) -  x^{\dagger} \right) f_1'\left(x^{\dagger} \right)=0$.  At other cases when $x^{\ast} \notin G, x^{\ast}\left(\epsilon \right) \notin G$ or $x^{\ast} \notin G, x^{\ast}\left(\epsilon \right) \in G$, we have that $\left|x^{\dagger}\left(\epsilon \right)- x^{\dagger} \right| < \left| x^{\ast}\left(\epsilon \right)- x^{\ast} \right|  = \mathcal{O}\left( \epsilon \right)$. It means that as $\epsilon$ goes to zero, $\left|x^{\dagger}\left(\epsilon \right)- x^{\dagger} \right|$ goes to zero faster than $\epsilon$. Therefore, we can write the difference  between $x^{\dagger}\left(\epsilon \right)$ and  $x^{\dagger}$ as  $x^{\dagger}\left(\epsilon \right) -  x^{\dagger} = \sum_{i \geq 0} \epsilon^{1+\delta_i} \hat{f}_i \left(x^{\ast}, y^{\ast},G_{min},G_{max} \right)$ for some function $\hat{f}_i $ and positive constants $\{\delta_i: i \geq 0\}$, where $0 < \delta_0 \leq \delta_1 \leq \delta_2 \leq \dots \leq \infty $. Since $x^{\dagger}$ is determined  based on the knowledge of $x^{\ast}$, $G_{min}$ and $G_{max}$, we have $x^{\dagger}=\vec{f}\left(x^{\ast}, G_{min},G_{max} \right)$ for some $\vec{f}$. Finally, we have $\left(x^{\dagger}\left(\epsilon \right) -  x^{\dagger} \right) f_1'\left(x^{\dagger} \right)= \sum_{i \geq 0} \epsilon^{1+\delta_i} \tilde{f}_i \left(x^{\ast}, y^{\ast},G_{min},G_{max} \right)$ with $\tilde{f}_i  = \hat{f}_i f_1'\vec{f}$ for all $i \geq 0$.
\end{proof}

The results in Lemma \ref{applemma} and  Corollary \ref{corrAppd} can be easily extended to the case where $x$ and $y$ are vectors and there are more perturbed terms like $\epsilon g \left(x,y \right)$ in the objective function of $\mathcal{P}_2$ in (\ref{p1p2lamma}). In the following, we easablish the property of  the difference between $\mathbf{u}_k^{c \ast}$ and $\mathbf{u}_k^{v \ast }$. Based on the definitions of $\mathbf{u}_k^{c \ast}$ and $\mathbf{u}_k^{v \ast }$ as well as the equations in  (\ref{transformHJB}) and (\ref{perflowHJB}), we have
\begin{align}
	\mathbf{u}_k^{c \ast}\left( \mathbf{x} \right) &= \arg \min_{\mathbf{u}_k} \left\{ g_k\left( \mathbf{u}_k \right) +   \nabla_{\mathbf{x}_k}J_k\left(\mathbf{x}_k\right)\left[ \ \overline{\mathbf{f}}_k\left(\mathbf{u}_k, \mathbf{0}, \mathbf{0}\right)  \right]^T  + \sum_{i=1}^K \sum_{j \neq i} \epsilon_{ij} \nabla_{\mathbf{x}_k} \widetilde{J}_{ij}\left( \mathbf{x} \right) \left[ \ \overline{\mathbf{f}}_k\left(\mathbf{u}_k, \mathbf{0},\mathbf{0} \right) \right]^T   \right. \notag \\
	  & \left.+ \nabla_{\mathbf{x}_k}J_k \left(\mathbf{x}_k \right) \left[   \sum_{j \neq k} \epsilon_{kj} \mathbf{u}_j   \nabla_{\epsilon_{kj} \mathbf{u}_j}  \overline{\mathbf{f}}_k\left(\mathbf{u}_k, \mathbf{0},\mathbf{0} \right)  \right]^T  + \sum_{j \neq k}  \nabla_{\mathbf{x}_j}J_j(\mathbf{x}_j) \left[ \epsilon_{jk} \mathbf{u}_k   \nabla_{\epsilon_{jk} \mathbf{u}_k}  \overline{\mathbf{f}}_j\left(\mathbf{u}_j, \mathbf{0},\mathbf{0} \right) \right]^T + \mathcal{O}\left(\epsilon^2 \right) \right\}	\label{ukc}	\\
	\mathbf{u}_k^{v \ast }(\mathbf{x}_k) &= \arg \min_{\mathbf{u}_k} \left\{ g_k\left( \mathbf{u}_k \right) +   \nabla_{\mathbf{x}_k}J_k\left(\mathbf{x}_k\right)\left[ \ \overline{\mathbf{f}}_k\left(\mathbf{u}_k, \mathbf{0},\mathbf{0} \right)  \right]^T \right\}	\label{ukv}
\end{align}
Let $u_{kn}$ be the $n$-th element of $\mathbf{u}_k$. Using Corollary \ref{corrAppd}, we have
\begin{align}
	u_{kn}^{c \ast}\left( \mathbf{x} \right) - u_{kn}^{v \ast } \left( \mathbf{x}_k \right) \leq \left| u_{kn}^{c \ast} \left( \mathbf{x} \right) - u_{kn}^{v \ast } \left( \mathbf{x}_k \right)\right|  \leq \mathcal{O}\left(\epsilon \right)  \label{policydiff}
\end{align}
where $\epsilon = \max \{\left|\epsilon_{kj}\right|: \forall k, j \in \mathcal{K},  k \neq j \}$.

Third, we obtain  the PDE defining $\widetilde{J}_{ij}\left(\mathbf{x} \right)$. Based on (\ref{transformHJB}), we have
\begin{align}
	 & (\ref{transformHJB})\overset{(e)} \Rightarrow   \sum_{k=1}^K \left(c_k\left(\mathbf{x}_k, \mathbf{u}_k^{c \ast }\left( \mathbf{x} \right) \right) + \nabla_{\mathbf{x}_k}J_k\left(\mathbf{x}_k\right)\left[ \ \overline{\mathbf{f}}_k\left(\mathbf{u}_k^{c \ast }\left( \mathbf{x} \right), \mathbf{0}, \mathbf{0}  \right)  \right]^T  + \sum_{i=1}^K \sum_{j \neq i} \epsilon_{ij} \nabla_{\mathbf{x}_k} \widetilde{J}_{ij}\left( \mathbf{x} \right) \left[ \ \overline{\mathbf{f}}_k\left(\mathbf{u}_k^{c \ast }\left( \mathbf{x} \right), \mathbf{0}, \mathbf{0}  \right) \right]^T   \right.	\qquad  \notag \\
	& \hspace{3cm}  \left.  +  \nabla_{\mathbf{x}_k}J_k \left(\mathbf{x}_k \right)   \left[   \sum_{j \neq k} \epsilon_{kj} \mathbf{u}_j^{c \ast }\left( \mathbf{x} \right)   \nabla_{\epsilon_{kj} \mathbf{u}_j}  \overline{\mathbf{f}}_k\left(\mathbf{u}_k^{c \ast }\left( \mathbf{x} \right), \mathbf{0}, \mathbf{0} \right)  \right]^T \right) + \mathcal{O}\left( \epsilon^2 \right)=0	\label{FirsTTSapp} \\
	 \overset{(f)} \Rightarrow & \sum_{k=1}^K \Bigg(\alpha_k \| \mathbf{x}_k \| _{\mathbf{v}_k,1} + g_k(\mathbf{u}_k^{v \ast }\left( \mathbf{x}_k \right))  + \left(\mathbf{u}_k^{c \ast }\left( \mathbf{x} \right) -\mathbf{u}_k^{v \ast }\left( \mathbf{x}_k \right) \right) \nabla_{\mathbf{u}_k} g_k\left( \mathbf{u}_k^{v \ast }\left( \mathbf{x}_k \right)\right)	\notag \\
	 & \hspace{2.7cm} + \nabla_{\mathbf{x}_k}J_k\left(\mathbf{x}_k\right)\left[ \ \overline{\mathbf{f}}_k\left(\mathbf{u}_k^{v \ast }\left( \mathbf{x}_k \right), \mathbf{0}, \mathbf{0} \right) + \left(\mathbf{u}_k^{c \ast }\left( \mathbf{x} \right) -\mathbf{u}_k^{v \ast }\left( \mathbf{x}_k \right) \right)  \nabla_{\mathbf{u}_k}\overline{\mathbf{f}}_k\left(\mathbf{u}_k^{v \ast }\left( \mathbf{x}_k \right), \mathbf{0}, \mathbf{0} \right)  \right]^T  \notag \\
	  & \hspace{2.7cm} + \sum_{i=1}^K \sum_{j \neq i} \epsilon_{ij} \nabla_{\mathbf{x}_k} \widetilde{J}_{ij}\left( \mathbf{x} \right) \left[  \ \overline{\mathbf{f}}_k\left(\mathbf{u}_k^{v \ast }\left( \mathbf{x}_k \right), \mathbf{0}, \mathbf{0} \right)  +   \left(\mathbf{u}_k^{c \ast }\left( \mathbf{x} \right) -\mathbf{u}_k^{v \ast }\left( \mathbf{x}_k \right) \right) \nabla_{\mathbf{u}_k} \overline{\mathbf{f}}_k\left(\mathbf{u}_k^{v \ast }\left( \mathbf{x}_k \right) , \mathbf{0}, \mathbf{0}  \right)  \right]^T   \notag \\
	  & \hspace{2.7cm} +  \nabla_{\mathbf{x}_k}J_k \left(\mathbf{x}_k \right)  \bigg[   \sum_{j \neq k} \epsilon_{kj} \left( \mathbf{u}_j^{v \ast }\left( \mathbf{x}_j \right)   \nabla_{\epsilon_{kj} \mathbf{u}_j}  \overline{\mathbf{f}}_k\left(\mathbf{u}_k^{v \ast }\left( \mathbf{x}_k \right), \mathbf{0}, \mathbf{0} \right) \right. 	\notag \\
	  & \hspace{2.7cm} +  \left. \left(\mathbf{u}_j^{c \ast }\left( \mathbf{x} \right)-\mathbf{u}_j^{v \ast }\left( \mathbf{x}_j \right)   \right)  \nabla_{\mathbf{u}_j}\left( \mathbf{u}_j^{v \ast }\left( \mathbf{x}_j \right)   \nabla_{\epsilon_{kj} \mathbf{u}_j}  \overline{\mathbf{f}}_k\left(\mathbf{u}_k^{v \ast }\left( \mathbf{x}_k \right), \mathbf{0}, \mathbf{0} \right) \right)\right) \bigg]^T \Bigg)+ \mathcal{O}\left( \epsilon^2 \right)=0	\notag \\
	  \overset{(g)}\Rightarrow & \sum_{k=1}^K \Bigg(\alpha_k \| \mathbf{x}_k \| _{\mathbf{v}_k,1} + g_k(\mathbf{u}_k^{v \ast }\left( \mathbf{x}_k \right))  + \left(\mathbf{u}_k^{c \ast }\left( \mathbf{x} \right) -\mathbf{u}_k^{v \ast }\left( \mathbf{x}_k \right) \right) \nabla_{\mathbf{u}_k} g_k\left( \mathbf{u}_k^{v \ast }\left( \mathbf{x}_k \right)\right)	\notag \\
	  & + \nabla_{\mathbf{x}_k}J_k\left(\mathbf{x}_k\right)\left[ \ \overline{\mathbf{f}}_k\left(\mathbf{u}_k^{v \ast }\left( \mathbf{x}_k \right) , \mathbf{0}, \mathbf{0}  \right) + \left(\mathbf{u}_k^{c \ast }\left( \mathbf{x} \right) -\mathbf{u}_k^{v \ast }\left( \mathbf{x}_k \right) \right)  \nabla_{\mathbf{u}_k}\overline{\mathbf{f}}_k\left(\mathbf{u}_k^{v \ast }\left( \mathbf{x}_k \right), \mathbf{0}, \mathbf{0}  \right)  \right]^T  \notag \\
	  &+ \sum_{i=1}^K \sum_{j \neq i} \epsilon_{ij} \nabla_{\mathbf{x}_k} \widetilde{J}_{ij}\left( \mathbf{x} \right) \left[ \ \overline{\mathbf{f}}_k\left(\mathbf{u}_k^{v \ast }\left( \mathbf{x}_k \right) , \mathbf{0}, \mathbf{0}  \right) \right]^T  +  \nabla_{\mathbf{x}_k}J_k \left(\mathbf{x}_k \right) \left[   \sum_{j \neq k} \epsilon_{kj} \mathbf{u}_j^{v \ast }\left( \mathbf{x}_j \right)   \nabla_{\epsilon_{kj} \mathbf{u}_j}  \overline{\mathbf{f}}_k\left(\mathbf{u}_k^{v \ast }\left( \mathbf{x}_k \right), \mathbf{0}, \mathbf{0} \right)  \right]^T \Bigg)\notag \\
	  & + \mathcal{O}\left( \epsilon^2 \right)=0	\label{taylorHJB}
\end{align}
where $(e)$ is due to the fact that $\mathbf{u}_k^{c \ast }\left(\mathbf{x} \right) $ attains the minimum  in  (\ref{transformHJB}), $(f)$ is due to the first order Taylor expansion of (\ref{FirsTTSapp}) at $\mathbf{u}_k^{c \ast}\left(\mathbf{x} \right) = \mathbf{u}_k^{v \ast}\left(\mathbf{x}_k \right)$, $(g)$ is due to $ \sum_{i=1}^K \sum_{j \neq i} \epsilon_{ij} \nabla_{\mathbf{x}_k} \widetilde{J}_{ij}\left( \mathbf{x} \right) \left[ \left(\mathbf{u}_k^{c \ast }\left( \mathbf{x} \right) -\mathbf{u}_k^{v \ast } \left(\mathbf{x}_k \right)\right)  \overline{\mathbf{f}}_k\left(\mathbf{u}_k^{v \ast } \left(\mathbf{x}_k \right), \mathbf{0}, \mathbf{0} \right)  \right]^T  \\ +  \nabla_{\mathbf{x}_k}J_k \left(\mathbf{x}_k \right)   \left[   \sum_{j \neq k} \epsilon_{kj}   \left(\mathbf{u}_j^{c \ast }\left( \mathbf{x} \right)-\mathbf{u}_j^{v \ast } \left(\mathbf{x}_j \right)  \right)\nabla_{\mathbf{u}_j}\left( \mathbf{u}_j^{v \ast } \left(\mathbf{x}_j \right)  \nabla_{\epsilon_{kj} \mathbf{u}_j}  \overline{\mathbf{f}}_k\left(\mathbf{u}_k^{v \ast }\left(\mathbf{x}_k \right), \mathbf{0}, \mathbf{0} \right) \right)  \right]^T + \mathcal{O}\left( \epsilon^2 \right)  =  \mathcal{O}\left( \epsilon^2 \right)   $  according to (\ref{policydiff}). Because $\mathbf{u}_k^{v \ast }$ is the optimal policy that achieves the minimum of the per-flow HJB equation in (\ref{perflowHJB}), we have $\alpha_k \| \mathbf{x}_k \| _{\mathbf{v}_k,1} + g_k\left(\mathbf{u}_k^{v \ast }\left( \mathbf{x}_k\right)\right)  + \nabla_{\mathbf{x}_k} J(\mathbf{x}_k)   \left[ \ \overline{\mathbf{f}}_k\left(\mathbf{u}_k^{v \ast }\left( \mathbf{x}_k\right), \mathbf{0}, \mathbf{0}\right) \right]^T  =0$. Therefore, the equation in (\ref{taylorHJB}) can be simplified as
\begin{align}
	 &\sum_{k=1}^K \Bigg(  \left(\mathbf{u}_k^{c \ast } \left(\mathbf{x} \right)-\mathbf{u}_k^{v \ast } \left(\mathbf{x}_k \right) \right)\left( \nabla_{\mathbf{u}_k} g_k\left( \mathbf{u}_k^{v \ast } \left(\mathbf{x}_k \right)\right)+ \nabla_{\mathbf{u}_k}\overline{\mathbf{f}}_k\left(\mathbf{u}_k^{v \ast } \left(\mathbf{x}_k \right), \mathbf{0}, \mathbf{0}\right) \left[ \nabla_{\mathbf{x}_k}J_k\left(\mathbf{x}_k\right) \right]^T \right) \notag \\
	  &+ \sum_{i=1}^K \sum_{j \neq i} \epsilon_{ij} \nabla_{\mathbf{x}_k} \widetilde{J}_{ij}\left( \mathbf{x} \right) \left[ \ \overline{\mathbf{f}}_k\left(\mathbf{u}_k^{v \ast } \left(\mathbf{x}_k \right), \mathbf{0}, \mathbf{0} \right) \right]^T  +  \nabla_{\mathbf{x}_k}J_k \left(\mathbf{x}_k \right) \left[   \sum_{j \neq k} \epsilon_{kj} \mathbf{u}_j^{v \ast } \left(\mathbf{x}_j \right)   \nabla_{\epsilon_{kj} \mathbf{u}_j}  \overline{\mathbf{f}}_k\left(\mathbf{u}_k^{v \ast } \left(\mathbf{x}_k \right), \mathbf{0}, \mathbf{0} \right)  \right]^T \Bigg) 	\notag \\
	 & + \mathcal{O}\left( \epsilon^2 \right)=0		\notag	\\
	 \Rightarrow & \sum_{k=1}^K \left( \left(\mathbf{u}_k^{c \ast } \left(\mathbf{x} \right)-\mathbf{u}_k^{v \ast } \left(\mathbf{x}_k \right) \right)\left( \nabla_{\mathbf{u}_k} g_k\left( \mathbf{u}_k^{v \ast } \left(\mathbf{x}_k \right)\right)+ \nabla_{\mathbf{u}_k}\overline{\mathbf{f}}_k\left(\mathbf{u}_k^{v \ast } \left(\mathbf{x}_k \right), \mathbf{0}, \mathbf{0}\right) \left[ \nabla_{\mathbf{x}_k}J_k\left(\mathbf{x}_k\right) \right]^T \right)  \right) \notag \\
	 & +  \sum_{i=1}^K \sum_{j \neq i} \epsilon_{ij}  \left(\sum_{k=1}^K \nabla_{\mathbf{x}_k} \widetilde{J}_{ij}\left( \mathbf{x} \right) \left[ \ \overline{\mathbf{f}}_k\left(\mathbf{u}_k^{v \ast } \left(\mathbf{x}_k \right), \mathbf{0}, \mathbf{0}\right) \right]^T +\nabla_{\mathbf{x}_i}J_i \left(\mathbf{x}_i \right) \left[  \mathbf{u}_j^{v \ast } \left(\mathbf{x}_j \right)   \nabla_{\epsilon_{ij} \mathbf{u}_j}  \overline{\mathbf{f}}_i\left(\mathbf{u}_i^{v \ast }\left( \mathbf{x}_i \right), \mathbf{0}, \mathbf{0} \right)  \right]^T  \right) \notag \\
	 & +  + \mathcal{O}\left( \epsilon^2 \right)=0		\notag	\\
	 \Rightarrow & \sum_{k=1}^K G_k\left(\mathbf{u}_k^{c \ast} \left(\mathbf{x} \right),\mathbf{u}_k^{v \ast}  \left(\mathbf{x}_k \right)\right)   + \sum_{i=1}^K \sum_{j \neq i} \epsilon_{ij} F_{ij}\left(\mathbf{x}, \left\{ \mathbf{u}_k^{v \ast} \right\} \right)	  + \mathcal{O}\left( \epsilon^2 \right)=0		\label{mostimp}
\end{align}
where $G_k\left(\mathbf{u}_k^{c \ast} \left(\mathbf{x} \right),\mathbf{u}_k^{v \ast}  \left(\mathbf{x}_k \right)\right) \triangleq   \left(\mathbf{u}_k^{c \ast } \left(\mathbf{x} \right)-\mathbf{u}_k^{v \ast } \left(\mathbf{x}_k \right) \right)\left( \nabla_{\mathbf{u}_k} g_k\left( \mathbf{u}_k^{v \ast } \left(\mathbf{x}_k \right)\right)+ \nabla_{\mathbf{u}_k}\overline{\mathbf{f}}_k\left(\mathbf{u}_k^{v \ast } \left(\mathbf{x}_k \right), \mathbf{0}, \mathbf{0}\right) \left[ \nabla_{\mathbf{x}_k}J_k\left(\mathbf{x}_k\right) \right]^T \right) $, $ F_{ij}\left(\mathbf{x}, \left\{ \mathbf{u}_k^{v \ast}\left(\mathbf{x}_k \right) \right\} \right) \triangleq \sum_{k=1}^K \nabla_{\mathbf{x}_k} \widetilde{J}_{ij}\left( \mathbf{x} \right) \left[ \ \overline{\mathbf{f}}_k\left(\mathbf{u}_k^{v \ast } \left(\mathbf{x}_k \right), \mathbf{0}, \mathbf{0}\right) \right]^T +\nabla_{\mathbf{x}_i}J_i \left(\mathbf{x}_i \right) \left[  \mathbf{u}_j^{v \ast } \left(\mathbf{x}_j \right)   \nabla_{\epsilon_{ij} \mathbf{u}_j}  \overline{\mathbf{f}}_i\left(\mathbf{u}_i^{v \ast }\left( \mathbf{x}_i \right), \mathbf{0}, \mathbf{0} \right)  \right]^T$. According to Corollary \ref{corrAppd}, we have that either $G_k\left(\mathbf{u}_k^{c \ast} \left(\mathbf{x} \right),\mathbf{u}_k^{v \ast}  \left(\mathbf{x}_k \right)\right) =0$ or $G_k\left(\mathbf{u}_k^{c \ast} \left(\mathbf{x} \right),\mathbf{u}_k^{v \ast}  \left(\mathbf{x}_k \right)\right) = \sum_{i \geq 0} \epsilon^{1+\delta_i} \tilde{f}_i \left( \left\{\mathbf{u}_k\left( \widetilde{\mathbf{x}}_k^{v \ast} \right) \right\}, B_k\right) $ for some function $\tilde{f}_i $ and positive constants $\{\delta_i: i \geq 0\}$, where $0 < \delta_0 \leq \delta_1 \leq \delta_2 \leq \dots \leq \infty $, $\mathbf{u}_k\left( \widetilde{\mathbf{x}}_k^{v \ast} \right) $ is the virtual  action that achieves the minimum in (\ref{ukv}) when $\mathcal{U}_k=\mathbb{R}$, $B_k$ is the boundary of the sub-system action space $\mathcal{U}_k$.
We then discuss the equation in (\ref{mostimp}) in the following two cases:
\begin{enumerate}
	\item  $G_k\left(\mathbf{u}_k^{c \ast} \left(\mathbf{x} \right),\mathbf{u}_k^{v \ast}  \left(\mathbf{x}_k \right)\right) =0$:  Note that $\mathcal{O}\left( \epsilon^2 \right)$ in (\ref{mostimp}) represents $\sum_{i \geq 0} \epsilon^{2+\delta_i'} \tilde{g}_i \left(\mathbf{x},\left\{ \mathbf{u}_k^{v \ast} \right\} \right)$ for some function $\tilde{g}_i$ where $0=\delta_0' \leq \delta_1' \leq \delta_2' \leq \dots \leq \infty$.  In this case, the equation in (\ref{mostimp}) can be written as
	\begin{align}
	  \sum_{i=1}^K \sum_{j \neq i} \epsilon_{ij}	F_{ij}\left(\mathbf{x}, \left\{ \mathbf{u}_k^{v \ast} \right\} \right)  + \sum_{i \geq 0} \epsilon^{2+\delta_i'} \tilde{g}_i \left(\mathbf{x}, \left\{ \mathbf{u}_k^{v \ast} \right\} \right)=0		\label{caseapp111}
	\end{align}
	  In order for the  equation in (\ref{caseapp111}) to hold for any coupling parameter $\epsilon_{ij}$, we have $F_{ij}\left(\mathbf{x}, \left\{ \mathbf{u}_k^{v \ast} \right\} \right)=0$ $(\forall i, j)$ and $\tilde{g}_i \left(\mathbf{x},\left\{ \mathbf{u}_k^{v \ast} \right\} \right)=0$ $(\forall i)$.
	\item  $G_k\left(\mathbf{u}_k^{c \ast} \left(\mathbf{x} \right),\mathbf{u}_k^{v \ast}  \left(\mathbf{x}_k \right)\right) = \sum_{i \geq 0} \epsilon^{1+\delta_i} \tilde{f}_i \left( \left\{\mathbf{u}_k\left( \widetilde{\mathbf{x}}_k^{v \ast} \right) \right\}, B_k\right)$: in this case, the equation in (\ref{mostimp}) can be written as
	\begin{align}
	  \sum_{k=1}^K \sum_{i \geq 0} \epsilon^{1+\delta_i} \tilde{f}_i \left( \left\{\mathbf{u}_k\left( \widetilde{\mathbf{x}}_k^{v \ast} \right) \right\}, B_k\right) + \sum_{i=1}^K \sum_{j \neq i} \epsilon_{ij}	F_{ij}\left(\mathbf{x}, \left\{ \mathbf{u}_k^{v \ast} \right\} \right)  + \sum_{i \geq 0} \epsilon^{2+\delta_i'} \tilde{g}_i \left(\mathbf{x}, \left\{ \mathbf{u}_k^{v \ast} \right\} \right)=0		\label{caseapp222}
	\end{align}
	In order for the  equation in (\ref{caseapp222}) to hold for any coupling parameter $\epsilon_{ij}$, we have $\tilde{f}_i \left( \left\{\mathbf{u}_k\left( \widetilde{\mathbf{x}}_k^{v \ast} \right) \right\}, B_k\right)=0$ $(\forall i,k)$, $F_{ij}\left(\mathbf{x}, \left\{ \mathbf{u}_k^{v \ast} \right\} \right)=0$ $(\forall i, j)$ and $\tilde{g}_i \left(\mathbf{x}, \left\{ \mathbf{u}_k^{v \ast} \right\} \right)=0$ $(\forall i)$.
\end{enumerate}
Therefore, based on the analysis in the above two cases, we conclude that in order for the  equation in (\ref{mostimp}) to hold for any coupling parameter $\epsilon_{ij}$,  we need $F_{ij}\left(\mathbf{x}, \left\{ \mathbf{u}_k^{v \ast} \right\} \right)=0$ $(\forall i, j)$, i.e., 
\begin{align}
	\sum_{k=1}^K \nabla_{\mathbf{x}_k} \widetilde{J}_{ij}\left( \mathbf{x} \right) \left[ \ \overline{\mathbf{f}}_k\left(\mathbf{u}_k^{v \ast }\left(\mathbf{x}_k \right), \mathbf{0}, \mathbf{0} \right) \right]^T +\nabla_{\mathbf{x}_i}J_i \left(\mathbf{x}_i \right) \left[  \mathbf{u}_j^{v \ast }\left(\mathbf{x}_j \right)   \nabla_{\epsilon_{ij} \mathbf{u}_j}  \overline{\mathbf{f}}_i\left(\mathbf{u}_i^{v \ast }\left(\mathbf{x}_i \right), \mathbf{0}, \mathbf{0} \right)  \right]^T =0	\label{PDEJijDef}
\end{align}

Finally, we obtain the boundary condition of the PDE in (\ref{PDEJijDef}).  Replacing $j$ with $k$ and letting $\mathbf{x}_j = \mathbf{0}$ in (\ref{approxJxx}), we have
\begin{align}		
	J\left(\mathbf{x} ; \boldsymbol{\epsilon}\right)\Big|_{\mathbf{x}_j=\mathbf{0}} &= \sum_{i=1}^K J_i\left(\mathbf{x}_i\right) + \sum_{i=1}^K \sum_{k \neq i} \epsilon_{ik} \widetilde{J}_{ik}\left( \mathbf{x}\right)\Big|_{\mathbf{x}_j=\mathbf{0}} + \mathcal{O}\left( \epsilon^2 \right)	\notag \\
	&\overset{(h)}= \sum_{i=1,  i \neq j}^K J_i\left(\mathbf{x}_i\right) + \sum_{i=1}^K \sum_{k \neq i} \epsilon_{ik} \widetilde{J}_{ik}\left( \mathbf{x}\right)\Big|_{\mathbf{x}_j=\mathbf{0}} + \mathcal{O}\left( \epsilon^2 \right)	\label{ohundequ}
\end{align}
where $(h)$ is due to $J_j\left(\mathbf{0} \right)=0$. According to the definition in (\ref{totalU}), $J\left(\mathbf{x} ; \boldsymbol{\epsilon}\right)\Big|_{\mathbf{x}_j=\mathbf{0}} $ is the optimal total cost when the initial  system state is $\mathbf{x}_i\left(0\right)=\mathbf{x}_i$ ($\forall i \neq j$) and $\mathbf{x}_j\left(0 \right)=\mathbf{0}$. At this initial condition, the $j$-th sub-system stays at the initial zero state to maintain stability. Therefore, when $\mathbf{x}_j = \mathbf{0}$, the original global $K$ dimensional system is equivalent to  a virtual $(K-1)$ dimensional system with global system state being $\left(\mathbf{x}_1, \dots, \mathbf{x}_{j-1},\mathbf{x}_{j+1},\dots,\mathbf{x}_K \right)$. We use $J^v \left(\mathbf{x}_1, \dots, \mathbf{x}_{j-1},\mathbf{x}_{j+1},\dots,\mathbf{x}_K;  \boldsymbol{\epsilon} \right)$ to denote the optimal total cost for the virtual $(K-1)$ dimensional system and hence, we have  
\begin{align}
	J\left(\mathbf{x} ; \boldsymbol{\epsilon}\right)\Big|_{\mathbf{x}_j=\mathbf{0}}=J^v \left(\mathbf{x}_1, \dots, \mathbf{x}_{j-1},\mathbf{x}_{j+1},\dots,\mathbf{x}_K; \boldsymbol{\epsilon} \right)	\label{equbase2}
\end{align}
Furthermore,  similar as (\ref{approxJxx}), we have
\begin{align}			\label{boundequ}
	  & J^v \left(\mathbf{x}_1, \dots, \mathbf{x}_{j-1},\mathbf{x}_{j+1},\dots,\mathbf{x}_K; \boldsymbol{\epsilon} \right) 	\notag \\
	 =  & \sum_{i=1,  i \neq j}^K J_i\left(\mathbf{x}_i\right) + \sum_{i=1, i \neq j}^K \sum_{k \neq i, j} \epsilon_{ik} \widetilde{J}_{ik}^v\left( \mathbf{x}_1, \dots, \mathbf{x}_{j-1},\mathbf{x}_{j+1},\dots,\mathbf{x}_K\right)+ \mathcal{O}\left( \epsilon^2 \right)
\end{align}
where we denote $\widetilde{J}_{ik}^v\left( \mathbf{x}_1, \dots, \mathbf{x}_{j-1},\mathbf{x}_{j+1},\dots,\mathbf{x}_K\right) \triangleq \frac{\partial J^v \left(\mathbf{x}_1, \dots, \mathbf{x}_{j-1},\mathbf{x}_{j+1},\dots,\mathbf{x}_K; \boldsymbol{\epsilon} \right)}{\partial \epsilon_{ik}} \Big|_{\boldsymbol{\epsilon}=\mathbf{0}}$. Based on (\ref{ohundequ}) and (\ref{boundequ}), we have
\begin{align}
	0 &\overset{(i)}= J\left(\mathbf{x} ; \boldsymbol{\epsilon}\right)\Big|_{\mathbf{x}_j=\mathbf{0}} - J^v \left(\mathbf{x}_1, \dots, \mathbf{x}_{j-1},\mathbf{x}_{j+1},\dots,\mathbf{x}_K; \boldsymbol{\epsilon} \right) \notag \\
	& = \text{R.H.S. of } (\ref{ohundequ}) - \text{R.H.S. of } (\ref{boundequ})	\notag \\
	& = \sum_{i=1}^K \sum_{k \neq i} \epsilon_{ik} \widetilde{J}_{ik}\left( \mathbf{x}\right)\Big|_{\mathbf{x}_j=\mathbf{0}} -  \sum_{i=1, i \neq j}^K \sum_{k \neq i, j} \epsilon_{ik} \widetilde{J}_{ik}^v\left( \mathbf{x}_1, \dots, \mathbf{x}_{j-1},\mathbf{x}_{j+1},\dots,\mathbf{x}_K\right)+ \mathcal{O}\left( \epsilon^2 \right)	\notag \\
	&\overset{(j)} = \sum_{i=1, i \neq j}^K	  \epsilon_{ij}  \widetilde{J}_{ij}\left( \mathbf{x}\right)\Big|_{\mathbf{x}_j=\mathbf{0}} + \sum_{i=1, i \neq j}^K	  \epsilon_{ji}  \widetilde{J}_{ji}\left( \mathbf{x}\right)\Big|_{\mathbf{x}_j=\mathbf{0}}   + \mathcal{O}\left( \epsilon^2 \right)	\label{boundequuu}
\end{align}
where $(i)$ is due to (\ref{equbase2}), $(j)$ is due to $\widetilde{J}_{ik}\left( \mathbf{x}\right)\big|_{\mathbf{x}_j=\mathbf{0}} = \frac{\partial J(\mathbf{x}; \boldsymbol{\epsilon})}{\partial \epsilon_{ik}}\big|_{\boldsymbol{\epsilon}=\mathbf{0},\mathbf{x}_j=\mathbf{0}} = \frac{\partial J(\mathbf{x}; \boldsymbol{\epsilon})|_{\mathbf{x}_j=\mathbf{0}}}{\partial \epsilon_{ik}}\big|_{\boldsymbol{\epsilon}=\mathbf{0}}= \\ \frac{\partial J^v \left(\mathbf{x}_1, \dots, \mathbf{x}_{j-1},\mathbf{x}_{j+1},\dots,\mathbf{x}_K; \boldsymbol{\epsilon} \right)}{\partial \epsilon_{ik}}\big|_{\boldsymbol{\epsilon}=\mathbf{0}}=\widetilde{J}_{ik}^v\left( \mathbf{x}_1, \dots, \mathbf{x}_{j-1},\mathbf{x}_{j+1},\dots,\mathbf{x}_K\right)$  ($\forall i\neq j,k \neq j$). In order for (\ref{boundequuu}) to hold  for any coupling parameter $\epsilon_{ij}$, we have $\widetilde{J}_{ij}\left( \mathbf{x}\right)\Big|_{\mathbf{x}_j=\mathbf{0}}=0$ ($\forall i \neq j$) and $ \widetilde{J}_{ji}\left( \mathbf{x}\right)\Big|_{\mathbf{x}_j=\mathbf{0}}=0$ ($\forall i \neq j$). Therefore, the boundary condition for the PDE that defines $\widetilde{J}_{ij}\left( \mathbf{x}\right)$ in (\ref{PDEJijDef}) is given by
\begin{align}			\label{boundequ1}
	  \widetilde{J}_{ij}\left( \mathbf{x}\right)\Big|_{\mathbf{x}_j=\mathbf{0}}=0
\end{align}
According to the \emph{transversality condition} \cite{IntroPDE},  the PDE in (\ref{PDEJijDef}) with boundary condition in (\ref{boundequ1}) has a unique solution.  Replacing the subscript $i$ with $k$ and $k$ with $i$ in (\ref{approxJthm}), (\ref{PDEJijDef}) and (\ref{boundequ1}), we obtain the result in Theorem \ref{RelashipJsJ}.

\section*{Appendix E: Proof of Lemma \ref{NUMprob}}	
Using the first order Taylor expansion of $V\left(\mathbf{x}\left(t+1 \right) \right)$ at $\mathbf{x}\left(t+1 \right)=\mathbf{x}$ (with $\mathbf{x}\left(t  \right)=\mathbf{x}$), minimizing  the R.H.S. of the Bellman equation in  (\ref{OrgBel}) is equivalent to
\begin{align}
	\text{R.H.S. of } (\ref{OrgBel}) &  \Leftrightarrow  \min_{ \mathbf{u }} \left[ \sum_{k=1}^K  g_k\left( \mathbf{u}_k \right) + \mathbb{E} \left[ \sum_{n=1}^\infty  \sum_{k=1}^K   \frac{\nabla_{\mathbf{x}_k}^{\left(n\right)} V_k \left(\mathbf{x}_k\right)}{n!} \left[   \left[ \mathbf{f}_k\left(\mathbf{u}_k, \mathbf{u}_{-k},\boldsymbol{\epsilon}_k\right) + \mathbf{z}_k \right]^{\left(n\right)}  \right]^T \right] \Bigg| \ \mathbf{x}_k, \mathbf{u}  \right]		\notag \\
	& =  \min_{ \mathbf{u }} \left[ \sum_{k=1}^K  g_k\left( \mathbf{u}_k \right) + \sum_{n=1}^\infty  \sum_{k=1}^K   \frac{\nabla_{\mathbf{x}_k}^{\left(n\right)} V_k \left(\mathbf{x}_k\right)}{n!} \left[  \mathbb{E} \left[ \left[ \mathbf{f}_k\left(\mathbf{u}_k, \mathbf{u}_{-k}, \boldsymbol{\epsilon}_k\right) + \mathbf{z}_k \right]^{\left(n\right)} \big| \ \mathbf{x}_k, \mathbf{u}  \right] \right]^T  \right]		\label{appendD1}
\end{align}
where $a \Leftrightarrow b$ means that $a$ is equivalent to $b$.
Using the per-flow fluid value function approximation in (\ref{approxVJ}), the above equation in (\ref{appendD1}) becomes
\begin{align}
	\text{R.H.S. of } (\ref{OrgBel}) & \overset{(a)} \Leftrightarrow \min_{ \mathbf{u }} \Bigg[ \sum_{k=1}^K \underbrace{ \left(  g_k\left( \mathbf{u}_k \right) + \sum_{n=1}^\infty     \frac{\nabla_{\mathbf{x}_k}^{\left(n\right)} J_k \left(\mathbf{x}_k\right)}{n!} \left[  \mathbb{E} \left[ \left[ \mathbf{f}_k\left(\mathbf{u}_k, \mathbf{u}_{-k},\boldsymbol{\epsilon}_k\right) + \mathbf{z}_k \right]^{\left(n\right)} \big| \ \mathbf{x}_k, \mathbf{u}  \right] \right]^T \right) }_{-U_k\left(\mathbf{u}_k,   \mathbf{u}_{-k}, \boldsymbol{\epsilon}_k  \right)}\Bigg]	\notag \\
	& \Leftrightarrow \max_{ \mathbf{u} } \sum_{k=1}^K  U_k\left(\mathbf{u}_k,   \mathbf{u}_{-k},\boldsymbol{\epsilon}_k  \right)
\end{align}
where $a$ is due to $\nabla_{\mathbf{x}_k} V_k \left(\mathbf{x}_k\right)  = \nabla_{\mathbf{x}_k} J_k\left(\mathbf{x}_k\right)$ under per-flow fluid value function approximation.  This proves the lemma.

\section*{Appendix F: Proof of Corollary \ref{collaryAlg}}	
	Under Assumption \ref{assumfg}, as the coupling parameter $\boldsymbol{\epsilon}$ goes to zero, the sum utility  $\sum_{k=1}^K  U_k\left(\mathbf{u}_k,   \mathbf{u}_{-k}, \boldsymbol{\epsilon} \right)$ of the NUM problem in (\ref{utility}) becomes asymptotically strictly concave in  control variable $\mathbf{u}$. Therefore, when $\boldsymbol{\epsilon}=\mathbf{0}$, the NUM problem in (\ref{utility}) is a  strictly concave maximization and hence it has a unique global optimal point.   Then according to Theorem \ref{conditionF}, when $\boldsymbol{\epsilon}=\mathbf{0}$ the limiting point $\mathbf{u}(\infty)$ of algorithm \ref{distgen} is the unique global optimal point of the NUM problem in (\ref{utility}).

\section*{Appendix G: Proof of Lemma \ref{Eg1Per}}	
Based on Lemma \ref{linearAp}, the per-flow fluid value function $J_k \left( q_k \right)$ of the $k$-th Tx-Rx pair is given by the solution of the following  per-flow HJB equation:
\begin{equation}			\label{perfloweg11}
		\min_{\mathbf{p}_k} \mathbb{E} \left[ \beta_k \frac{q_k}{\lambda_k}   + \gamma_k p_k^{\mathbf{H}} -c_k^{\infty} + J_k'\left(q_k \right)  \left( \lambda_k  - \log\left(1 + p_k^{\mathbf{H}} L_{kk} |H_{kk}|^2\right)  \tau \right)  \bigg| q_k\right] =0, \quad q_k \in \mathcal{Q}
	\end{equation}
The optimal policy that attains the minimum in (\ref{perfloweg11}) is given by
\begin{equation}	\label{poweropt}
	p_k^{\mathbf{H} \ast} = \bigg( \frac{J_k'\left(q_k \right)\tau}{\gamma_k} - \frac{1}{L_{kk} |H_{kk}|^2}  \bigg)^+
\end{equation}

Based on (\ref{perfloweg11}) and (\ref{poweropt}), the per-flow HJB equation can be transformed into the following ODE:
\begin{small}
\begin{equation}			\label{transformperf}
	 \beta_k \frac{q_k}{\lambda_k}  - c_k^{\infty} +  \mathbb{E} \left[ \left( J_k'\left(q_k \right)\tau  - \frac{\gamma_k }{L_{kk} |H_{kk}|^2}  \right)^+ \Bigg| q_k \right]  + J_k'\left(q_k \right)  \left[ \lambda_k -  \mathbb{E} \left[ \left(\log\left(\frac{J_k'\left(q_k \right)\tau L_{kk} |H_{kk}|^2}{\gamma_k}\right) \right)^+  \Bigg| q_k \right]  \tau \right]  =0	
\end{equation} 
\end{small}We next calculate $c_k^{\infty}$. Since $c_k^{\infty}$ satisfies the sufficient conditions in (\ref{suffcond11}), we have 
\begin{align}
	E_1\left(\frac{\gamma_k}{\tau L_{kk} J_k'\left(0 \right)}\right)\tau=\lambda_k, \quad J_k'\left(0 \right)\tau e^{- \frac{\gamma_k}{ L_{kk} J_k'\left(0 \right)\tau}} - \frac{\gamma_k}{L_{kk}}  E_1\left(\frac{\gamma_k}{ L_{kk} J_k'\left(0\right)\tau}\right) = c_k^{\infty}
\end{align}
Therefore, $c_k^{\infty}=v_k\tau e^{- \frac{\gamma_k}{ L_{kk} v_k\tau}} - \frac{\gamma_k}{L_{kk}}  E_1\left(\frac{\gamma_k}{ L_{kk} v_k\tau}\right)$ with $v_k$ satisfying $ E_1\left(\frac{\gamma_k}{\tau L_{kk} v_k}\right)\tau=\lambda_k$. 

To solve the ODE in (\ref{transformperf}), we need to calculate the two terms involving expectation operator. Since $H_{kk} \sim \mathcal{CN}(0,1)$, we have  $|H_{kk}|^2 \sim \exp(1)$.  Then, we have
\begin{align}
	&\mathbb{E} \left[  \left( J_k'\left(q_k \right)\tau - \frac{\gamma_k }{L_{kk} |H_{kk}|^2}  \right)^+  \Bigg| q_k \right] = \int_{\frac{\gamma_k}{ L_{kk} J_k'\left(q_k \right)\tau}}^{\infty}  \left( J_k'\left(q_k \right)\tau - \frac{\gamma_k }{L_{kk}x} \right)e^{-x}  \mathrm{d}x				\notag \\
	= & J_k'\left(q_k \right)\tau e^{- \frac{\gamma_k}{ L_{kk} J_k'\left(q_k \right)\tau}} - \frac{\gamma_k}{L_{kk}}  E_1\left(\frac{\gamma_k}{ L_{kk} J_k'\left(q_k \right)\tau}\right)	\label{expec1}	\\
	&\mathbb{E} \left[ \left(\log\left(  \frac{J_k'\left(q_k \right)\tau  L_{kk} |H_{kk}|^2}{\gamma_k}\right) \right)^+  \Bigg| q_k\right]	 =  \int_{\frac{\gamma_k}{\tau L_{kk} J_k'\left(q_k \right)}}^{\infty}  \log\left(\frac{J_k'\left(q_k \right)\tau L_{kk} x}{\gamma_k}\right)e^{-x}  \mathrm{d}x  \notag    \\
	=& \int_{\frac{\gamma_k}{\tau L_{kk} J_k'\left(q_k \right)}}^{\infty}  \log\left(\frac{J_k'\left(q_k \right)\tau L_{kk} }{\gamma_k}\right)e^{-x}   + \log \left(x \right) e^{-x}  \mathrm{d}x  \notag    \\
	=& \left(\log\left(\frac{\gamma_k}{\tau L_{kk}J_k'\left(q_k \right)}\right) + \log\left(\frac{\tau L_{kk}J_k'\left(q_k \right)}{\gamma_k}\right) \right)e^{- \frac{\gamma_k}{\tau L_{kk} J_k'\left(q_k \right)}}  +  E_1\left(\frac{\gamma_k}{\tau L_{kk} J_k'\left(q_k \right)}\right)	\notag	\\
	= & E_1\left(\frac{\gamma_k}{\tau L_{kk} J_k'\left(q_k \right)}\right)	\label{expec2}
\end{align}
where $E_1(z)  \triangleq \int_1^{\infty} \frac{e^{-tz}}{t}\mathrm{d}t = \int_z^{\infty} \frac{e^{-t}}{t}\mathrm{d}t $ is the exponential integral function. Substituting (\ref{expec1}) and (\ref{expec2}) into (\ref{transformperf}) and letting $a_k=\frac{\tau L_{kk}}{\gamma_k}$, we have
\begin{align}		
	 \beta_k \frac{q_k}{\lambda_k \tau}-\frac{c_k^{\infty}}{\tau}   +   J_k'\left(q_k \right) e^{- \frac{1}{a_k J_k'\left(q_k \right)}} -  \frac{1}{a_k}  E_1\left(\frac{1}{a_k J_k'\left(q_k \right)}\right) + J_k'\left(q_k \right)  \left( \frac{\lambda_k}{\tau}   -  E_1\left(\frac{1}{a_k J_k'\left(q_k \right)}\right)  \right)   =0		\label{transformODEapp}
\end{align} 
According to \cite{HandODE}, the parametric solution (w.r.t. $y$) of the ODE in (\ref{transformODEapp}) is given below
\begin{equation}		\label{appendjk}
 \left\{
	\begin{aligned}	 
		q_k(y) &=  \frac{\lambda_k \tau}{\beta_k} \left(  \left(\frac{1}{a_k}+y\right)E_1\left(\frac{1}{a_k y}\right) - \frac{\lambda_k}{\tau}y   -  e^{- \frac{1}{a_k y}} y + \frac{c_k^{\infty}}{\tau}\right)  		 \\
		J_k(y) &=   \frac{\lambda_k \tau}{\beta_k} \left( \frac{\left(1-a_k y\right)}{4a_k}ye^{-\frac{1}{a_k y}} -\frac{\lambda_k }{2 \tau}y^2 +\left(\frac{y^2}{2}-\frac{1}{4a_k^2}\right) E_1\left(\frac{1}{a_k y}\right) \right)  + b_k
	   \end{aligned}
   \right.
  \end{equation}
where $b_k$ is chosen such that the boundary condition $J(0)=0$ is satisfied. To find $b_k$, first find $y_k^0$  such that  $q_k(y_k^0)=0$. Then $b_k$ is chosen such that $J_k(y_k^0)=0$.

\section*{Appendix H: Proof of Corollary \ref{PropJk}}	
First, we obtain the highest order term of $J_k\left(q_k \right)$. The  series expansions of the exponential integral function and exponential function are given below
\begin{align}	\label{expansopmSs}
	E_1(x) = -\gamma_{eu} - \log x - \sum_{n=1}^{\infty} \frac{\left(-x \right)^n}{n ! n}, \quad e^x &=\sum_{n=0}^{\infty} \frac{x^n}{n!}
\end{align}
Then $q_k(y)$ in (\ref{pareJkEg1}) can be written as
\begin{align}
	q_k(y) =  \frac{\lambda_k \tau}{\beta_k} \left(  \left(\frac{1}{a_k}+y\right)\left( -\gamma_{eu} + \log \left(a_k y \right) - \sum_{n=1}^{\infty} \frac{\left(-1\right)^n}{n ! n \left(a_k y \right)^n} \right) - \frac{\lambda_k y}{\tau}   -  \left(\sum_{n=0}^{\infty} \frac{\left(-1 \right)^n}{n!y^n}  \right) y \right)  		\label{expand}
\end{align}
By expanding the above equation, it can be seen that $q_k(y)= \mathcal{O} \left(y \log y \right)$ as $y \rightarrow \infty$. In other words, there exist finite positive constants $C_1$ and $C_2$, such that for sufficiently large $y$, 
\begin{align}	
	& C_2 y \log y \leq q_k \left(y  \right)  \leq C_1 y \log y	\label{order1}		\\
	\Rightarrow & \frac{q_k/C_1}{W\left( q_k/C_1 \right)} \leq y \leq \frac{q_k/C_2}{W\left( q_k/C_2 \right)}	\label{order2111}
\end{align}
where $W$ is the \emph{Lambert} function.  Again, using the  series expansions in (\ref{expansopmSs}), $J_k(y)$ in (\ref{pareJkEg1})  has a similar  property:  there exist finite positive constants $C_1'$ and $C_2'$, such that for sufficiently large $y$,
\begin{align}	\label{order2}
	C_2' y^2 \log y \leq J_k \left(y  \right) \leq C_1' y^2 \log y
\end{align}
Combining (\ref{order2111}) and (\ref{order2}), we can improve the upper bound in (\ref{order2}) as
\begin{align}
	& C_1' y^2 \log y  \leq C_1' \frac{\left( q_k/C_2 \right)^2}{W^2\left( q_k/C_2 \right)} \left( \log\left( q_k/C_2\right)  - \log \left( W\left( q_k/C_2 \right) \right) \right) \notag \\
	\overset{(a)}\leq & C_1' \frac{\left( q_k/C_2 \right)^2}{W\left( q_k/C_2 \right)}  \overset{(b)}\leq C_1' \frac{\left( q_k/C_2 \right)^2}{\log\left( q_k/C_2 \right)-\log\log\left( q_k/C_2 \right)}  \leq \overline{C}_1 \frac{q_k^2}{\log \left( q_k \right)}	\label{upperq21}
\end{align}
for some positive constant $ \overline{C}_1$. $(a)$ is due to $W \left( x\right) = \log x - \log\left(W\left( x\right) \right)$ and $(b)$ is due to $\log x - \log \log x \leq W\left(x\right) \leq \log x$ ($x > e$) \cite{lambert}. Similarly, the lower bound in (\ref{order2}) can be improved as
\begin{align}
	C_2' y^2 \log y  \geq  C_2' \frac{\left( q_k/C_1 \right)^2}{W\left( q_k/C_1 \right)} \geq C_2' \frac{\left( q_k/C_1 \right)^2}{\log\left( q_k/C_1 \right)} \geq \overline{C}_2 \frac{q_k^2}{\log \left( q_k \right)}	\label{upperq22}
\end{align}
for some positive constant $\overline{C}_2$. Therefore, based on (\ref{upperq21}) and (\ref{upperq22}), we have
\begin{align}	
	& \overline{C}_2 \frac{q_k^2}{\log \left( q_k \right)}  \leq  J_k\left(q_k \right) \leq \overline{C}_1 \frac{q_k^2}{\log \left( q_k \right)} \label{orderJkapppre} \\
	\Rightarrow & J_k\left(q_k \right) = \mathcal{O} \left(\frac{q_k^2}{\log \left( q_k \right)}  \right),\quad  \text{as } q_k \rightarrow \infty		\label{orderJkapp}
\end{align}

Next, we  obtain the coefficient of the highest order term $\frac{q_k^2}{\log \left( q_k \right)}$. Again, using the series expansion in (\ref{expansopmSs}), the per-flow HJB equation in (\ref{transformODEapp}) can be written as
\begin{align}
	 \beta_k \frac{q_k}{\lambda_k \tau}   +   J_k'\left(q_k \right) & \left( \sum_{n=0}^{\infty} \frac{\left(-1 \right)^n}{n! \left( a_k J_k'\left(q_k \right) \right)^n} \right)  -  \frac{1}{a_k}  \left( -\gamma_{eu} + \log \left( a_k J_k'\left(q_k \right) \right) - \sum_{n=1}^{\infty} \frac{\left(-1 \right)^n}{n ! n \left( a_k J_k'\left(q_k \right) \right)^n}  \right)  \notag \\
	 &+ J_k'\left(q_k \right)  \left[ \frac{\lambda_k}{\tau}   -  \left(- \gamma_{eu} + \log \left( a_k J_k'\left(q_k \right) \right) - \sum_{n=1}^{\infty} \frac{\left(-1 \right)^n}{n ! n \left( a_k J_k'\left(q_k \right) \right)^n}  \right) \right]   =0		\label{orderqk}
\end{align}
According to the asymptotic property of $J_k\left(q_k \right)$ in (\ref{orderJkapp}), based on the  ODE in (\ref{orderqk}), we have  
\begin{align}	\label{simpJlogJ}
	J_k'\left(q_k \right) \log \left( J_k'\left(q_k \right) \right) = \frac{\beta_k}{\lambda_k \tau}\mathcal{O} \left(q_k\right),\quad  \text{as } q_k \rightarrow \infty
\end{align}
Furthermore, from (\ref{orderJkapppre}),  we have
\begin{align}
	 & 	\overline{C}_2 \frac{2q_k \log \left( q_k \right)-q_k}{\left(\log \left( q_k \right)\right)^2}	\leq J_k'\left(q_k \right) \leq \overline{C}_1 \frac{2q_k \log \left( q_k \right)-q_k}{\left(\log \left( q_k \right)\right)^2}	\notag \\
	\Rightarrow &  \overline{B}_2 \overline{C}_2 \frac{q_k}{\log \left( q_k \right)}	\leq J_k'\left(q_k \right) \leq \overline{B}_1\overline{C}_1 \frac{q_k}{\log \left( q_k \right)}	\label{kprimeorder}		\\
	\Rightarrow & \log \left(  \overline{B}_2 \overline{C}_2 \right) +  \log \left( q_k \right) - \log \log \left( q_k \right)	\leq \log \left( J_k'\left(q_k \right) \right) \leq \log \left( \overline{B}_1\overline{C}_1 \right) + \log \left( q_k \right) - \log \log \left( q_k \right)		\notag \\
	\Rightarrow & \overline{B}_2\overline{C}_2  q_k + o\left( q_k\right) \leq J_k'\left(q_k \right) \log \left( J_k'\left(q_k \right) \right) \leq \overline{B}_1\overline{C}_1  q_k + o\left( q_k\right),\quad  \text{as } q_k \rightarrow \infty	\label{jlogjorder}
\end{align}
where  $\overline{B}_1$ and $\overline{B}_2$  are some constants that are independent of  system parameters.  Comparing (\ref{jlogjorder}) with (\ref{simpJlogJ}), we have 
\begin{align}	
	& \overline{B}_1\overline{C}_1  \propto  \frac{\beta_k}{\lambda_k \tau}, \quad \overline{B}_2\overline{C}_2 \propto  \frac{\beta_k}{\lambda_k \tau}		\label{bcbcprop}	\\
	\overset{(c)} \Rightarrow &  \ \overline{C}_1  \propto  \frac{\beta_k}{\lambda_k \tau}, \hspace{0.5cm} \overline{C}_2 \propto  \frac{\beta_k}{\lambda_k \tau}	\label{c1c2prop}
\end{align}
where $x \propto y$ means that $x$ is proportional to $y$. $(c)$ is due to the fact that  $\overline{B}_1$ and $\overline{B}_2$  are independent of system parameters. Finally, based on (\ref{bcbcprop}) and (\ref{c1c2prop}), we conclude
\begin{align}
	& J_k\left(q_k \right) = \frac{\beta_k}{\lambda_k \tau}\mathcal{O}\left(\frac{q_k^2}{\log \left( q_k \right)} \right),\quad  \text{as } q_k \rightarrow \infty	\\
	& J_k'\left(q_k \right) = \frac{\beta_k}{\lambda_k \tau}\mathcal{O}\left(\frac{q_k}{\log \left( q_k \right)} \right),\quad  \text{as } q_k \rightarrow \infty	
\end{align}
This completes the proof.

\section*{Appendix I: Proof of Lemma \ref{ErrorEg1}}	
We use the linear architecture $\sum_{k=1}^K J_k\left(q_k\right)$ to approximate $J\left(\mathbf{q}; \mathbf{L}\right)$, i.e., $J\left(\mathbf{q}; \mathbf{L}\right) \approx \sum_{k=1}^K J_k\left(q_k\right)$. According to Theorem \ref{RelashipJsJ}, the approximation error is given by 
\begin{align}	\label{Jijapprox}
	\left| J\left(\mathbf{q}; \mathbf{L}\right) - \sum_{k=1}^K J_k\left(q_k\right) \right| =\sum_{k=1}^K \sum_{j \neq k} L_{kj}  \left| \widetilde{J}_{kj}\left(\mathbf{q}\right) \right|  + \mathcal{O}(L^2), \quad \text{as } L \rightarrow 0
\end{align}
$\widetilde{J}_{kj}\left(\mathbf{q}\right)$ is the solution of the following first order PDE,
\begin{align}	\label{PDEappd}
	\sum_{i=1}^K \left(\lambda_i  - \mathbb{E} \left[  \log\left(1 + p_i^{\mathbf{H} \ast} L_{ii} |H_{ii}|^2  \right) | \mathbf{q}  \right]\tau  \right) \frac{\mathrm{d} \widetilde{J}_{kj}\left(\mathbf{q}\right) }{\mathrm{d} q_i} -  \mathbb{E} \left[  \frac{J_k'\left(q_k\right)p_k^{\mathbf{H}\ast}L_{kk}|H_{kk}|^2}{1+p_k^{\mathbf{H}\ast} L_{kk}|H_{kk}|^2} p_j^{\mathbf{H}\ast} |H_{kj}|^2 \Bigg| \mathbf{q}  \right] =0
\end{align}
where $p_i^{\mathbf{H}\ast}$ and  $J_i\left(q_i\right)$ are given in (\ref{perfloweg11}) and (\ref{appendjk}), respectively. Next, we calculate the two terms involving expectation operator  in (\ref{PDEappd}). First, we have $\mathbb{E} \left[  \log\left(1 + p_i^{\mathbf{H}} L_{ii} |H_{ii}|^2  \right) | \mathbf{q}  \right] = E_1\left(\frac{\gamma_i}{\tau L_{ii} J_i'\left(q_i \right)}\right)$ according to (\ref{expec2}). Using the  series expansions in (\ref{expansopmSs}), it can be further written as 
\begin{align}
	& \mathbb{E} \left[  \log\left(1 + p_i^{\mathbf{H}\ast} L_{ii} |H_{ii}|^2  \right) | \mathbf{q}  \right]  	\notag \\
	=&  -\gamma_{eu} + \log \frac{ L_{ii} J_i'\left(q_i \right)\tau}{\gamma_i}  - \sum_{n=1}^{\infty} \frac{\left(-x \right)^n}{n ! n} \left(- \frac{\gamma_i}{ L_{ii} J_i'\left(q_i \right)\tau} \right)^n	\notag \\
	=&  -\gamma_{eu} + \log \frac{ L_{ii} J_i'\left(q_i \right)\tau}{\gamma_i} + o\left(1 \right)		\notag \\
	\overset{(a)}=&  \mathcal{O}\left( \log \left(q_i \right)   \right)-\gamma_{eu}+  \log \frac{ L_{ii} \tau}{\gamma_i} + o\left(1 \right) = \mathcal{O}\left( \log \left(q_i \right)   \right),\quad  \text{as } q_i \rightarrow \infty	     \label{pdeapped1}
\end{align}
where $(a)$ is due to $ \overline{B}_2 \overline{C}_2 \frac{q_i}{\log \left( q_i \right)} \leq J_i'\left(q_i \right) \leq \overline{B}_1\overline{C}_1 \frac{q_i}{\log \left( q_i \right)} \left(\text{according to } (\ref{kprimeorder}) \right) \Rightarrow \log \left ( \overline{B}_2 \overline{C}_2\right) + \log \left( q_i\right)  - \log \log \left(q_i \right) \leq J_i'\left(q_i \right) \leq \log \left( \overline{B}_1\overline{C}_1\right) + \log \left(q_i \right)  + \log \log \left( q_i \right) \Rightarrow  \log \left( q_i\right)  + o\left( \log \left( q_i\right)  \right) \leq J_i'\left(q_i \right) \leq  \log \left( q_i\right)  + o\left( \log \left( q_i\right)  \right) \Rightarrow J_i'\left(q_i \right) = \mathcal{O} \left( \log \left(q_i \right)   \right)$. Second, we calculate $\mathbb{E} \left[  \frac{J_k'\left(q_k\right)p_k^{\mathbf{H}\ast}L_{kk}|H_{kk}|^2}{1+p_k^{\mathbf{H}\ast} L_{kk}|H_{kk}|^2} p_j^{\mathbf{H}\ast} |H_{kj}|^2 \Big| \mathbf{q}  \right]$. Note that $\mathbb{E} \left[  \frac{J_k'\left(q_k\right)p_k^{\mathbf{H}\ast}L_{kk}|H_{kk}|^2}{1+p_k^{\mathbf{H}\ast} L_{kk}|H_{kk}|^2}  p_j^{\mathbf{H}\ast} |H_{kj}|^2 \Big| \mathbf{q}  \right] = \mathbb{E} \left[ \frac{J_k'\left(q_k\right)p_k^{\mathbf{H}\ast}L_{kk}|H_{kk}|^2}{1+p_k^{\mathbf{H}\ast} L_{kk}|H_{kk}|^2}  \Big| \mathbf{q}  \right] \cdot  \mathbb{E} \left[ p_j^{\mathbf{H}\ast}  \Big| \mathbf{q}  \right]  \cdot  \mathbb{E} \left[ |H_{kj}|^2    \right]  $ and each term is calculated as follows:
\begin{align}
	 &\mathbb{E} \left[ \frac{J_k'\left(q_k\right)p_k^{\mathbf{H}\ast}L_{kk}|H_{kk}|^2}{1+p_k^{\mathbf{H}\ast} L_{kk}|H_{kk}|^2}  \Big| \mathbf{q}  \right] =  \int_{\frac{\gamma_k}{ L_{kk} J_k'\left(q_k \right)\tau}}^{\infty}  \left( J_k'\left(q_k \right) - \frac{\gamma_k }{\tau L_{kk} x} \right)e^{-x}  \mathrm{d}x 	\notag	 \\
	 &\hspace{4.5cm}= J_k'\left(q_k \right)  e^{- \frac{\gamma_k}{ L_{kk} J_k'\left(q_k \right)\tau}} - \frac{\gamma_k}{L_{kk}\tau}  E_1\left(\frac{\gamma_k}{ L_{kk} J_k'\left(q_k \right)\tau}\right)		\label{equapped11}	 \\
	&\mathbb{E} \left[ p_j^{\mathbf{H}\ast}  \Big| \mathbf{q}  \right] \overset{(b)}= \frac{J_j'\left(q_j \right)\tau}{\gamma_j} e^{- \frac{\gamma_j}{ L_{jj} J_j'\left(q_j \right)\tau}} - \frac{1}{L_{jj}}  E_1\left(\frac{\gamma_j}{ L_{jj} J_j'\left(q_j \right)\tau}\right)		 \label{equapped2} \\
	&\mathbb{E} \left[ |H_{kj}|^2    \right] = \int_{0}^{\infty} x e^{-x} \mathrm{d} x =1
\end{align}
where $(b)$  is calculated according to (\ref{expec1}). Using the  series expansions in (\ref{expansopmSs}), the equations in (\ref{equapped11})  and (\ref{equapped2})  can be further written as 
\begin{align}
	&\mathbb{E} \left[ \frac{J_k'\left(q_k\right)p_k^{\mathbf{H}\ast}L_{kk}|H_{kk}|^2}{1+p_k^{\mathbf{H}\ast} L_{kk}|H_{kk}|^2}  \Big| \mathbf{q}  \right]    \notag \\
	=& J_k'\left(q_k \right)  \sum_{n=0}^{\infty} \frac{1}{n!}\left(- \frac{\gamma_k}{ L_{kk} J_k'\left(q_k \right)\tau} \right)^n - \frac{\gamma_k}{L_{kk}\tau} \left(  -\gamma_{eu} + \log \frac{ L_{kk} J_k'\left(q_k \right)\tau}{\gamma_k}  - \sum_{n=1}^{\infty} \frac{\left(-x \right)^n}{n ! n} \left(- \frac{\gamma_k}{ L_{kk} J_k'\left(q_k \right)\tau} \right)^n \right)  	\notag \\
	=&  J_k'\left(q_k \right) - \frac{\gamma_k}{ L_{kk} \tau}  + \frac{\gamma_k}{L_{kk}\tau} \left(  \gamma_{eu} + \log \frac{\gamma_k}{ L_{kk} J_k'\left(q_k \right)\tau}   \right)   + o\left(1 \right),\quad  \text{as } q_k \rightarrow \infty			 \\
	&\mathbb{E} \left[ p_j^{\mathbf{H}\ast}  \Big| \mathbf{q}  \right] \notag \\
	=& \frac{J_j'\left(q_j \right)\tau}{\gamma_j} \sum_{n=0}^{\infty} \frac{1}{n!}\left(- \frac{\gamma_j}{ L_{jj} J_j'\left(q_j \right)\tau} \right)^n - \frac{1}{L_{jj}} \left(  -\gamma_{eu} + \log \frac{ L_{jj} J_j'\left(q_j \right)\tau}{\gamma_j}  - \sum_{n=1}^{\infty} \frac{\left(-x \right)^n}{n ! n} \left(- \frac{\gamma_j}{ L_{jj} J_j'\left(q_j \right)\tau} \right)^n \right)		\notag \\
	= & \frac{\tau}{\gamma_j} \left(  J_j'\left(q_j \right) - \frac{\gamma_j}{ L_{jj} \tau}  + \frac{\gamma_j}{L_{jj}\tau} \left(  \gamma_{eu} + \log \frac{\gamma_j}{ L_{jj} J_j'\left(q_j \right)\tau}   \right)   \right),\quad  \text{as } q_j \rightarrow \infty	
\end{align}
Therefore, we have
\begin{align}
	 & \mathbb{E} \left[  \frac{J_k'\left(q_k\right)p_k^{\mathbf{H}\ast}L_{kk}|H_{kk}|^2}{1+p_k^{\mathbf{H}\ast} L_{kk}|H_{kk}|^2} p_j^{\mathbf{H}\ast} |H_{kj}|^2 \Big| \mathbf{q}  \right] 	\notag \\
	 =  & \frac{\tau}{\gamma_j}J_k'\left(q_k \right)  J_j'\left(q_j \right) + o\left( J_k'\left(q_k \right)  J_j'\left(q_j \right) \right)		\notag	\\
	 = &  \frac{\beta_k \beta_j}{\gamma_j \lambda_k \lambda_j \tau} \mathcal{O} \left(\frac{q_k q_j}{\log \left( q_k \right) \log \left(q_j \right)} \right),\quad  \text{as } q_j, q_k \rightarrow \infty 	\label{pdeapped2}
\end{align}
Based on (\ref{pdeapped1}) and (\ref{pdeapped2}), the PDE in (\ref{PDEappd}) can be written as 
\begin{align}		\label{highlestT}
	\sum_{i=1}^K \left(\lambda_i - \tau \mathcal{O} \left( \log \left(q_i \right)   \right)\right) \frac{\mathrm{d} \widetilde{J}_{kj}\left(\mathbf{q}\right) }{\mathrm{d} q_i} = \frac{\beta_k \beta_j}{\gamma_j \lambda_k \lambda_j \tau} \mathcal{O} \left(\frac{q_k q_j}{\log \left( q_k \right) \log \left(q_j \right)} \right) ,\quad  \text{as } q_j, q_k \rightarrow \infty 
\end{align}
To balance the highest order terms on both sizes of (\ref{highlestT}),  $ \widetilde{J}_{kj}\left(\mathbf{q}\right)$  should be at the order of  $\frac{q_k q_j^2}{\log \left( q_k \right) \log \left(q_j \right)} $, $\frac{q_k^2 q_j}{\log \left( q_k \right) \log \left(q_j \right)}$, $\frac{q_k q_j}{\log \left( q_k \right) \log \left(q_j \right)} \text{li}\left(q_k \right)$ where $\text{li}\left( x \right) = \int_{0}^{x} \frac{1}{\log\left( t\right)} \mathrm{d}t$. Furthermore, since $\text{li}\left( x \right)  = \mathcal{O} \left(\frac{x}{\log\left(x \right)} \right) < \mathcal{O}\left( x\right)$,  we have the following asymptotic property of $\widetilde{J}_{kj}\left(\mathbf{q}\right)$:
\begin{align}		\label{Happened}
	\widetilde{J}_{kj}\left(\mathbf{q}\right) = - \frac{\beta_k \beta_j}{\gamma_j \lambda_k \lambda_j \tau^2} \mathcal{O} \left(\frac{q_k q_j^2+ q_j q_k^2}{\log \left( q_k \right) \log \left(q_j \right)} \right),\quad  \text{as } q_j, q_k \rightarrow \infty
\end{align}
Substituting (\ref{Happened}) into (\ref{Jijapprox}), we have
\begin{align}	
	\left| J\left(\mathbf{q}; \mathbf{L}\right) - \sum_{k=1}^K J_k\left(q_k\right) \right| =\sum_{k=1}^K \sum_{j \neq k} L_{kj}   \frac{\beta_k \beta_j}{\gamma_j \lambda_k \lambda_j \tau^2} \mathcal{O} \left(\frac{q_k q_j^2+ q_j q_k^2}{\log \left( q_k \right) \log \left(q_j \right)} \right)  + \mathcal{O}(L^2), \quad \text{as } q_j, q_k \rightarrow \infty, L \rightarrow 0
\end{align}
This proves the lemma.


\begin{thebibliography}{1}

\bibitem{CompApp}
S. P. Meyn, \emph{Control Techniques for Complex Networks}. \  Cambridge University Press, 2007.

\bibitem{WirelessApp}
Y. Cui, V. K. N. Lau, R. Wang, H. Huang, and S. Zhang, ``A survey on delay-aware resource control for wireless systems - large deviation theory, stochastic Lyapunov drift and distributed stochastic Learning," \emph{IEEE Transactions on Information Theory}, vol. 58, no. 3, pp. 1677--1701, Mar. 2012. 

\bibitem{EcomApp}
C. S. Tapiero,  \emph{Applied Stochastic Models and Control for Finance and Insurance}.  \  Boston: Kluwer Academic Publishers, 1998.

\bibitem{DP_Bertsekas}
D. P. Bertsekas, \emph{Dynamic Programming and Optimal Control}, 3rd ed. \ Massachusetts:  Athena Scientific, 2007.

\bibitem{Cao}
X. Cao, \emph{Stochastic Learning and Optimization: A Sensitivity-Based Approach}. \ Springer, 2008.

\bibitem{NeuNet}
D. P. Bertsekas and J. N. Tsitsiklis, \emph{Neuro-Dynamic Programming}.  \ Massachusetts: Athena Scientific,1996.

\bibitem{ADP1}
W. B. Powell, \emph{Approximate Dynamic Programming}. \ Hoboken, NJ: Wiley-Interscience, 2008.

\bibitem{ADP2}
S. Ji, A. Barto, W. Powell, and D. Wunsch, \emph{Handbook of Learning and Approximate Dynamic Programming}. \ Englewood Cliffs, NJ: Wiley, 2004.

\bibitem{StateAggr_1}
B. Van Roy, ``Performance loss bounds for approximate value iteration with state aggregation,"  \emph{Mathematics of Operations Research}, vol. 31, no.2, pp. 234--244, May 2006.

\bibitem{StateAggr_2}
T. Dean, R. Givan, and S. Leach, ``Model reduction techniques for computing approximately optimal solutions for markov decision processes," \emph{in Proceedings of the 13th Conference on Uncertainty in ArtiÞcial Intelligence}, pp. 124--131, Providence, Rhode Island, Aug. 1997.

\bibitem{BasisFunc1}
G. D. Konidaris and S. Osentoski, ``Value function approximation in reinforcement learning using the Fourier basis," \emph{in Proceedings of the 25th Conference on Artificial Intelligence}, pp. 380--385, Aug. 2011. 

\bibitem{BasisFunc2}
S. Mahadevan and M. Maggioni, ``Value function approximation with diffusion wavelets and Laplacian eigenfunctions," \emph{in Advances in Neural Information Processing Systems}.  Cambridge, MA: MIT Press, 2006.

\bibitem{ParaBasis1}
I. Menache, S. Mannor, and N. Shimkin, ``Basis function adaptation in temporal difference reinforcement learning," \emph{Annals of Operations Research}, vol. 134, no. 1, pp. 215--238, Feb. 2005.

\bibitem{ParaBasis2}
H. Yu and D. Bertsekas, ``Basis function adaptation methods for cost approximation in MDP,"  \emph{in Proceedings of  IEEE Symposium on Approximate Dynamic Programming and Reinforcement Learning}, pp. 74--81, 2009.

\bibitem{ADP_Bertsekas}
D. Bertsekas, ``Approximate policy iteration: A survey and some new methods,"  \emph{Journal of  Control Theory and Applications}, vol. 9, no. 3, pp. 310--335, 2010.

\bibitem{fluidiff}
W. Chen, D. Huang, A.A. Kulkarni, J. Unnikrishnan, Q. Zhu, P. Mehta, S. Meyn, and A. Wierman, ``Approximate dynamic programming using 
fluid and diffusion approximations with applications to power management," \emph{in Proceedings of the 48th IEEE Conference on  Decision and Control}, pp. 3575--3580,  Dec. 2009.

\bibitem{fluid2}
C. C. Moallemi, S. Kumar, and B. Van Roy, ``Approximate and data-driven dynamic programming for queueing networks,"  working paper, Stanford University, 2008.

\bibitem{fluidRel}
S. P. Meyn, ``The policy iteration algorithm for average reward Markov decision processes with general state space,"  \emph{IEEE Transactions on  Automation Control},  vol. 42, no. 12, pp. 1663--1680, 1997.

\bibitem{fluid3}
M. H. Veatch, ``Approximate dynamic programming for networks: Fluid models and constraint reduction,"  working paper, Department of Math, Gordon College, Wenham, MA, 2005.

\bibitem{closeformJ}
S. P.  Meyn, W. Chen, and D.  O'Neill, ``Optimal cross-layer wireless control policies using td learning,"  \emph{in Proceedings of the 49th IEEE Conference on Decision and Control}, pp. 1951--1956, 2010.

\bibitem{FLLN}
A. Mandelbaum, W. A. Massey, and M. Reiman, ``Strong approximations for Markovian service networks," \emph{Queueing Systems: Theory and Applications}, 30(1-2), pp. 149--201, 1998.


\bibitem{NUM1}
D. P. Palomar and M. Chiang, ``A tutorial on decomposition methods for network utility maximization,"  \emph{IEEE Journal on Selected Areas in Communications}, vol. 24, no. 8, pp. 1439--1451, Aug. 2006.

\bibitem{NUM2}
D. P. Palomar and M. Chiang, ``Alternative distributed algorithms for network utility maximization: Framework and applications," \emph{ IEEE Transactions on Automatic Control}, vol. 52, no. 12, pp. 2254--2269, Dec. 2007.

\bibitem{littlelaw}
L. Kleinrock, \emph{Queueing Systems. Volume 1: Theory.} \ London: Wiley-Interscience, 1975.


\bibitem{HuangJW}
J. Huang, ``Distributed algorithm design for network optimization problems with coupled objectives," \emph{in Proceedings of IEEE TENCON 2009 - 2009 IEEE Region 10 Conference}, pp. 1--6, Jan. 2009.

\bibitem{Game}
G. Scutari, D. P. Palomar, and S. Barbarossa, ``Competitive design of multiuser MIMO systems based on game theory: A unified view," \emph{IEEE Journal on Selected Areas in Communications}, vol. 26, no. 7, pp. 1089--1103, Sep. 2008.

\bibitem{DBPalg}
M. Andrews, K. Kumaran, K. Ramanan, A. Stolyar, R. Vijayakumar, and P. Whiting, ``Scheduling in a queueing system with asynchronously varying service rates,''  \emph{Probability in the Engineering and Informational Sciences}, vol. 18, no. 2, pp. 191-217, 2004.

\bibitem{simtopo}
G. Arslan, M. F. Demirkol, and Y. Song, ``Equilibrium efficiency improvement in MIMO interference systems: A decentralized stream control approach," \emph{IEEE Transactions on Wireless Communications}, vol. 6, pp. 2984--2993, Aug. 2007.

\bibitem{PolicyMap}
M. H. Veatch, ``Using fluid solutions in dynamic scheduling,"  \emph{Analysis and Modeling of Manufacturing Systems}, vol. 60, pp. 399--426, 2003.

\bibitem{IntroPDE}
Y. Pinchover and J. Rubinstein,  \emph{An Introduction to Partial Differential Equations}. \  Cambridge University Press, UK, 2005.

\bibitem{HandODE}
A. D. Polyanin, V. F. Zaitsev, and A. Moussiaux, \emph{Handbook of Exact Solutions for Ordinary Differential Equations}, 2nd ed.  \ Chapman  \& Hall/CRC Press, Boca Raton, 2003.

\bibitem{lambert}
A. Hoorfar and M. Hassani, ``Inequalities on the Lambert W function and hyperpower function,"  \emph{Journal of Inequalities in Pure and Applied
Mathematics (JIPAM)}, 9(2), 2008.





\end{thebibliography}
\end{document}